\documentclass[12pt]{article}
\usepackage{a41}
\usepackage{epsfig}
\usepackage{cite}
\usepackage{amsmath}
\usepackage{amssymb}
\usepackage{amsthm} %for proof environment, has to be loaded after amsmath
\usepackage{float}
\usepackage{verbatim}
\usepackage{amsxtra}
\usepackage{afterpage,color}

\newtheorem{example}{Example} 
 
\newcommand{\SigmaP}{\textsf{Sigma}}

\newcommand{\N}{\nonumber} 
\newcommand{\ep}{\varepsilon}

\newcommand{\bea}{\begin{eqnarray}}
\newcommand{\bq}{\begin{equation}}
\newcommand{\eea}{\end{eqnarray}}
\newcommand{\eq}{\end{equation}}

\newcommand\SH{\,\mbox{$\sqcup \! \sqcup$}\,}

\newcommand\be{\begin{eqnarray}}
\newcommand\ee{\end{eqnarray}}

\newcommand\sign{{\rm sign}}
\newcommand\Li{{\rm Li}}
\newcommand\HA{{\rm H}}

\newcommand\Mvec{\,\mbox{\bf M}}

\renewcommand{\HA}{\ensuremath{\mathrm{H}^*}}
\DeclareMathOperator{\arccosh}{arccosh}
\DeclareMathOperator{\arctanh}{arctanh}
\DeclareMathOperator{\arccoth}{arccoth}
\DeclareMathOperator{\arcsech}{arcsech}
\DeclareMathOperator{\arcsinh}{arcsinh}

\DeclareMathOperator{\arccot}{arccot}
\newtheorem{theorem}{Theorem}
\newtheorem{lemma}[theorem]{Lemma}

\newtheorem{prop}[theorem]{Proposition}
\newtheorem{cor}[theorem]{Corollary}

\numberwithin{equation}{section}

\allowdisplaybreaks[4]

\newcommand{\GL}[2]{\textnormal{GL}\left(#1;#2\right)}
\newcommand{\ie}{i.e.,\ }

\newcounter{mmacnt}
\def\restartmma{\setcounter{mmacnt}{0}}
\restartmma \catcode`|=\active
\def|#1|{\mathrm{#1}}
\catcode`|=12
\newenvironment{mma}{
 \par\smallskip
 \catcode`|=\active
 \parskip=0pt\parindent=0pt % locally
 \small
 \def\In##1\\{%
   \def\linebreak{\hfill\break\null\qquad}%
   \refstepcounter{mmacnt}
   \hangindent=2.5em\hangafter=0
   \leavevmode
   \llap{\tiny\sffamily In[\arabic{mmacnt}]:=\kern.5em}%
   \mathversion{bold}\footnotesize$\displaystyle##1$\normalsize
   \mathversion{normal}\par
 }%
 \def\Print##1\\{%
   \def\linebreak{\hfill\break}%
   \hangindent=2.5em\hangafter=0
   \leavevmode ##1\par}%
 \def\Out##1\\{%
   \def\linebreak{$\hfill\break\null\hfill$}%
   \kern\abovedisplayskip\par
   \hangindent=2.5em\hangafter=0
   \leavevmode
   \llap{\tiny\sffamily Out[\arabic{mmacnt}]=\kern.5em}
   \footnotesize$\displaystyle##1$\normalsize\hfill\null\par
   \kern\belowdisplayskip
 }%
 \def\Warning##1##2\\{%
   \def\linebreak{\hfill\break}%
   \hangindent=2.5em\hangafter=0
   \leavevmode
   {\scriptsize##1 : ##2}\par}%
}{%
 \par\smallskip
}

\renewcommand{\SigmaP}{\texttt{Sigma}}

\begin{document}
\noindent
\sloppy
\thispagestyle{empty}
\begin{flushleft}
DESY 14-021  \hfill  %{\tt arXiv:1407.xxx [math-ph]}
\\
DO-TH-13/22
\\
SFB/CPP-14-35
\\
LPN 14-082
\\
June 2014
\end{flushleft}
%
%\setcounter{page}{0}
% 1
%\mbox{}
\vspace*{\fill}
\hspace{-3mm}
{\begin{center}
{\LARGE \bf Iterated Binomial Sums and their}

\vspace*{3mm}
{\LARGE \bf Associated Iterated Integrals}
\end{center}
}

\begin{center}
\vspace{2cm}
\large
J.~Ablinger$^a$, J.~Bl\"umlein$^b$, C.G.~Raab$^b$, and C.~Schneider$^a$

\vspace{5mm}
\normalsize {\itshape $^a$~Research
Institute for Symbolic Computation (RISC),\\ Johannes Kepler
University, Altenbergerstra\ss{}e 69, A-4040 Linz, Austria}
\\ 

\vspace{5mm}
\normalsize
{\itshape $^b$~Deutsches Elektronen--Synchrotron, DESY,\\
Platanenallee 6, D--15738 Zeuthen, Germany}
\\

%\today
\end{center}

\vspace*{\fill} %
%%%%%%%%%%%%%%%%%%%%%%%%%%%%%%%%%%%%%%%%%%%%%%%%%%%%%%%%%%%%%%%%%%%%%%%%
\begin{abstract}
\noindent
We consider finite iterated generalized harmonic sums weighted by the binomial $\binom{2k}{k}$ 
in numerators and denominators. A large class of these functions emerges in the calculation 
of massive Feynman diagrams with local operator insertions starting at 3-loop order in the 
coupling constant and extends the classes of the nested harmonic, generalized harmonic and 
cyclotomic sums. The binomially weighted sums are associated by the Mellin transform to iterated 
integrals over square-root valued alphabets. The values of the sums for $N \rightarrow \infty$ 
and the iterated integrals at $x=1$ lead to new constants, extending the set of special numbers 
given by the multiple zeta values, the cyclotomic zeta values and special constants which emerge 
in the limit $N \rightarrow \infty$ of generalized harmonic sums. We develop algorithms to obtain
the Mellin representations of these sums in a systematic way. They are of importance for the
derivation of the asymptotic expansion of these sums and their analytic continuation to $N \in 
\mathbb{C}$. The associated convolution relations are derived for real parameters and can 
therefore be used in a wider context, as e.g. for multi-scale processes. We also derive algorithms 
to transform iterated integrals over root-valued alphabets into binomial sums. Using generating 
functions we study a few aspects of infinite (inverse) binomial sums. 
\end{abstract}
%%%%%%%%%%%%%%%%%%%%%%%%%%%%%%%%%%%%%%%%%%%%%%%%%%%%%%%%%%%%%%%%%%%%%%%%

\vspace*{\fill} 

\newpage
%%%%%%%%%%%%%%%%%%%%%%%%%%%%%%%%%%%%%%%%%%%%%%%%%%%%%%%%%%%%%%%%%%%%%%%
\section{Introduction}
\label{sec:1}
%%%%%%%%%%%%%%%%%%%%%%%%%%%%%%%%%%%%%%%%%%%%%%%%%%%%%%%%%%%%%%%%%%%%%%%

\vspace{1mm}
\noindent
In loop calculations of the different observables in renormalizable quantum field theories \cite{QED,Glashow:1961tr,
Weinberg:1967tq,'tHooft:1971fh,'tHooft:1971rn,'tHooft:1972fi,'tHooft:1978xw,Gross:1973id,Politzer:1973fx,Fritzsch:1973pi,
'tHooft:1973pz,DIAG:BOOK} a sequence of number- and function spaces arises, growing with the complexity of the process, 
which is measured by its loop order, the number of Lorentz invariants involved, and the number of legs, cf. \cite{Ablinger:2013jta,
Ablinger:2013eba}. In the simplest cases the observables can be expressed by special constants, the multiple zeta values 
\cite{Borwein:1999js,Blumlein:2009cf}, 
cf. e.g.~\cite{Laporta:1996mq,vanRitbergen:1997va,Czakon:2004bu,Baikov:2012er}. 
In case of single differential distributions harmonic sums \cite{Vermaseren:1998uu,Blumlein:1998if} or harmonic polylogarithms 
\cite{Remiddi:1999ew} arise. They are generalized by the cyclotomic harmonic sums, polylogarithms and special numbers 
\cite{Ablinger:2011te} on the one hand and the generalized harmonic sums, polylogarithms, and special numbers 
\cite{Moch:2001zr,Ablinger:2013cf} on the other hand. Both classes can be united into the generalized cyclotomic harmonic sums, their polylogarithms 
and special numbers \cite{Ablinger:2013cf}. These function spaces form quasi shuffle or shuffle algebras \cite{Reutenauer1993,
Hoffman,Blumlein:2003gb}, for which basis representations can be derived. Furthermore, they obey structural relations 
\cite{Blumlein:2009ta,Blumlein:2009fz,Ablinger:2011te,Ablinger:2013cf} leading to a further reduction of the bases. These 
function spaces apply to massless 3-loop calculations such as the massless Wilson coefficients in deep-inelastic scattering 
\cite{Vermaseren:2005qc} and a wider class of the single mass 3-loop graphs such as heavy flavor Wilson coefficients in 
deep-inelastic scattering at large virtualities \cite{Ablinger:2010ty,Blumlein:2012vq,Ablinger:2012qm,Ablinger:2014lka,
Ablinger:2014vwa,Ablinger14}. Up to 2-loop order only harmonic sums and polylogarithms appear in the representation 
\cite{Buza:1995ie,Buza:1996xr,Buza:1996wv,Buza:1997mg,Bierenbaum:2007qe,Bierenbaum:2007pn,Bierenbaum:2008yu,
Bierenbaum:2009zt,Blumlein:2014fqa,Behring:2014eya}.

In the massive case at 3-loop order new structures arise. These are the finite nested binomial sums.
They occur for some of the graphs with two fermion line of equal mass \cite{Ablinger:2014uka}, but also in case of 
a single massive fermion line carrying 4-leg local operators \cite{Ablinger:2014yaa}. Furthermore, they are known 
from massive 2-loop calculations of $2 \rightarrow 2$ scattering, as the Bhabha-process \cite{Aglietti:2004tq} and 
from Ref.~\cite{Fleischer:1998nb}. These functions extend the above classes. It is worthwhile to mention that this type of 
sums also arises in number theory and combinatorics as illustrated, e.g., in~\cite{Borwein:2001,Schneider:09a,SIG2}.

We consider finite sums of the form
%-----------------------------------------------------------------------------------------------------------
\begin{equation}
\sum_{i_1=1}^Na_1(i_1)\sum_{i_2=1}^{i_1}a_2(i_2)\dots\sum_{i_k=1}^{i_{k-1}}a_k(i_k),
\label{eq:SUM1}
\end{equation}
%-----------------------------------------------------------------------------------------------------------
where the summands are of the form
%-----------------------------------------------------------------------------------------------------------
\begin{equation}
a_j(N; b,c,m)=\binom{2N}{N}^b\frac{c^N}{N^m},
\end{equation}
%-----------------------------------------------------------------------------------------------------------
with 
%-----------------------------------------------------------------------------------------------------------
\begin{eqnarray}
b \in \{-1,0,1\},~~~c \in \mathbb{R}\setminus\{0\},~~~m \in \mathbb{N}. 
\end{eqnarray}
%-----------------------------------------------------------------------------------------------------------
Later, in concrete examples, we will consider values 
%-----------------------------------------------------------------------------------------------------------
\begin{eqnarray}
c \in \{\pm 2^k \mid k \in \mathbb{Z}\}.
\end{eqnarray}
%-----------------------------------------------------------------------------------------------------------
We also treat some examples with a slightly more general structure, e.g.
%-----------------------------------------------------------------------------------------------------------
\begin{equation}
 \frac{c^N}{\displaystyle(2N+1)\binom{2N}{N}}.
\label{eq:SUM:3}
\end{equation}
%-----------------------------------------------------------------------------------------------------------
Infinite binomial sums have been studied in
Refs.~\cite{Ogreid:1997bx,Fleischer:1998nb,Davydychev:2000na,Kalmykov:2000qe,Borwein:2001,Davydychev:2003mv,Weinzierl:2004bn,
Kalmykov:2007dk}.

The weighted binomial sums considered in the present paper can be expressed by a Mellin transform 
\cite{Mellin:1886,Mellin:1902} of iterated integrals \cite{Kummer40a,Poincare1884} containing root-valued letters. The analytic 
Mellin-inversion of these sums is important to know since for a wide range of physical applications
the corresponding observables are measured in $x$-space. Moreover, one may also perform the
Mellin inversion by a numerical contour integral around the singularities of the given problem 
in the complex plane, which requires the analytic continuation of the sums considered from the even or
the odd integer values $N$ to the complex plane. This is usually obtained by considering
the asymptotic expansion of theses sums for $|N| \rightarrow \infty$ outside the singularities of the problem 
together with the shift relations of the corresponding sums \cite{Blumlein:2000hw,Blumlein:2005jg,Blumlein:2009ta,
Blumlein:2009fz}. As will be shown below, the asymptotic expansion is derived easiest referring to the 
Mellin-representation of the sums. 

The paper is organized as follows. In Section~\ref{sec:2} we summarize the main properties of the Mellin transform,
as it is extensively used subsequently. The building blocks of the iterated integrals associated to the nested binomial 
sums Eq.~(\ref{eq:SUM1}) are summarized in Section~\ref{sec:3}. There we give an outline on the structure of the
alphabet of the iterated integrals and representations of the weight {\sf w=1} integrals. The Mellin representations can
be built using convolution integrals. In Section~\ref{sec:4} we prove a series of lemmata and theorems allowing to express 
the 
corresponding convolutions in terms of iterated integrals. In Section~\ref{sec:5} the Mellin representations for nested 
finite binomial sums are presented.
The Mellin transform of $D$-finite functions is discussed in Section~\ref{sec:6}. Here we present different algorithms 
to transform root-valued iterated integrals into the corresponding nested sums. In Section~\ref{sec:7} some aspects of
infinite (inverse) binomial sums are dealt with using the framework of generating functions and Section~\ref{sec:8} contains 
the conclusions. In an Appendix we summarize special constants related to the sums and iterated integrals studied. In a series
of cases only a suitable integral representation over special functions could be derived, which allow for a precise numerical
representation.  
%----------------------------------------------------------------------------
\section{The Mellin Transform}
\label{sec:2}
%----------------------------------------------------------------------------

\vspace*{1mm}
\noindent
We briefly recall the definition and basic properties of the Mellin transform as it will play a crucial role 
in the integral representations which will be derived for the nested (inverse)
binomial sums. Subsequently, we denote the Mellin transform by
%----------------------------------------------------------------------------
\begin{equation}\label{eq:MellinDef}
\Mvec[f(x)](N) = \int_0^1 dx x^N f(x)~.
\end{equation}
%----------------------------------------------------------------------------
Obviously it inherits the linearity from the integral and shifts in $N$ correspond to multiplication by 
powers of $x$, i.e.,
%----------------------------------------------------------------------------
\begin{equation}\label{eq:MellinShifts}
\Mvec[f(x)](N+k) = \Mvec[x^kf(x)](N)~.
\end{equation}
%----------------------------------------------------------------------------
As a consequence we have the following summation formula
%----------------------------------------------------------------------------
\begin{equation}\label{eq:MellinSummation}
\sum_{i=1}^Nc^i\Mvec[f(x)](i) = c^N\Mvec\left[\frac{x}{x-\frac{1}{c}}f(x)\right](N)
-\Mvec\left[\frac{x}{x-\frac{1}{c}}f(x)\right](0)~.
\end{equation}
%----------------------------------------------------------------------------
Furthermore, the following properties are immediate, where $a>0$:
%----------------------------------------------------------------------------
\begin{eqnarray}
  \Mvec\left[\ln(x)^m f(x)\right](N)&=&\frac{d^m}{dN^m}\Mvec[f(x)](N),\\
  \Mvec[f(a x)](N)&=& \frac{1}{a^{N+1}} \Mvec[f(x) \theta(a-x)](N),~~a \leq 1,\\
  \Mvec[f(x^a)](N)&=&\frac{1}{a}\Mvec[f(x)]\left(\frac{N+1-a}{a}\right),\\
  \Mvec\left[f^{(p)}(x)\right](N)&=&\frac{(-1)^p N!}{(N-p)!}\Mvec[f(x)](N-p)+\sum_{i=0}^{p-1}
\frac{(-1)^i N!}{(N-i)!}f^{(p-1-i)}(1)~,
\nonumber\\  && \text{with}~f(x) \in {\cal C}^{(p)}[0,1]~,
\end{eqnarray}
%----------------------------------------------------------------------------
where $\theta(x)$ denotes the Heaviside function
%----------------------------------------------------------------------------
\begin{equation}\label{eq:HEAVYS}
\theta(x) = \Biggl\{\begin{array}{c}{1~~~\text{for}~~~x \geq 0}\\
                                     0~~~\text{for}~~~x < 0 \end{array}~.
\end{equation}
%----------------------------------------------------------------------------
The Mellin-convolution of two real functions with support $[0,1]$ is defined by
%----------------------------------------------------------------------------
\begin{equation}\label{eq:MellinMult}
f(x) \ast g(x) = \int_0^1 dx_1 \int_0^1 dx_2 \delta(x-x_1 x_2) f(x_1) g(x_2)~.
\end{equation}
%----------------------------------------------------------------------------
The Mellin transform obeys the relation
%----------------------------------------------------------------------------
\begin{equation}\label{eq:MellinMult}
\Mvec[f(x) \ast g(x)](N) = \Mvec[f(x)](N) \cdot \Mvec[g(x)](N)~. 
\end{equation}
%----------------------------------------------------------------------------
We define the Mellin transformation for functions with $+$-prescription by 
%----------------------------------------------------------------------------
\begin{equation}\label{eq:MellinMult}
\Mvec[[f(x)]_+](N) = \int_0^1 dx (x^N-1) f(x)~.
\end{equation}
%----------------------------------------------------------------------------
Depending on the regularity of $f$ and $g$ we have the following formulae:
%----------------------------------------------------------------------------
\begin{eqnarray}
f(x) \ast g(x) &=& \int_x^1 \frac{dy}{y} f(y) g\left(\frac{x}{y}\right)\label{eq:MellinConv}
\\
%-------
\left[f(x)\right]_+ \ast g(x) 
&=& \int_x^1 dy f\left(\frac{x}{y}\right)\left[\frac{1}{y} g(y) - \frac{x}{y^2} g(x) \right] - g(x) \int_0^x dy f(y)
\nonumber\\
&=& \int_x^1 dy f(y)\left[\frac{1}{y} g\left(\frac{x}{y}\right)- g(x)\right] - g(x) \int_0^x dy f(y)~.
\end{eqnarray}
%----------------------------------------------------------------------------
In the presence of singularities or branch cuts in the interval from 0 to 1 we define the Mellin transform by
%----------------------------------------------------------------------------
\begin{equation}\label{eq:modifiedMellin}
 \Mvec[f(x)](N) := \lim_{\varepsilon\to0}\int_0^1 dx x^N \frac{f(x+i\varepsilon)+f(x-i\varepsilon)}{2}~.
\end{equation}
%----------------------------------------------------------------------------

We conclude this Section by one of the main applications of the results
worked out in this article.
In recent 3-loop calculations we have to derive the asymptotic expansion of a
special class of the binomial sums, cf.~\cite{Ablinger:2012qm,Ablinger:2014yaa}.
Here it is instrumental to obtain first the analytic continuation of the nested
sums from even or odd integer values $N$ to $N \in \mathbb{C}$. 
Suppose we are given such a representation (in form of integral
representations of the occurring sums) by the toolbox presented in this article.
Then we follow Refs.~\cite{Nielsen1906b,Landau1906} and obtain the
asymptotic expansion using
factorial series as follows. By the change of the variable $x=e^{-z}$ we have
%----------------------------------------------------------------------------
\begin{equation}
\Mvec[f(x)](N) = \int_0^\infty dz e^{-Nz}f(e^{-z})e^{-z}~.
\end{equation}
%----------------------------------------------------------------------------
Now, we expand $f(e^{-z})e^{-z}$ at $z=0$ and integrate term-wise to obtain the asymptotic expansion for $N \to \infty$ 
using
%----------------------------------------------------------------------------
\begin{equation}
\int_0^\infty dz e^{-Nz}z^c\ln(z)^k = \frac{\partial^k}{\partial{c}^k}\frac{\Gamma(c+1)}{N^{c+1}} 
= \sum_{i=0}^k(-1)^i\binom{k}{i}\Gamma^{(k-i)}(c+1)\frac{\ln(N)^i}{N^{c+1}}~,
\end{equation}
%----------------------------------------------------------------------------
for $c>-1$ and $k \in \mathbb{N}$. In particular, assume that $f(x)$ is regular
at $x=1$ and has the expansion
%----------------------------------------------------------------------------
\begin{equation}
f(x)=\sum_{i=0}^\infty c_i(1-x)^i
\end{equation}
%----------------------------------------------------------------------------
as $x$ approaches $1$ from the left. Then the asymptotic expansion of its Mellin transform only involves integer powers of 
$N$ and a formula for the coefficients in the expansion can be given in terms of Stirling 
numbers of second kind ${\sf S}_2$~\cite{Nielsen1906b,Landau1906,STIRLING}:
%----------------------------------------------------------------------------
\begin{equation}
\Mvec[f(x)](N) = \sum_{k=0}^\infty\frac{\sum\limits_{i=0}^k(-1)^{k+i}i!{\sf S}_2(k+1,i+1)c_i}{N^{k+1}}~.
\end{equation}
%----------------------------------------------------------------------------
In addition, using also the shift relations of the nested sums for $N
\rightarrow N+1$ one may cover the whole
analyticity range of the sums for $N \in \mathbb{C}$. For the harmonic, cyclotomic, and generalized harmonic(cyclotomic)
sums this has been shown in Refs.~\cite{Blumlein:2009ta,Ablinger:2011te,Ablinger:2013cf}. 

%------------------------------------------------------------------------------------------------------------------
\section{Iterated Integrals over Root-valued Alphabets}
\label{sec:3}
%------------------------------------------------------------------------------------------------------------------

\vspace*{1mm}
\noindent
The general aim is to represent our nested sums (\ref{eq:SUM1}--\ref{eq:SUM:3})
in terms of Mellin transforms in the form
%------------------------------------------------------------------------------------------------------------------
\begin{equation}\label{eq:IntRepForm}
c_0+\sum_{j=1}^k c_j^N\Mvec[f_j(x)](N),
\end{equation}
%------------------------------------------------------------------------------------------------------------------
where the constants $c_j$ and functions $f_j(x)$ do not depend on $N$. This is achieved by virtue of the properties 
of the Mellin transform listed above. Due to the use of the summation property \eqref{eq:MellinSummation} the constant 
$c_0$ needs not be made explicit and it will be convenient to write \eqref{eq:IntRepForm} as
%------------------------------------------------------------------------------------------------------------------
\begin{equation}\label{eq:IntRepForm2}
\sum_{j=1}^k \int_0^1dx \frac{(c_jx)^N-1}{x-\frac{1}{c_j}}g_j(x).
\end{equation}
%------------------------------------------------------------------------------------------------------------------
As starting point we only need the following basic integral representations:
%----------------------------------------------------------------------------
\begin{eqnarray}
\frac{1}{N} &=& \Mvec\left[\frac{1}{x}\right](N) \label{eq:RepresentInverse}
\\
%-------
\binom{2N}{N} &=& \frac{4^N}{\pi}\Mvec\left[\frac{1}{\sqrt{x(1-x)}}\right](N) \label{eq:RepresentBinomial}
\\
%-------
\frac{1}{N\displaystyle\binom{2N}{N}} &=&
\frac{1}{4^N}\Mvec\left[\frac{1}{x\sqrt{1-x}}\right](N).
\label{eq:RepresentInverseBinomial}
\end{eqnarray}
%----------------------------------------------------------------------------
From these we can obtain integral representations for sums and nested sums step by step. In general 
the computation proceeds as follows. Starting from the innermost sum we move outwards maintaining an 
integral representation of the sub-expressions visited so far. For each intermediate sum
%------------------------------------------------------------------------------------------------------------------
\begin{equation}\label{eq:AlgorithmIntermediate}
\sum_{i_j=1}^Na_j(i_j)\sum_{i_{j+1}=1}^{i_j}a_{j+1}(i_{j+1})\dots\sum_{i_k=1}^{i_{k-1}}a_k(i_k)
\end{equation}
%------------------------------------------------------------------------------------------------------------------
this first involves setting up an integral representation for the building block $a_j(N)$ of the form 
\eqref{eq:IntRepForm}. 
This may require the computation of Mellin convolutions, which we will describe in more detail below. Next we obtain an 
integral 
representation of the same form of
%------------------------------------------------------------------------------------------------------------------
\begin{equation}
a_j(N)\sum_{i_{j+1}=1}^Na_{j+1}(i_{j+1})\dots\sum_{i_k=1}^{i_{k-1}}a_k(i_k)
\end{equation}
%------------------------------------------------------------------------------------------------------------------
by Mellin convolution with the result for the inner sums computed so far. Then by the summation property 
\eqref{eq:MellinSummation} we obtain an integral representation for the sum \eqref{eq:AlgorithmIntermediate}. 
These steps are repeated until the outermost sum has been processed. For a detailed example of this construction process we 
refer to Example~1 below.

Let $x \in [0,1]$ and the functions $a_i(x) \in \mathbb{R}$
be integrable on $]0,1]$. They are called letters forming the words $\{a_{i_1}, \ldots, a_{i_1}\}$. 
In analogy to harmonic polylogarithms we define the Poincar\'e-type iterated integral
%----------------------------------------------------------------------------
\begin{eqnarray}
\HA_{\emptyset}(x)     &=&  1 \\
\HA_{\sf b,\vec{c}}(x) &=&  \int_x^1 dt b(t) \HA_{\sf \vec{c}}(t)~.
\end{eqnarray}
%----------------------------------------------------------------------------
Note that the integration is over the interval $[x,1]$ in contrast to the harmonic polylogarithms $\mathrm{H}_{\sf 
\vec{a}}(x)$
and related iterated integrals Refs.~\cite{Remiddi:1999ew,Ablinger:2011te,Ablinger:2013cf}, where the
integration is over $[0,x]$. We use the star to make the notation unambiguous. Following this convention we obtain the identity
%----------------------------------------------------------------------------
\begin{equation}
\Mvec[\HA_{\sf a,\vec{b}}(x)](N) = \frac{1}{N+1}\Mvec[xa(x)\HA_{\sf \vec{b}}(x)](N).
\end{equation}
%----------------------------------------------------------------------------
By partial integration one shows validity of the relations:
%----------------------------------------------------------------------------
\begin{eqnarray}
\HA_{\sf a}(x) \HA_{\sf \vec{b}}(x) = \HA_{\sf a, \vec{b}}(x)
                                    + \HA_{\sf b_1,a,b_2,...,b_n}(x) + \ldots
                                    + \HA_{\sf \vec{b}, a}(x)~.
\end{eqnarray}
%----------------------------------------------------------------------------
\begin{proof}
For two letters one obtains
%----------------------------------------------------------------------------
\begin{eqnarray}
\HA_{\sf a,b}(x) &=& \int_x^1 dy a(y) \int_y^1 dz b(z) \N\\ 
                     &=& \HA_{\sf a}(x) \HA_{\sf b}(x) - \int_x^1 dy \HA_{\sf a}(y) b(y) 
\N\\
                     &=& \HA_{\sf a}(x) \HA_{\sf b}(x) - \HA_{\sf b,a}(x)
\end{eqnarray}
%----------------------------------------------------------------------------
and
%----------------------------------------------------------------------------
\begin{eqnarray}
\HA_{\sf a_1}(x) \HA_{\sf a_2}(x) &=& \HA_{\sf a_1,a_2}(x) + \HA_{\sf a_2,a_1}(x)~.
\end{eqnarray}
%----------------------------------------------------------------------------
For induction consider
%----------------------------------------------------------------------------
\begin{eqnarray}
\HA_{\sf a,b,\vec{c}}(x) &=& \int_x^1 dy a(y) \HA_{\sf b,\vec{c}}(y) \N\\ 
                     &=& \HA_{\sf a}(x) \HA_{\sf b,\vec{c}}(x) - \int_x^1 dy 
b(y) \HA_{\sf a}(y) \HA_{\sf \vec{c}}(y).
\end{eqnarray}
\end{proof}
%----------------------------------------------------------------------------
It is known, see \cite{Ree}, that iterated integrals satisfy the shuffle relations regardless of the specific 
form of the integrands. For the iterated integrals $\HA_{\sf \vec{a}}(x)$ they are defined by
%----------------------------------------------------------------------------
\begin{eqnarray}
\HA_{\sf \vec{a}}(x) \cdot \HA_{\sf \vec{b}}(x) = \HA_{\sf \vec{a}}(x) \SH \HA_{\sf \vec{b}}(x) = 
\sum_k \HA_{\sf \vec{c}_k}(x),
\end{eqnarray}
%----------------------------------------------------------------------------
where the sum runs over all shuffles of the words ${\sf \vec{a}}$ and ${\sf \vec{b}}$, i.e.
sequences $\sf \vec{c}_k$ out of the letters of $\sf \vec{a}$ and 
$\sf \vec{b}$ which preserve the order of the latter sets. 
Carefully choosing which specific integrands are included in the alphabet, one can 
ensure that there are no additional algebraic relations among the iterated integrals over this alphabet apart 
from the shuffle relations. We define an alphabet containing letters with root-singularities which has exactly 
this property: all algebraic relations among the iterated integrals over this alphabet are induced by the shuffle 
relations. This relies on a theorem proven in Ref.~\cite{DeneufchatelEtAl}, which gives a criterion on the 
linear independence of iterated integrals over a given alphabet. Since any polynomial expression in terms 
of iterated integrals can be reduced to a linear combination of iterated integrals over the same alphabet 
by shuffling, establishing linear independence implies that all algebraic relations among the iterated integrals 
are due to shuffling. Using results from Refs.~\cite{Risch,Trager79} the iterated integrals over the alphabet 
we define below can be proven to be linearly independent over the algebraic functions.
%----------------------------------------------------------------------------
\begin{eqnarray}
 f_a(x)&:=&\frac{\sign(1-a-0)}{x-a},\\
 f_{\{a_1,\dots,a_k\}}(x)&:=&f_{a_1}(x)^{1/2}\dots f_{a_k}(x)^{1/2}\quad\quad k\ge2,\\
 f_{(a_0,\{a_1,\dots,a_k\})}(x)&:=&f_{a_0}(x)f_{a_1}(x)^{1/2}\dots f_{a_k}(x)^{1/2}\quad\quad k\ge1,\\
 f_{(\{a_1,\dots,a_k\},j)}(x)&:=&x^jf_{(a_1,\dots,a_k)}(x) \quad\quad j\in\{1,\dots,k-2\}.
\end{eqnarray}
%----------------------------------------------------------------------------
Restricting to at most two root-singularities we are left with the following cases:
\begin{eqnarray}
 f_a(x)&:=&\frac{\sign(1-a-0)}{x-a},\\
 f_{(a,\{b\})}(x)&:=&f_{a}(x)\sqrt{f_{b}(x)},\\
 f_{\{a,b\}}(x)&:=&\sqrt{f_{a}(x)}\sqrt{f_{b}(x)},\\
 f_{(a,\{b,c\})}(x)&:=&f_{a}(x)\sqrt{f_{b}(x)}\sqrt{f_{c}(x)}.
\end{eqnarray}
%----------------------------------------------------------------------------
Already in 2004 the following six letters with root-singularities were considered in the context 
of 2-loop integrals with massive propagators \cite{Aglietti:2004tq}, see also 
\cite{Bonciani:2010ms}:
%----------------------------------------------------------------------------
\begin{equation}
 \frac{1}{\sqrt{x(4 \pm x)}} \quad\text{and}\quad \frac{1}{(1 \pm x)\sqrt{x(4 \pm x)}}.
\end{equation}
%----------------------------------------------------------------------------
For a more compact notation we define the following letters, where we use the standard definition of 
$\sqrt{\ }$ with the branch cut on the negative real axis. Note that expressions like $\frac{1}{\sqrt{x(1+x)}}$ 
have extra branch cuts compared to $\frac{1}{\sqrt{x}\sqrt{1+x}}$, so we use the latter. However, the expressions 
$\frac{1}{\sqrt{x}\sqrt{1+x}}$ and $\sqrt{\frac{1}{x}}\sqrt{\frac{1}{1+x}}$ have the same domain of analyticity 
and agree on that domain, they differ only in their values on the cuts.

In the following we consider an alphabet given by the 37 letters
%----------------------------------------------------------------------------
\begin{eqnarray}
f_{\sf 0}(x) &=& \frac{1}{x} 
\\
%-------
f_{\sf 1}(x) &=& \frac{1}{1-x} 
\\
%-------
f_{\sf -1}(x) &=& \frac{1}{1+x} 
\\
%-------
f_{\sf 2}(x) &=& \frac{1}{2-x} 
\\
%-------
f_{\sf -2}(x) &=& \frac{1}{2+x} 
\\
%-------
f_{\sf w_1}(x) &=& \frac{1}{\sqrt{x}\sqrt{1-x}} \equiv f_{\sf w_1}(1-x)
\\
%-------
f_{\sf w_2}(x) &=& \frac{1}{\sqrt{x}\sqrt{1+x}} 
\\
%-------
f_{\sf w_3}(x) &=& \frac{1}{x\sqrt{1-x}} 
\\
%-------
f_{\sf w_4}(x) &=& \frac{1}{x\sqrt{1+x}} \equiv -f_{\sf w_3}(-x)
\\
%-------
f_{\sf w_5}(x) &=& \frac{1}{\sqrt{1+x}\sqrt{2+x}} \equiv f_{\sf w_2}(x+1) 
\\
%-------
f_{\sf w_6}(x) &=& \frac{1}{\sqrt{1-x}\sqrt{2-x}} \equiv f_{\sf w_2}(1-x) 
\\
%-------
f_{\sf w_7}(x) &=& \frac{1}{\sqrt{1-x}\sqrt{2+x}} 
\\
%-------
f_{\sf w_8}(x) &=& \frac{1}{x\sqrt{x-\frac{1}{4}}}
\\
%-------
f_{\sf w_9}(x) &=& \frac{1}{(1-x)\sqrt{x}} \equiv f_{\sf w_3}(1-x)
\\
%-------
f_{\sf w_{10}}(x) &=& \frac{1}{x\sqrt{2-x}}
\\
%-------
f_{\sf w_{11}}(x) &=& \frac{1}{x\sqrt{1-x}\sqrt{2-x}}
\\
%-------
f_{\sf w_{12}}(x) &=& \frac{1}{\sqrt{x}\sqrt{8-x}}
\\
%-------
f_{\sf w_{13}}(x) &=& \frac{1}{(2-x)\sqrt{x}\sqrt{8-x}}
\\
%-------
f_{\sf w_{14}}(x) &=& \frac{1}{x\sqrt{x+\frac{1}{4}}}
\\
%-------
f_{\sf w_{15}}(x) &=& \frac{1}{x\sqrt{1+x}\sqrt{2+x}}
\\
%-------
f_{\sf w_{16}}(x) &=& \frac{1}{x\sqrt{2+x}}
\\
%-------
f_{\sf w_{17}}(x) &=& \frac{1}{\sqrt{x}\sqrt{8+x}}
\\
%-------
f_{\sf w_{18}}(x) &=& \frac{1}{(x+2)\sqrt{x}\sqrt{8+x}}
\\
%-------
f_{\sf w_{19}}(x) &=& \frac{1}{\sqrt{x}\sqrt{4-x}}
\\
%-------
f_{\sf w_{20}}(x) &=& \frac{1}{\sqrt{x}\sqrt{4+x}}
\\
%-------
f_{\sf w_{21}}(x) &=& \frac{1}{\sqrt{4+x}\sqrt{8+x}}
\\
%-------
f_{\sf w_{22}}(x) &=& \frac{1}{(2+x)\sqrt{x-\frac{1}{4}}}
\\
%-------
f_{\sf w_{23}}(x) &=& \frac{1}{(1+x)\sqrt{x}\sqrt{4+x}}
\\
%-------
f_{\sf w_{24}}(x) &=& \frac{1}{(2-x)\sqrt{x+\frac{1}{4}}}
\\
%-------
f_{\sf w_{25}}(x) &=& \frac{1}{\sqrt{4-x}\sqrt{8-x}}
\\
%-------
f_{\sf w_{26}}(x) &=& \frac{1}{(1-x)\sqrt{x}\sqrt{4-x}}
\\
%-------
f_{\sf w_{27}}(x) &=& \frac{1}{(x+\frac{1}{2})\sqrt{x}\sqrt{4-x}}
\\
%-------
f_{\sf w_{28}}(x) &=& \frac{1}{x\sqrt{x+\frac{1}{8}}}
\\
%-------
f_{\sf w_{29}}(x) &=& \frac{1}{(1-x)\sqrt{x-\frac{1}{4}}}
\\
%-------
f_{\sf w_{30}}(x) &=& \frac{1}{(1+x)\sqrt{x}\sqrt{8-x}}
\\
%-------
f_{\sf w_{31}}(x) &=& \frac{1}{(1+x)\sqrt{x+\frac{1}{4}}}
\\
%-------
f_{\sf w_{32}}(x) &=& \frac{1}{(1-x)\sqrt{x}\sqrt{8+x}}.
\end{eqnarray}
%----------------------------------------------------------------------------
Their choice is motivated by a class of functions which emerged in calculating massive 3--loop
Feynman diagrams with local operator insertions in Ref.~\cite{Ablinger:2014yaa}. 
Many iterated integrals formed out of them can be divided into classes
exhibiting special convolution properties. Those classes are parameterized by general parameters, which in course
generalizes the above alphabet in other applications. These relations are proven in Section~\ref{sec:4}.

Since some letters have a non-integrable singularity at the base point $x=1$ we consistently define
%----------------------------------------------------------------------------
\begin{equation}
 \HA_{\sf w}(x):=\int_x^1dt\left(f_{\sf w}(t)-\frac{c}{1-t}\right)+c\ln(1-x),
\end{equation}
%----------------------------------------------------------------------------
where $c$ takes the unique value such that the integrand on the right hand side is integrable at $t=1$. 
It is important to note that this definition preserves the derivative $\frac{d}{dx}\HA_{\sf w}(x)=-f_{\sf w}(x)$. 
In general, for any word $\sf \vec{w}$ we set
%----------------------------------------------------------------------------
\begin{equation}
 \HA_{\sf w,\vec{w}}(x):=\int_x^1dt\left(f_{\sf w}(t)\HA_{\sf \vec{w}}(t)
-\sum_{i=0}^kc_i\frac{\ln(1-t)^i}{1-t}\right)+\sum_{i=0}^k\frac{c_i}{i+1}\ln(1-x)^{i+1},
\end{equation}
%----------------------------------------------------------------------------
where $k$ and $c_0,\dots,c_k$ are chosen to remove any non-integrable singularity. 
This corresponds to Hadamard's finite part \cite{HADAMARD,SCHWARTZ}. Again the result is unique 
and retains 
$\frac{d}{dx}\HA_{\sf w,\vec{w}}(x)=-f_{\sf w}(x)\HA_{\sf \vec{w}}(x)$. 

With this convention, functions of depth 1 can be expressed in terms of elementary functions as follows.
%----------------------------------------------------------------------------
\begin{eqnarray}
\HA_{\sf 0}(x) &=& -\ln(x)
\\
%-------
\HA_{\sf 1}(x) &:=& \ln(1-x)
\\
%-------
\HA_{\sf -1}(x) &=& \ln(2) - \ln(1+x)
\\
%-------
\HA_{\sf -2}(x) &=& \ln(3) - \ln(2+x)
\\
%-------
\HA_{\sf w_1}(x) &=& \arccos(2x-1)
\\
%-------
\HA_{\sf w_2}(x) &=& -\arccosh(2x+1)-2\ln(\sqrt{2}-1)
\\
%-------
\HA_{\sf w_3}(x) &=& - \ln\left[\frac{1-\sqrt{1-x}}{1+\sqrt{1-x}}\right] 
= 2\arctanh(\sqrt{1-x}) \nonumber\\*
&=& 2\arcsech(\sqrt{x}) = \arccosh\left(\frac{2}{x}-1\right)
\\
%-------
\HA_{\sf w_4}(x) &=& 2\arccoth(\sqrt{1+x})+2\ln(\sqrt{2}-1)
\\
%-------
\HA_{\sf w_5}(x) &=& -\arccosh(2x+3)-2\ln(\sqrt{3}-\sqrt{2})
\\
%-------
\HA_{\sf w_6}(x) &=& \arccosh(3-2x)
\\
%-------
\HA_{\sf w_7}(x) &=& \arccos\left(\frac{2x+1}{3}\right)
\\
%-------
\HA_{\sf w_8}(x) &=& 4\left(\frac{\pi}{3}-\arctan(\sqrt{4x-1})\right)
\\
%-------
\HA_{\sf w_9}(x) &:=& 2\ln(2)-2\arctanh(\sqrt{x})
\\
%-------
\HA_{\sf w_{10}}(x) &=& \sqrt{2}\left(\arctanh\left(\sqrt{1-\tfrac{x}{2}}\right)+\ln(\sqrt{2}-1)\right)
\\
%-------
\HA_{\sf w_{11}}(x) &=& \frac{1}{\sqrt{2}}\arccosh\left(\frac{4-3x}{x}\right)
\\
%-------
\HA_{\sf w_{12}}(x) &=& \arccos\left(\frac{3}{4}\right)-\arccos\left(1-\frac{x}{4}\right)
\\
%-------
\HA_{\sf w_{13}}(x) &=& -\frac{1}{2\sqrt{3}}\left(\arccosh\left(\frac{x+4}{2(2-x)}\right)+\ln\left(\frac{5-\sqrt{21}}{2}\right)\right)
\\
%-------
\HA_{\sf w_{14}}(x) &=& 4\left(\arccoth(\sqrt{4x+1})+\ln\left(\frac{\sqrt{5}-1}{2}\right)\right)
\\
%-------
\HA_{\sf w_{15}}(x) &=& \frac{1}{\sqrt{2}}\arccosh\left(\frac{4+3x}{x}\right)+\sqrt{2}\ln(2-\sqrt{3})
\\
%-------
\HA_{\sf w_{16}}(x) &=& \sqrt{2}\left(\arccoth\left(\sqrt{1+\tfrac{x}{2}}\right)+\ln(\sqrt{3}-\sqrt{2})\right)
\\
%-------
\HA_{\sf w_{17}}(x) &=& \ln(2)-\arccosh\left(1+\frac{x}{4}\right)
\\
%-------
\HA_{\sf w_{18}}(x) &=& \frac{1}{\sqrt{3}}\left(\frac{\pi}{6}-\arctan\left(\sqrt{\frac{3x}{8+x}}\right)\right)
\\
%-------
\HA_{\sf w_{19}}(x) &=& \arccos\left(\frac{x}{2}-1\right)-\frac{2\pi}{3}
\\
%-------
\HA_{\sf w_{20}}(x) &=& -\arccosh\left(\frac{x}{2}+1\right)-2\ln\left(\frac{\sqrt{5}-1}{2}\right)
\\
%-------
\HA_{\sf w_{21}}(x) &=& -\arccosh\left(\frac{x}{2}+3\right)-4\ln\left(\frac{\sqrt{5}-1}{2}\right)
\\
%-------
\HA_{\sf w_{22}}(x) &=& \frac{2\pi}{9}-\frac{4}{3}\arctan\left(\frac{2}{3}\sqrt{x-\frac{1}{4}}\right)
\\
%-------
\HA_{\sf w_{23}}(x) &=& \frac{2}{\sqrt{3}}\left(\arctan\left(\sqrt{\frac{3}{5}}\right)-\arctan\left(\sqrt{\frac{3x}{4+x}}\right)\right)
\\
%-------
\HA_{\sf w_{24}}(x) &=& -\frac{4}{3}\left(\arctanh\left(\frac{2}{3}\sqrt{1+\frac{x}{4}}\right)+2\ln\left(\frac{\sqrt{5}-1}{2}\right)\right)
\\
%-------
\HA_{\sf w_{25}}(x) &=& \arccosh\left(3-\frac{x}{2}\right)+\ln\left(\frac{5-\sqrt{21}}{2}\right)
\\
%-------
\HA_{\sf w_{26}}(x) &:=& \frac{1}{\sqrt{3}}\left(\ln(3)-\arccosh\left(\frac{2+x}{2(1-x)}\right)\right)
\\
%-------
\HA_{\sf w_{27}}(x) &=& -\frac{4}{3}\left(\arctan\left(3\sqrt{\frac{x}{4-x}}\right)-\frac{\pi}{3}\right)
\\
%-------
\HA_{\sf w_{28}}(x) &=& 4\sqrt{2}\left(\arccoth\left(\sqrt{8x+1}\right)-\frac{\ln(2)}{2}\right)
\\
%-------
\HA_{\sf w_{29}}(x) &:=& \frac{2}{\sqrt{3}}\left(\ln(3)-\arccosh\left(\frac{2x+1}{2(1-x)}\right)\right)
\\
%-------
\HA_{\sf w_{30}}(x) &=& \frac{1}{3}\left(\arccos\left(-\frac{1}{8}\right)-\arccos\left(\frac{4-5x}{4(1+x)}\right)\right)
\\
%-------
\HA_{\sf w_{31}}(x) &=& \frac{2}{\sqrt{3}}\left(\arccos\left(-\frac{1}{4}\right)-\arccos\left(\frac{1-2x}{2(1+x)}\right)\right)
\\
%-------
\HA_{\sf w_{32}}(x) &:=& \frac{1}{3}\left(2\ln(3)-\ln(2)-\arccosh\left(\frac{5x+4}{4(1-x)}\right)\right).
\end{eqnarray}
%----------------------------------------------------------------------------
Already at weight {\sf w = 1} a series of new constants as
%----------------------------------------------------------------------------
\begin{eqnarray}
\ln(\sqrt{3} - \sqrt{2}), \ln(2 - \sqrt{3}), 
\arccos\left(-\frac{1}{8}\right), 
\arccos\left(-\frac{1}{4}\right),
\arccos\left(\frac{3}{4}\right)
\end{eqnarray}
%----------------------------------------------------------------------------
appears beyond the multiple zeta values \cite{Blumlein:2009cf} and the special numbers in the case 
of cyclotomic and generalized harmonic sums \cite{Ablinger:2011te,Ablinger:2013cf}.
In the following calculations we made frequent use of the computer algebra packages 
{\tt Sigma} \cite{SIG1,SIG2}, {\tt HarmonicSums} 
\cite{Ablinger:2010kw,Ablinger:2011te,Ablinger:2013cf,Ablinger:2013hcp}, 
{\tt EvaluateMultiSums} \cite{Ablinger:2010pb,Blumlein:2012hg,Schneider:2013zna},
{\tt HolonomicFunctions} \cite{HOLF}, {\tt Integrator} \cite{RaabPhD}  and {\tt Singular} 
\cite{SINGULAR}. In different 
representations we will also need the value of the iterated integrals at $x=0$. Many of these 
special constants have been calculated analytically and will be given in Appendix~\ref{sec:A}.
In a series of cases only integral-representations over special functions could be derived,
which are well suited to derive numerical representations.
%----------------------------------------------------------------------------
\section{Convolution Integrals}
\label{sec:4}
%----------------------------------------------------------------------------

\vspace*{1mm}
\noindent
Now we take a closer look at how we calculate Mellin convolutions, which is the most challenging part of the 
computation. Formally, we rely on the convolution formula \eqref{eq:MellinConv}, which gives us a definite 
integral depending on a continuous parameter and can be written in the form
%----------------------------------------------------------------------------
\begin{equation}
F(x)=\int_x^1dyf(x,y).
\end{equation}
%----------------------------------------------------------------------------
In order to obtain a closed form for this integral, we first set up a differential equation satisfied by 
$F(x)$ and then obtain a solution of this equation satisfying appropriate initial conditions. In the 
first step we exploit the principle of differentiation under the integral. If we have a relation for 
the integrand $f(x,y)$ of the form
%----------------------------------------------------------------------------
\begin{equation}\label{eq:ParametricIntegration}
c_m(x)\frac{\partial^mf}{\partial{x}^m}(x,y)+\dots+c_0(x)f(x,y)=\frac{\partial g}{\partial{y}}(x,y)
\end{equation}
%----------------------------------------------------------------------------
for some coefficients $c_i(x)$ independent of $y$ and some function $g(x,y)$, then by applying $\int_x^1dy$ 
this gives rise to a linear ordinary differential equation for the integral $F(x)$
%----------------------------------------------------------------------------
\begin{equation}\label{eq:ResultingODE}
c_m(x)F^{(m)}(x)+\dots+c_0(x)F(x)=g(x,1)-g(x,x)+\text{additional boundary terms}.
\end{equation}
%----------------------------------------------------------------------------
Proper care has to be taken for evaluating the right hand side of this relation in the presence of singularities. 
There are several computer algebra algorithms for various types of integrands $f(x,y)$ which, given $f(x,y)$, 
compute relations of the form \eqref{eq:ParametricIntegration}. They either utilize differential fields 
\cite{Risch,Bronstein,RaabPhD} or holonomic systems and 
Ore algebras \cite{Ore:1933,AlmkvistZeilberger,Chyzak,Koutschan}. Corresponding algorithms are implemented in the 
packages {\tt Integrator} \cite{RaabPhD} and {\tt HolonomicFunctions} \cite{HOLF}, respectively.
For obtaining solutions to the generated differential 
equations the following two observations are crucial. All differential
equations obtained during our computations factor 
completely into first-order equations with rational function coefficients (the
resulting solutions are also called d'Alembertian solutions~\cite{Abramov:94})
and, moreover, these factors all have algebraic 
functions of degree at most two as their solutions. These two observations imply that solutions are of the form
%----------------------------------------------------------------------------
\begin{equation}\label{Equ:NestedIntegrals}
\frac{r_1(x)}{\sqrt{p_1(x)}} \int dx \frac{r_2(x)}{\sqrt{p_2(x)}} \int dx \dots \int dx \frac{r_k(x)}{\sqrt{p_k(x)}},
\end{equation}
%----------------------------------------------------------------------------
where $r_i(x)$ are rational functions and $p_i(x)$ are square-free polynomials. Using a dedicated rewrite procedure 
based on integration by parts we can write a basis of the solution space in terms of the functions $\HA$ over the alphabet 
defined earlier, which is then used to match initial conditions. Already Hermite considered a reduction procedure for simple 
integrals of the form 
%----------------------------------------------------------------------------
\begin{equation}
\int dx \frac{r(x)}{\sqrt{p(x)}} 
\end{equation}
%----------------------------------------------------------------------------
similar to ours, see \cite{Hermite83}. 

As will be shown later, in the concrete examples we computed,  some patterns emerge. Before proving those 
patterns we give a 
few general identities, which enable us to recursively rewrite many of the convolution integrals we encountered in 
terms of nested integrals in a rather direct way. These identities are at the core of what is going on in the 
calculations. They can also serve as an alternative way of generating the differential equations mentioned above, 
explaining their nice factorization properties. Already in Refs.~\cite{GeddesLeLi,ChenKauersSinger,BostanEtAl} 
the respective authors proposed algorithms which for certain types of integrands compute the differential operators 
in partially factored form. Our identities below follow a similar spirit.
%----------------------------------------------------------------------------
\begin{lemma}\label{lem:Convolution1a}
 Let $c<1$ and let $f(x)$ be a differentiable function on $]c,1[$ with locally bounded derivative. Then for
 sufficiently small $\ep>0$ and all $x \in {]c,1-\ep[}$ we have
%----------------------------------------------------------------------------
 \begin{equation}\label{eq:Convolution1a}
  \int_x^{1-\ep} dt \frac{f(t)}{(t-c)\sqrt{t-x}} 
= \frac{1}{\sqrt{x-c}}\int_x^{1-\ep} dt \frac{1}{\sqrt{t-c}} 
\left(\frac{f(1-\ep)}{\sqrt{1-t-\ep}}-\int_t^{1-\ep} du \frac{f^\prime(u)}{\sqrt{u-t}}\right).
 \end{equation}
%----------------------------------------------------------------------------
\end{lemma}
%----------------------------------------------------------------------------
\begin{proof}
 We prove \eqref{eq:Convolution1a} by verifying that both sides of the equation
satisfy the following initial value problem for $y(x)$:
%----------------------------------------------------------------------------
 \begin{eqnarray*}
  y^\prime(x)+\frac{1}{2(x-c)}y(x) &=& -\frac{1}{x-c}\left(\frac{f(1-\ep)}{\sqrt{1-x-\ep}}-\int_x^{1-\ep} dt 
\frac{f^\prime(t)}{\sqrt{t-x}}\right)\\
  y(1-\ep) &=& 0.
 \end{eqnarray*}
%----------------------------------------------------------------------------
 To this end, we first plug $y(x)=\int_{x+\delta}^{1-\ep} dt \frac{f(t)}{(t-c)\sqrt{t-x}}$, where $\delta>0$ is 
small, into the left hand side of the differential equation to obtain
%----------------------------------------------------------------------------
 \begin{eqnarray*}
  y^\prime(x)+\frac{1}{2(x-c)}y(x) &=& -\frac{f(x+\delta)}{(x+\delta-c)\sqrt{
\delta}}+\frac{1}{x-c}\int_{x+\delta}^{1-\ep} dt \frac{f(t)}{2(t-x)^{3/2}}\\
  &=& 
\frac{\sqrt{\delta}f(x+\delta)}{(x-c)(x+\delta-c)}-\frac{f(1-\ep)}{(x-c)\sqrt{1-\ep-x}}+\frac{1}{x-c}\int_{x+\delta}^{1-\ep}dt\frac{f^\prime(t)}{\sqrt{t-x}}.
%----------------------------------------------------------------------------
 \end{eqnarray*}
We can now send $\delta\to0$ to arrive at the right hand side of the
differential equation. By the assumptions we conclude that the limit
$y(x)=\int_x^{1-\ep} dt \frac{f(t)}{(t-c)\sqrt{t-x}}$ satisfies the differential
equation above. It is also continuous at $x=1-\ep$ and satisfies the initial
condition $y(1-\ep)=0$. Next, we set $y(x) = \frac{1}{\sqrt{x-c}}\int_x^{1-\ep}
dt \frac{1}{\sqrt{t-c}} \left(\frac{f(1-\ep)}{\sqrt{1-t-\ep}}-\int_t^{1-\ep} du
\frac{f^\prime(u)}{\sqrt{u-t}}\right)$ and straightforwardly verify that it
satisfies the initial value problem as well. By uniqueness of the solution of
the initial value problem on the interval $]c,1-\ep]$ we infer
\eqref{eq:Convolution1a}.
\end{proof}
%----------------------------------------------------------------------------
Similarly, we get the following lemmas.
%----------------------------------------------------------------------------
\begin{lemma}\label{lem:Convolution1b}
For $\ep>0$, $x<1-\ep$, $a<x$,
 \[
  \int_x^{1-\ep} dt \frac{f(t)}{\sqrt{(t-a)(t-x)}} = \int_x^{1-\ep} dt \frac{1}{t-a} 
\left(\frac{\sqrt{1-a-\ep}f(1-\ep)}{\sqrt{1-t-\ep}}-\int_t^{1-\ep} du \frac{\sqrt{u-a}f^\prime(u)}{\sqrt{u-t}}\right).
 \]
\end{lemma}
%----------------------------------------------------------------------------
%----------------------------------------------------------------------------
\begin{lemma}\label{lem:Convolution1c}
For $\ep>0$, $x<1-\ep$, $a<x$, $c<x$,
 \begin{multline*}
  \int_x^{1-\ep} dt \frac{f(t)}{(t-c)\sqrt{(t-a)(t-x)}} = \frac{1}{\sqrt{x-c}} \int_x^{1-\ep} dt \frac{1}{(t-a)\sqrt{t-c}} \bigg(\frac{\sqrt{1-a-\ep}f(1-\ep)}{\sqrt{1-t-\ep}}-
 \\*
  -\int_t^{1-\ep} du \frac{\sqrt{u-a}f^\prime(u)}{\sqrt{u-t}}\bigg).
 \end{multline*}
\end{lemma}
%----------------------------------------------------------------------------

From Lemma~\ref{lem:Convolution1a}, using Lemma~\ref{lem:Convolution1b}, we also obtain
%----------------------------------------------------------------------------
\begin{eqnarray}\label{eq:Convolution3}
 \int_x^{1-\ep} dt \frac{\sqrt{t}f(t)}{(t-c)\sqrt{t-x}} &=& \int_x^{1-\ep} dt \frac{1}{t} \left(\frac{\sqrt{1-\ep}f(1-\ep)}{\sqrt{1-t-\ep}}-\int_t^{1-\ep} du \frac{\sqrt{u}f^\prime(u)}{\sqrt{u-t}}\right)\\*
 &&+\frac{c}{\sqrt{x-c}}\int_x^{1-\ep} dt \frac{1}{t\sqrt{t-c}} \left(\frac{\sqrt{1-\ep}f(1-\ep)}
{\sqrt{1-t-\ep}}-\int_t^{1-\ep} du \frac{\sqrt{u}f^\prime(u)}{\sqrt{u-t}}\right).\N
\end{eqnarray}
%----------------------------------------------------------------------------
We may combine Lemma~\ref{lem:Convolution1a} and Lemma~\ref{lem:Convolution1c}
into one formula using $e \in \{0,1\}$. This gives
%----------------------------------------------------------------------------
\begin{eqnarray}\label{eq:Convolution5}
 \int_x^{1-\ep} dt \frac{f(t)}{(t-c)\sqrt{(t-a)^e(t-x)}} &=& \frac{1}{\sqrt{x-c}} \int_x^{1-\ep} dt \frac{1}{(t-a)^e\sqrt{t-c}} \bigg(\frac{\sqrt{(1-a-\ep)^e}f(1-\ep)}{\sqrt{1-t-\ep}}-\N\\*
 &&-\int_t^{1-\ep} du \frac{\sqrt{(u-a)^e}f^\prime(u)}{\sqrt{u-t}}\bigg).
\end{eqnarray}
%----------------------------------------------------------------------------

%----------------------------------------------------------------------------
\begin{lemma}\label{lem:Convolution4}
For $\ep>0$, $0<x<1-\ep$, $a<x$, $c>1$, $ac<x$,
%----------------------------------------------------------------------------
 \begin{eqnarray*}
  \int_x^{1-\ep} dt \frac{f(t)}{(x-ct)\sqrt{(t-a)(t-x)}} &=& \frac{1}{\sqrt{x(x-ac)}} \int_x^{1-\ep} dt \sqrt{\frac{t}{t-ac}} \cdot\\*
  &&\cdot\bigg(\frac{(t^2-ac(1-\ep))\sqrt{1-a-\ep}f(1-\ep)}{t(t-a)(t-c(1-\ep))\sqrt{1-t-\ep}}-\\*
  &&-\int_t^{1-\ep} du \frac{\sqrt{u-a}f^\prime(u)}{(t-cu)\sqrt{u-t}}-\frac{a}{t(t-a)}\int_t^{1-\ep} du 
\frac{\sqrt{u-a}f^\prime(u)}{\sqrt{u-t}}\bigg).
 \end{eqnarray*}
%----------------------------------------------------------------------------
\end{lemma}
%----------------------------------------------------------------------------

%----------------------------------------------------------------------------
\begin{lemma}\label{lem:Convolution41i}
For $\ep>0$, $0<x<1-\ep$, $a<x$, $c>1$, $ac<x$,
%----------------------------------------------------------------------------
 \begin{eqnarray*}
  \int_x^{1-\ep} dt \frac{\HA_{\sf w_3}(\frac{x}{t})f(t)}{(x-ct)\sqrt{t(t-a)}} &=& \frac{1}{\sqrt{x(x-ac)}} \int_x^{1-\ep} dt \frac{1}{\sqrt{t(t-ac)}} \cdot\\*
  &&\cdot\bigg(\frac{c\sqrt{(1-\ep)(1-a-\ep)}\HA_{\sf w_3}(\frac{t}{1-\ep})f(1-\ep)}{t-c(1-\ep)}-\\*
  &&-c\int_t^{1-\ep} du \frac{\sqrt{u(u-a)}\HA_{\sf w_3}(\frac{t}{u})f^\prime(u)}{t-cu}+\int_t^{1-\ep} du 
\frac{f(u)}{\sqrt{(u-a)(u-t)}}\bigg).
 \end{eqnarray*}
%----------------------------------------------------------------------------
\end{lemma}
%----------------------------------------------------------------------------

%----------------------------------------------------------------------------
\begin{lemma}\label{lem:Convolution41ii}
For $\ep>0$, $0<x<1-\ep$, $a<x$, $c>1$, $ac<x$,
%  \[
%   \int_x^{1-\ep} dt \frac{t\HA_{\sf w_3}(\frac{x}{t})f(t)}{(x-ct)\sqrt{t(t-a)}} = \int_x^{1-\ep} dt \frac{1}{t-ac} \left(\frac{\sqrt{(1-\ep)(1-a-\ep)}\HA_{\sf w_3}(\frac{t}{1-\ep})f(1-\ep)}{t-c(1-\ep)}
% \right.\]\[\left.
%   -\int_t^{1-\ep} du \frac{\sqrt{u(u-a)}\HA_{\sf w_3}(\frac{t}{u})f^\prime(u)}{t-cu}+\frac{a}{2}\int_t^{1-\ep} du \frac{\HA_{\sf w_3}(\frac{t}{u})f(u)}{(t-cu)\sqrt{u(u-a)}}+\frac{a}{t}\int_t^{1-\ep} du \frac{f(u)}{\sqrt{(u-a)(u-t)}}\right)
%  \]
%----------------------------------------------------------------------------
 \begin{eqnarray*}
  \int_x^{1-\ep} dt \frac{t\HA_{\sf w_3}(\frac{x}{t})f(t)}{(x-ct)\sqrt{t(t-a)}} &=& \sqrt{\frac{x}{x-ac}} \int_x^{1-\ep} dt \frac{1}{\sqrt{t(t-ac)}} \cdot\\*
  &&\cdot\bigg(\frac{\sqrt{(1-\ep)(1-a-\ep)}\HA_{\sf w_3}(\frac{t}{1-\ep})f(1-\ep)}{t-c(1-\ep)}\\*
  &&-\int_t^{1-\ep} du \frac{\sqrt{u(u-a)}\HA_{\sf w_3}(\frac{t}{u})f^\prime(u)}{t-cu}+\frac{1}{c}\int_t^{1-\ep} du 
\frac{f(u)}{\sqrt{(u-a)(u-t)}}\bigg)\\*
  &&- \frac{1}{c}\int_x^{1-\ep} dt \frac{1}{t}\int_t^{1-\ep} du \frac{f(u)}{\sqrt{(u-a)(u-t)}}.
 \end{eqnarray*}
%----------------------------------------------------------------------------
\end{lemma}
%----------------------------------------------------------------------------
Note that the previous formula may also be written as
%----------------------------------------------------------------------------
\begin{eqnarray}\label{eq:Convolution41ii}
 \int_x^{1-\ep} dt \frac{t\HA_{\sf w_3}(\frac{x}{t})f(t)}{(x-ct)\sqrt{t(t-a)}} &=& \frac{x}{c}\int_x^{1-\ep} dt \frac{\HA_{\sf w_3}(\frac{x}{t})f(t)}{(x-ct)\sqrt{t(t-a)}}
\N\\*&&
 - \frac{1}{c}\int_x^{1-\ep} dt \frac{1}{t}\int_t^{1-\ep} du \frac{f(u)}{\sqrt{(u-a)(u-t)}}.
\end{eqnarray}
For later use we also note that the following identities hold:
\begin{eqnarray}
 \int_x^1 dt \frac{\HA_{\sf w_3}(\frac{x}{t})}{(x-ct)\sqrt{t(1-t)}} &=& -\frac{\pi}{\sqrt{x(c-x)}}\int_x^1 dt \frac{1}{\sqrt{t(c-t)}}\label{eq:Convolution41iBase}\\
 \int_x^1 dt \frac{t\HA_{\sf w_3}(\frac{x}{t})}{(x-ct)\sqrt{t(1-t)}} &=&
-\frac{\pi}{c}\left(\HA_{\sf 0}(x)+\sqrt{\frac{x}{c-x}}\int_x^1 dt
\frac{1}{\sqrt{t(c-t)}}\right).\label{eq:Convolution41iiBase}
\end{eqnarray}
%----------------------------------------------------------------------------

The following theorem allows us to set up Mellin representations of expressions
of the form
\begin{equation}\label{eq:SumPattern1}
 \binom{2N}{N}S_{k_1,\dots,k_m}(x_1,\dots,x_m)(N)
\end{equation}
in a direct way, thereby removing the need to carry out the calculations for individual Mellin convolutions in such cases.

\begin{theorem}\label{thm:Pattern1}
 Let $a_0,\dots,a_k<0$ and define $f_{\sf b_i}(x):=\frac{1}{\sqrt{(x-a_i)(x-a_{i+1})}}$ for $i \in \{0,\dots,k-1\}$ and $f_{\sf b_k}(x):=\frac{1}{\sqrt{(1-x)(x-a_k)}}$. Then we have
 \begin{equation}\label{eq:ConvolutionPattern1}
  \int_x^1dt\frac{\HA_{\sf a_1,\dots,a_k}(t)}{(t-a_0)\sqrt{t-x}} = \frac{\HA_{\sf b_0,\dots,b_k}(x)}{\sqrt{x-a_0}}
 \end{equation}
 and
 \begin{equation}\label{eq:Pattern1}
  \binom{2N}{N}\Mvec\left[\frac{\HA_{\sf a_1,\dots,a_k}(x)}{x-a_0}\right](N) = \frac{4^N}{\pi}\Mvec\left[\frac{\HA_{\sf b_0,\dots,b_k}(x)}{\sqrt{x(x-a_0)}}\right](N).
 \end{equation}
\end{theorem}
\begin{proof}
 We prove \eqref{eq:ConvolutionPattern1} by induction on $k$. For $k=0$ we easily obtain
 \[
  \int_x^1dt\frac{1}{(t-a_0)\sqrt{t-x}} = \frac{1}{\sqrt{x-a_0}}\int_x^1 dt \frac{1}{\sqrt{(1-t)(t-a_0)}} = \frac{\HA_{\sf b_0}(x)}{\sqrt{x-a_0}}
 \]
 from Lemma~\ref{lem:Convolution1a}. For $k>0$ we assume
that \eqref{eq:ConvolutionPattern1} holds for all values smaller than $k$.
Lemma~\ref{lem:Convolution1a} now yields
 \[
  \int_x^1dt\frac{\HA_{\sf a_1,\dots,a_k}(t)}{(t-a_0)\sqrt{t-x}} = \frac{1}{\sqrt{x-a_0}}\int_x^1 dt \frac{1}{\sqrt{t-a_0}} \int_t^1 du \frac{\HA_{\sf a_2,\dots,a_k}(u)}{(u-a_1)\sqrt{u-t}}.
 \]
 Applying the induction hypothesis to the inner integral we obtain
 \[
  \int_t^1 du \frac{\HA_{\sf a_2,\dots,a_k}(u)}{(u-a_1)\sqrt{u-t}} = \frac{\HA_{\sf b_1,\dots,b_k}(t)}{\sqrt{t-a_1}},
 \]
 which concludes the proof of \eqref{eq:ConvolutionPattern1}.
 Finally, by \eqref{eq:RepresentBinomial} and \eqref{eq:MellinMult} we get
 \[
  \binom{2N}{N}\Mvec\left[\frac{\HA_{\sf a_1,\dots,a_k}(x)}{x-a_0}\right](N) =
\frac{4^N}{\pi}\Mvec\left[\frac{1}{\sqrt{x}}\int_x^1dt\frac{\HA_{\sf
a_1,\dots,a_k}(t)}{(t-a_0)\sqrt{t-x}}\right](N).
 \]
\end{proof}

Similarly, the next theorem provides formulae which facilitate the computation of Mellin representations of expressions of the form
\begin{equation}\label{eq:SumPattern2}
 \frac{1}{(2N+1)\binom{2N}{N}}\sum_{i=1}^Nx_0^i\binom{2i}{i}S_{k_1,\dots,k_m}(x_1,\dots,x_m)(i)
\end{equation}
or having similar pre-factors.

\begin{theorem}\label{thm:Pattern2}
 Let $a_0,\dots,a_k,c<0$ and define $f_{\sf b_i}(x):=\frac{1}{\sqrt{(x-a_i)(x-a_{i+1})}}$ for $i \in \{0,\dots,k-1\}$, $f_{\sf b_k}(x):=\frac{1}{\sqrt{(1-x)(x-a_k)}}$, $f_{\sf v}(x):=\frac{1}{x\sqrt{x-c}}$, and $f_{\sf w}(x):=\frac{1}{(x-a_0)\sqrt{x-c}}$. Then we have
 \begin{eqnarray}
  \int_x^1 dt \frac{\sqrt{t-a_0}f_{\sf b_0}(t)\HA_{\sf b_1,\dots,b_k}(t)}{\sqrt{t-x}} &=& \pi \HA_{\sf a_1,\dots,a_k}(x)\label{eq:ConvolutionPattern2a}\\
  \int_x^1 dt \frac{\HA_{\sf b_0,\dots,b_k}(t)}{(t-c)\sqrt{(t-a_0)(t-x)}} &=& \frac{\pi}{\sqrt{x-c}} \int_x^1 dt \frac{\HA_{\sf a_1,\dots,a_k}(t)}{(t-a_0)\sqrt{t-c}}\label{eq:ConvolutionPattern2}
 \end{eqnarray}
 and
 \begin{eqnarray}
  \frac{1}{(2N+1)\binom{2N}{N}}\Mvec\left[\frac{x\HA_{\sf b_0,\dots,b_k}(x)}{(x-c)\sqrt{x(x-a_0)}}\right](N) &=& \frac{\pi}{2{\cdot}4^N}\Mvec\left[\frac{\HA_{\sf w,a_1,\dots,a_k}(x)}{\sqrt{x-c}}\right](N)\label{eq:Pattern2}\\
  \frac{1}{(N+1)\binom{2N}{N}}\Mvec\left[\frac{x\HA_{\sf b_0,\dots,b_k}(x)}{(x-c)\sqrt{x(x-a_0)}}\right](N) &=& \frac{\pi}{4^N}\Mvec\Bigg[\frac{\HA_{\sf w,a_1,\dots,a_k}(x)}{\sqrt{x-c}}
  \N\\*&&
  -\frac{1}{2}\HA_{\sf v,w,a_1,\dots,a_k}(x)\Bigg](N)\label{eq:Pattern2a}\\
  \frac{1}{N\binom{2N}{N}}\Mvec\left[\frac{x\HA_{\sf b_0,\dots,b_k}(x)}{(x-c)\sqrt{x(x-a_0)}}\right](N) &=& \frac{\pi}{4^N}\Mvec\Bigg[\frac{\HA_{\sf a_0,\dots,a_k}(x)}{x}
  \N\\*&&
  +c\frac{\HA_{\sf w,a_1,\dots,a_k}(x)}{x\sqrt{x-c}}\Bigg](N)\label{eq:Pattern2b}.
 \end{eqnarray}
\end{theorem}
\begin{proof}
 We have $\frac{1}{(2N+1)\binom{2N}{N}}=\frac{2}{(N+1)\binom{2(N+1)}{N+1}}$, hence \eqref{eq:RepresentInverseBinomial} implies
 \[
  \frac{1}{(2N+1)\binom{2N}{N}}=\frac{1}{2{\cdot}4^N}\Mvec\left[\frac{1}{\sqrt{1-x}}\right](N).
 \]
 Then, \eqref{eq:MellinMult} yields
 \[
  \frac{1}{(2N+1)\binom{2N}{N}}\Mvec\left[\frac{x\HA_{\sf b_0,\dots,b_k}(x)}{(x-c)\sqrt{x(x-a_0)}}\right](N) = \frac{1}{2{\cdot}4^N}\Mvec\left[\int_x^1 dt \frac{\HA_{\sf b_0,\dots,b_k}(t)}{(t-c)\sqrt{(t-a_0)(t-x)}}\right](N)
 \]
 implying \eqref{eq:Pattern2} by \eqref{eq:ConvolutionPattern2}. In addition, the identity $\frac{1}{N+1}=\left(2-\frac{1}{N+1}\right)\frac{1}{2N+1}$ allows us to infer \eqref{eq:Pattern2a} from \eqref{eq:Pattern2} and \eqref{eq:MellinMult}. Furthermore, \eqref{eq:RepresentInverseBinomial} and \eqref{eq:MellinMult} imply
 \begin{eqnarray*}
  \frac{1}{N\binom{2N}{N}}\Mvec\left[\frac{x\HA_{\sf b_0,\dots,b_k}(x)}{(x-c)\sqrt{x(x-a_0)}}\right](N) &=& \frac{1}{4^N}\Mvec\left[\int_x^1 dt \frac{t\HA_{\sf b_0,\dots,b_k}(t)}{x(t-c)\sqrt{(t-a_0)(t-x)}}\right](N)\\
  &=&\frac{1}{4^N}\Mvec\left[\frac{1}{x}\int_x^1 dt \frac{\HA_{\sf b_0,\dots,b_k}(t)}{\sqrt{(t-a_0)(t-x)}}
\right.\\*&&\left.
  +\frac{c}{x}\int_x^1 dt \frac{\HA_{\sf b_0,\dots,b_k}(t)}{(t-c)\sqrt{(t-a_0)(t-x)}}\right](N).
 \end{eqnarray*}
 By \eqref{eq:ConvolutionPattern2a} and \eqref{eq:ConvolutionPattern2} this implies \eqref{eq:Pattern2b}. Next, we apply Lemma~\ref{lem:Convolution1c} to obtain
 \[
  \int_x^1 dt \frac{\HA_{\sf b_0,\dots,b_k}(t)}{(t-c)\sqrt{(t-a_0)(t-x)}} = \frac{1}{\sqrt{x-c}} \int_x^1 dt \frac{1}{(t-a_0)\sqrt{t-c}} \int_t^1 du \frac{\sqrt{u-a_0}f_{\sf b_0}(u)\HA_{\sf b_1,\dots,b_k}(u)}{\sqrt{u-t}},
 \]
 which reduces \eqref{eq:ConvolutionPattern2} to \eqref{eq:ConvolutionPattern2a}. Finally, it remains to prove \eqref{eq:ConvolutionPattern2a}. We proceed by induction on $k$. For $k=0$ we directly obtain
 \[
  \int_t^1 du \frac{\sqrt{u-a_0}f_{\sf b_0}(u)}{\sqrt{u-t}} = \int_t^1 du \frac{1}{\sqrt{(1-u)(u-t)}} = \pi.
 \]
 For $k>0$ we assume \eqref{eq:ConvolutionPattern2a} holds for values smaller than $k$. 
Then, Lemma~\ref{lem:Convolution1b} implies
 \[
  \int_t^1 du \frac{\HA_{\sf b_1,\dots,b_k}(u)}{\sqrt{(u-a_1)(u-t)}} = \int_t^1 du \frac{1}{u-a_1}\int_u^1 dv \frac{\sqrt{v-a_1}f_{\sf b_1}(v)\HA_{\sf b_2,\dots,b_k}(v)}{\sqrt{v-t}},
 \]
 where we just apply the induction hypothesis to the inner integral now.
\end{proof}

The pattern shown in the next theorem emerges in Mellin representations of expressions of the form
\begin{equation}\label{eq:SumPattern4}
 \sum_{i=1}^N\frac{x_0^i}{i\binom{2i}{i}}\cdot\binom{2N}{N}S_{k_1,\dots,k_m}(x_1,\dots,x_m)(N).
\end{equation}
Note that expressions of this type are not strictly nested sums, but rather the product of two nested sums. The same will be true for the variation considered afterwards.

\begin{theorem}\label{thm:Pattern4}
 Let $a_0,\dots,a_k<0$ and $c>1$. Define $f_{\sf b_i}(x):=\frac{1}{\sqrt{(x-a_i)(x-a_{i+1})}}$ and $f_{\sf c_i}(x):=\frac{1}{\sqrt{(x-a_ic)(x-a_{i+1}c)}}$ for $i \in \{0,\dots,k-1\}$ as well as $f_{\sf b_k}(x):=\frac{1}{\sqrt{(1-x)(x-a_k)}}$ and $f_{\sf c_k}(x):=\frac{1}{\sqrt{(c-x)(x-a_kc)}}$. Define further $f_{\sf d_i}(x):=\frac{1}{(x-a_i)\sqrt{x(x-a_ic)}}$ for $i \in \{0,\dots,k\}$. Then we have
 \begin{equation}\label{eq:ConvolutionPattern4}
  \int_x^1dt\frac{\HA_{\sf b_0,\dots,b_k}(t)}{(x-ct)\sqrt{(t-a_0)(t-x)}} = \frac{\pi}{\sqrt{x(x-a_0c)}}\left(\sum_{i=0}^ka_i\HA_{\sf c_0,\dots,c_{i-1},d_i,a_{i+1},\dots,a_k}(x)-\frac{\HA_{\sf c_0,\dots,c_k}(x)}{\sqrt{c-1}}\right)
 \end{equation}
 and
 \begin{eqnarray}
  \sum_{i=1}^N\frac{(\frac{4}{c})^i}{i\binom{2i}{i}}\cdot\Mvec\left[\frac{\HA_{\sf b_0,\dots,b_k}(x)}{\sqrt{x(x-a_0)}}\right](N) &=& \frac{\pi}{c^N}\Mvec\left[\frac{\sum\limits_{i=0}^ka_i\HA_{\sf c_0,\dots,c_{i-1},d_i,a_{i+1},\dots,a_k}(x)-\frac{1}{\sqrt{c-1}}\HA_{\sf c_0,\dots,c_k}(x)}{\sqrt{x(x-a_0c)}}\right](N)\N\\*
  &&+\sum_{i=1}^\infty\frac{(\frac{4}{c})^i}{i\binom{2i}{i}}\cdot\Mvec\left[\frac{\HA_{\sf b_0,\dots,b_k}(x)}{\sqrt{x(x-a_0)}}\right](N)\label{eq:Pattern4}.
 \end{eqnarray}
\end{theorem}
\begin{proof}
 We prove \eqref{eq:ConvolutionPattern4} by induction on $k$. For $k=0$ Lemma~\ref{lem:Convolution4} yields
 \begin{eqnarray*}
  \int_x^1dt\frac{\HA_{\sf b_0}(t)}{(x-ct)\sqrt{(t-a_0)(t-x)}} &=& \frac{1}{\sqrt{x(x-a_0c)}} \int_x^1 dt \sqrt{\frac{t}{t-a_0c}} \left(\int_t^1 du \frac{\sqrt{u-a_0}f_{\sf b_0}(u)}{(t-cu)\sqrt{u-t}}+\right.\\*
  && \left.+\frac{a_0}{t(t-a_0)}\int_t^1 du \frac{\sqrt{u-a_0}f_{\sf b_0}(u)}{\sqrt{u-t}}\right)\\
  &=& \frac{1}{\sqrt{x(x-a_0c)}} \int_x^1 dt \sqrt{\frac{t}{t-a_0c}} \Bigg(-\frac{\pi}{\sqrt{c-1}\sqrt{t(c-t)}}+\\*
  &&+\frac{a_0\pi}{t(t-a_0)}\Bigg)\\
  &=& \frac{\pi}{\sqrt{x(x-a_0c)}}\left(a_0\HA_{\sf d_0}(x)-\frac{\HA_{\sf c_0}(x)}{\sqrt{c-1}}\right).
 \end{eqnarray*}
 For $k>0$ we assume that \eqref{eq:ConvolutionPattern4} holds for values smaller than $k$. 
Using Lemma~\ref{lem:Convolution4}, the induction hypothesis, and \eqref{eq:ConvolutionPattern2a} we obtain
 \begin{eqnarray*}
  \int_x^1dt\frac{\HA_{\sf b_0,\dots,b_k}(t)}{(x-ct)\sqrt{(t-a_0)(t-x)}} &=& \frac{1}{\sqrt{x(x-a_0c)}} \int_x^1 dt \sqrt{\frac{t}{t-a_0c}} \cdot\\*
  &&\cdot\Bigg(\int_t^1 du \frac{\HA_{\sf b_1,\dots,b_k}(u)}{(t-cu)\sqrt{(u-a_1)(u-t)}}+\\*
  &&+\frac{a_0}{t(t-a_0)}\int_t^1 du \frac{\HA_{\sf b_1,\dots,b_k}(u)}{\sqrt{(u-a_1)(u-t)}}\Bigg)\\
  &=& \frac{1}{\sqrt{x(x-a_0c)}} \int_x^1 dt \sqrt{\frac{t}{t-a_0c}} \left(\frac{\pi}{\sqrt{t(t-a_1c)}}\cdot\right.\\*
  && \left.\cdot\left(\sum_{i=1}^ka_i\HA_{\sf c_1,\dots,c_{i-1},d_i,a_{i+1},\dots,a_k}(t)-\frac{\HA_{\sf c_1,\dots,c_k}(t)}{\sqrt{c-1}}\right)+\right.\\*
  &&\left.+\frac{a_0}{t(t-a_0)}\pi\HA_{\sf a_1,\dots,a_k}(t)\right)\\
  &=& \frac{\pi}{\sqrt{x(x-a_0c)}} \Bigg(\sum_{i=1}^ka_i\HA_{\sf c_0,\dots,c_{i-1},d_i,a_{i+1},\dots,a_k}(x)-\frac{\HA_{\sf c_0,\dots,c_k}(x)}{\sqrt{c-1}}\\*
  &&+a_0\HA_{\sf d_0,a_1,\dots,a_k}(x)\Bigg).
 \end{eqnarray*}
 This establishes \eqref{eq:ConvolutionPattern4}. Next, from \eqref{eq:RepresentInverseBinomial} and \eqref{eq:MellinSummation} we obtain
 \[
  \sum_{i=1}^N\frac{(\frac{4}{c})^i}{i\binom{2i}{i}} = c^{-N}\Mvec\left[\frac{1}{(x-c)\sqrt{1-x}}\right](N)+\sum_{i=1}^\infty\frac{(\frac{4}{c})^i}{i\binom{2i}{i}}.
 \]
 Hence by \eqref{eq:MellinMult} we have
 \begin{eqnarray*}
  \sum_{i=1}^N\frac{(\frac{4}{c})^i}{i\binom{2i}{i}}\cdot\Mvec\left[\frac{\HA_{\sf b_0,\dots,b_k}(x)}{\sqrt{x(x-a_0)}}\right](N) &=& c^{-N}\Mvec\left[\int_x^1dt\frac{\HA_{\sf b_0,\dots,b_k}(t)}{(x-ct)\sqrt{(t-a_0)(t-x)}}\right](N)\\*
  && +\sum_{i=1}^\infty\frac{(\frac{4}{c})^i}{i\binom{2i}{i}}\Mvec\left[\frac{\HA_{\sf b_0,\dots,b_k}(x)}{\sqrt{x(x-a_0)}}\right](N).
 \end{eqnarray*}
 Finally, this implies \eqref{eq:Pattern4} by virtue of \eqref{eq:ConvolutionPattern4}.
\end{proof}
Note that we can write the infinite sum above in terms of $\HA$ again.
\begin{equation}
 \sum_{i=1}^\infty\frac{(\frac{4}{c})^i}{i\binom{2i}{i}} = \frac{\mathrm{H}_{\sf w_1}(\frac{1}{c})}{\sqrt{c-1}} = \frac{\pi-\HA_{\sf w_1}(\frac{1}{c})}{\sqrt{c-1}}
\end{equation}

As a variation of the previous type of expressions we also deal with expressions of the form
\begin{equation}\label{eq:SumPattern4a}
 \sum_{i=1}^N\frac{x_0^i}{i^2\binom{2i}{i}}\cdot\binom{2N}{N}S_{k_1,\dots,k_m}(x_1,\dots,x_m)(N).
\end{equation}

\begin{theorem}\label{thm:Pattern4a}
 Let $a_0,\dots,a_k<0$ and $c>1$. Define $f_{\sf b_i}(x):=\frac{1}{\sqrt{(x-a_i)(x-a_{i+1})}}$ and $f_{\sf c_i}(x):=\frac{1}{\sqrt{(x-a_ic)(x-a_{i+1}c)}}$ for $i \in \{0,\dots,k-1\}$ as well as $f_{\sf b_k}(x):=\frac{1}{\sqrt{(1-x)(x-a_k)}}$ and $f_{\sf c_k}(x):=\frac{1}{\sqrt{(c-x)(x-a_kc)}}$. Define further $f_{\sf d_i}(x):=\frac{1}{\sqrt{x(x-a_ic)}}$ for $i \in \{0,\dots,k\}$ and $f_{\sf w}(x):=\frac{1}{\sqrt{x(c-x)}}$. Then we have
 \begin{eqnarray}\label{eq:ConvolutionPattern4a}
  \int_x^1dt\frac{\HA_{\sf w_3}(\frac{x}{t})\HA_{\sf b_0,\dots,b_k}(t)}{(x-ct)\sqrt{t(t-a_0)}} &=& \frac{\pi}{\sqrt{x(x-a_0c)}}\left(\sum_{i=0}^k\HA_{\sf c_0,\dots,c_{i-1},d_i,a_i,\dots,a_k}(x)-\right.\\*
  &&\left.-\HA_{\sf c_0,\dots,c_k,w}(x)-\sum_{i=0}^k\HA_{\sf c_0,\dots,c_{i-1},d_i,0,a_{i+1},\dots,a_k}(x)\right)\N\\
  \int_x^1dt\frac{t\HA_{\sf w_3}(\frac{x}{t})\HA_{\sf b_0,\dots,b_k}(t)}{(x-ct)\sqrt{t(t-a_0)}} &=& \frac{x}{c}\int_x^1dt\frac{\HA_{\sf w_3}(\frac{x}{t})\HA_{\sf b_0,\dots,b_k}(t)}{(x-ct)\sqrt{t(t-a_0)}}-\frac{\pi}{c}\HA_{\sf 0,a_0,\dots,a_k}(x)\label{eq:ConvolutionPattern4b}
 \end{eqnarray}
 and
 \begin{eqnarray}
  \sum_{i=1}^N\frac{(\frac{4}{c})^i}{i^2\binom{2i}{i}}\cdot\Mvec\left[\frac{\HA_{\sf b_0,\dots,b_k}(x)}{\sqrt{x(x-a_0)}}\right](N) &=& \frac{\pi}{c^N}\Mvec\left[\frac{1}{\sqrt{x(x-a_0c)}}\left(\sum_{i=0}^k\HA_{\sf c_0,\dots,c_{i-1},d_i,a_i,\dots,a_k}(x)-\right.\right.\N\\*
  &&\left.\left.-\HA_{\sf c_0,\dots,c_k,w}(x)-\sum_{i=0}^k\HA_{\sf c_0,\dots,c_{i-1},d_i,0,a_{i+1},\dots,a_k}(x)\right)\right](N)\N\\*
  &&+\sum_{i=1}^\infty\frac{(\frac{4}{c})^i}{i^2\binom{2i}{i}}\cdot\Mvec\left[\frac{\HA_{\sf b_0,\dots,b_k}(x)}{\sqrt{x(x-a_0)}}\right](N)\label{eq:Pattern4a}.
 \end{eqnarray}
\end{theorem}
\begin{proof}
 Equation \eqref{eq:ConvolutionPattern4b} immediately follows from \eqref{eq:Convolution41ii} and \eqref{eq:ConvolutionPattern2a}. We prove \eqref{eq:ConvolutionPattern4a} by induction on $k$. If $k=0$, then applying Lemma~\ref{lem:Convolution41i} and \eqref{eq:ConvolutionPattern2a} yields
 \begin{eqnarray*}
  \int_x^1dt\frac{\HA_{\sf w_3}(\frac{x}{t})\HA_{\sf b_0}(t)}{(x-ct)\sqrt{t(t-a_0)}} &=& \frac{1}{\sqrt{x(x-a_0c)}} \int_x^1 dt \frac{1}{\sqrt{t(t-a_0c)}} \bigg(c\int_t^1 du \frac{u\HA_{\sf w_3}(\frac{t}{u})}{(t-cu)\sqrt{u(1-u)}}\\*
  &&+\pi\HA_{\sf a_0}(t)\bigg),
 \end{eqnarray*}
 from which, by \eqref{eq:Convolution41iiBase}, we obtain the right hand side of \eqref{eq:ConvolutionPattern4a}
 \[
  \frac{\pi}{\sqrt{x(x-a_0c)}} \left(-\HA_{\sf d_0,0}(x)-\HA_{\sf c_0,w}(x)+\HA_{\sf d_0,a_0}(x)\right).
 \]
 If $k>0$, then by Lemma~\ref{lem:Convolution41i} and \eqref{eq:ConvolutionPattern2a} we rewrite the left hand side of \eqref{eq:ConvolutionPattern4a} as
 \[
  \frac{1}{\sqrt{x(x-a_0c)}} \int_x^1 dt \frac{1}{\sqrt{t(t-a_0c)}} \left(c\int_t^1 du \frac{u\HA_{\sf w_3}(\frac{t}{u})\HA_{\sf b_1,\dots,b_k}(u)}{(t-cu)\sqrt{u(u-a_1)}}
  +\pi\HA_{\sf a_0,\dots,a_k}(t)\right),
 \]
 which in turn, by \eqref{eq:ConvolutionPattern4b}, equals
 \[
  \frac{1}{\sqrt{x(x-a_0c)}} \int_x^1 dt \frac{1}{\sqrt{t(t-a_0c)}} \left(t\int_t^1 du \frac{\HA_{\sf w_3}(\frac{t}{u})\HA_{\sf b_1,\dots,b_k}(u)}{(t-cu)\sqrt{u(u-a_1)}}-\pi\HA_{\sf 0,a_1,\dots,a_k}(t)
  +\pi\HA_{\sf a_0,\dots,a_k}(t)\right).
 \]
 Altogether, using the induction hypothesis we obtain
 \begin{eqnarray*}
  \int_x^1dt\frac{\HA_{\sf w_3}(\frac{x}{t})\HA_{\sf b_0,\dots,b_k}(t)}{(x-ct)\sqrt{t(t-a_0)}} &=& \frac{\pi}{\sqrt{x(x-a_0c)}} \int_x^1 dt \frac{1}{\sqrt{t(t-a_0c)}} \Bigg(\sqrt{\frac{t}{t-a_1c}}\cdot\\*
  &&\Bigg(\sum_{i=1}^k\HA_{\sf c_1,\dots,c_{i-1},d_i,a_i,\dots,a_k}(t)-\HA_{\sf c_1,\dots,c_k,w}(t)-\\*
  &&-\sum_{i=1}^k\HA_{\sf c_1,\dots,c_{i-1},d_i,0,a_{i+1},\dots,a_k}(t)\Bigg)-\HA_{\sf 0,a_1,\dots,a_k}(t)+\HA_{\sf a_0,\dots,a_k}(t)\Bigg)\\
  &=& \frac{\pi}{\sqrt{x(x-a_0c)}}\Bigg(\sum_{i=1}^k\HA_{\sf c_0,\dots,c_{i-1},d_i,a_i,\dots,a_k}(x)-\HA_{\sf c_0,\dots,c_k,w}(x)-\\*
  &&-\sum_{i=1}^k\HA_{\sf c_0,\dots,c_{i-1},d_i,0,a_{i+1},\dots,a_k}(x)-\HA_{\sf d_0,0,a_1,\dots,a_k}(x)+\HA_{\sf d_0,a_0,\dots,a_k}(x)\Bigg).
 \end{eqnarray*}
 This concludes the proof of \eqref{eq:ConvolutionPattern4a}. Finally, from \eqref{eq:RepresentInverse} and \eqref{eq:RepresentInverseBinomial} by \eqref{eq:MellinMult} and \eqref{eq:MellinSummation} we obtain
 \[
  \sum_{i=1}^N\frac{(\frac{4}{c})^i}{i^2\binom{2i}{i}} = c^{-N}\Mvec\left[\frac{\HA_{\sf w_3}(x)}{x-c}\right](N)+\sum_{i=1}^\infty\frac{(\frac{4}{c})^i}{i^2\binom{2i}{i}}.
 \]
%pos
 Hence by \eqref{eq:MellinMult} we have
 \begin{eqnarray*}
  \sum_{i=1}^N\frac{(\frac{4}{c})^i}{i^2\binom{2i}{i}}\cdot\Mvec\left[\frac{\HA_{\sf b_0,\dots,b_k}(x)}{\sqrt{x(x-a_0)}}\right](N) &=& c^{-N}\Mvec\left[\int_x^1dt\frac{\HA_{\sf w_3}(\frac{x}{t})\HA_{\sf b_0,\dots,b_k}(t)}{(x-ct)\sqrt{t(t-a_0)}}\right](N)\\*
  && +\sum_{i=1}^\infty\frac{(\frac{4}{c})^i}{i^2\binom{2i}{i}}\Mvec\left[\frac{\HA_{\sf b_0,\dots,b_k}(x)}{\sqrt{x(x-a_0)}}\right](N),
 \end{eqnarray*}
 which implies \eqref{eq:Pattern4a} by virtue of \eqref{eq:ConvolutionPattern4a}.
\end{proof}
Note that we can write the infinite sum above in terms of $\HA$ again.
\begin{equation}
 \sum_{i=1}^\infty\frac{(\frac{4}{c})^i}{i^2\binom{2i}{i}} = \frac{\mathrm{H}_{\sf w_1}(\frac{1}{c})^2}{2} = \frac{\left(\pi-\HA_{\sf w_1}(\frac{1}{c})\right)^2}{2}
\end{equation}

The following theorem allows to set up Mellin representations of expressions of the form
\begin{equation}\label{eq:SumPattern3}
 \frac{1}{N\binom{2N}{N}}S_{k_1,\dots,k_m}(x_1,\dots,x_m)(N).
\end{equation}

\begin{theorem}\label{thm:Pattern3}
 Let $a_0,\dots,a_k<0$ and with $e_0,e_1 \in \{0,1\}$ define $f_{\sf b_i^{e_0,e_1}}(x):=\frac{1}{x\sqrt{(x-a_i)^{e_0}(x-a_{i+1})^{e_1}}}$ for $i \in \{0,\dots,k-1\}$ and $f_{\sf b_k^{e_0}}(x):=\frac{1}{x\sqrt{(1-x)(x-a_k)^{e_0}}}$. Then we have
 \begin{equation}\label{eq:ConvolutionPattern3}
  \int_x^1dt\frac{\sqrt{t}\HA_{\sf a_1,\dots,a_k}(t)}{(t-a_0)\sqrt{t-x}} = \sum_{e_0,\dots,e_k=0}^1\frac{\prod_{i=0}^ka_i^{e_i}}{\sqrt{(x-a_0)^{e_0}}}\HA_{\sf b_0^{e_0,e_1},b_1^{e_1,e_2},\dots,b_k^{e_k}}(x)
 \end{equation}
 and
 \begin{equation}\label{eq:Pattern3}
  \frac{1}{N\binom{2N}{N}}\Mvec\left[\frac{\HA_{\sf a_1,\dots,a_k}(x)}{x-a_0}\right](N) = \frac{1}{4^N}\Mvec\left[\frac{1}{x}\sum_{e_0,\dots,e_k=0}^1\frac{\prod_{i=0}^ka_i^{e_i}}{\sqrt{(x-a_0)^{e_0}}}\HA_{\sf b_0^{e_0,e_1},b_1^{e_1,e_2},\dots,b_k^{e_k}}(x)\right](N).
 \end{equation}
\end{theorem}
\begin{proof}
 We prove \eqref{eq:ConvolutionPattern3} by induction on $k$. For $k=0$ we easily obtain
 \[
  \int_x^1dt\frac{\sqrt{t}}{(t-a_0)\sqrt{t-x}} = \int_x^1 dt \frac{1}{t\sqrt{1-t}}+\frac{a_0}{\sqrt{x-a_0}}\int_x^1 dt \frac{1}{t\sqrt{(1-t)(t-a_0)}} = \HA_{\sf b_0^0}(x)+a_0\frac{\HA_{\sf b_0^1}(x)}{\sqrt{x-a_0}}
 \]
 using \eqref{eq:Convolution3}. For $k>0$ we assume \eqref{eq:ConvolutionPattern3} holds for all values smaller than $k$. Now \eqref{eq:Convolution3} and the induction hypothesis yield
 \begin{eqnarray*}
  \int_x^1dt\frac{\sqrt{t}\HA_{\sf a_1,\dots,a_k}(t)}{(t-a_0)\sqrt{t-x}} &=& \int_x^1 dt \frac{1}{t} \int_t^1 du \frac{\sqrt{u}\HA_{\sf a_2,\dots,a_k}(u)}{(u-a_1)\sqrt{u-t}}\\*
  &&+\frac{a_0}{\sqrt{x-a_0}}\int_x^1 dt \frac{1}{t\sqrt{t-a_0}} \int_t^1 du \frac{\sqrt{u}\HA_{\sf a_2,\dots,a_k}(u)}{(u-a_1)\sqrt{u-t}}\\
  &=&\sum_{e_1,\dots,e_k=0}^1\left(\prod_{i=1}^ka_i^{e_i}\right)\HA_{\sf b_0^{0,e_1},b_1^{e_1,e_2},\dots,b_k^{e_k}}(x)\\*
  &&+\frac{a_0}{\sqrt{x-a_0}}\sum_{e_1,\dots,e_k=0}^1\left(\prod_{i=1}^ka_i^{e_i}\right)\HA_{\sf b_0^{1,e_1},b_1^{e_1,e_2},\dots,b_k^{e_k}}(x),
 \end{eqnarray*}
 which completes the proof of \eqref{eq:ConvolutionPattern3}.
 Finally, by \eqref{eq:RepresentInverseBinomial} and \eqref{eq:MellinMult} we obtain
 \[
  \frac{1}{N\binom{2N}{N}}\Mvec\left[\frac{\HA_{\sf a_1,\dots,a_k}(x)}{x-a_0}\right](N) = \frac{1}{4^N}\Mvec\left[\frac{1}{x}\int_x^1dt\frac{\sqrt{t}\HA_{\sf a_1,\dots,a_k}(t)}{(t-a_0)\sqrt{t-x}}\right](N),
 \]
 from which we infer \eqref{eq:Pattern3} by virtue of \eqref{eq:ConvolutionPattern3}.
\end{proof}

The next theorem deals with formulae arising in expressions of the form
\begin{equation}\label{eq:SumPattern5}
 \binom{2N}{N}\sum_{i=1}^N\frac{x_0^i}{i\binom{2i}{i}}S_{k_1,\dots,k_m}(x_1,\dots,x_m)(i).
\end{equation}

\begin{theorem}\label{thm:Pattern5}
 Let $a_0,\dots,a_k,c<0$ and with $e_0,e_1 \in \{0,1\}$ define $f_{\sf b_i^{e_0,e_1}}(x):=\frac{1}{x\sqrt{(x-a_i)^{e_0}(x-a_{i+1})^{e_1}}}$ for $i \in \{0,\dots,k-1\}$ as well as $f_{\sf b_k^{e_0}}(x):=\frac{1}{x\sqrt{(1-x)(x-a_k)^{e_0}}}$ and $f_{\sf w^{e_0}}(x):=\frac{1}{(x-a_0)^{e_0}\sqrt{x(x-c)}}$. Then with $e_1,\dots,e_k \in \{0,1\}$ we have
 \begin{equation}\label{eq:ConvolutionPattern5}
  \int_x^1dt\frac{\HA_{\sf b_1^{e_1,e_2},\dots,b_k^{e_k}}(t)}{t\sqrt{(t-a_1)^{e_1}(t-x)}} = \frac{\pi}{\sqrt{x}}\sum_{i_1=0}^{e_1}\dots\sum_{i_k=0}^{e_k}\left(\prod_{j=1}^k\frac{(1-e_j)a_j+2i_j-e_j}{a_j}\right)\HA_{\sf i_1a_1,\dots,i_ka_k}(x)
 \end{equation}
 and
 \begin{eqnarray}\label{eq:Pattern5}
  \binom{2N}{N}\Mvec\left[\sum_{e_0,\dots,e_k=0}^1\frac{\left(\prod_{i=0}^ka_i^{e_i}\right)\HA_{\sf b_0^{e_0,e_1},b_1^{e_1,e_2},\dots,b_k^{e_k}}(x)}{(x-c)\sqrt{(x-a_0)^{e_0}}}\right](N) &=& 4^N\Mvec\Bigg[\frac{\HA_{\sf w_0^0,a_1,\dots,a_k}(x)}{\sqrt{x(x-c)}}\\*
  &&+a_0\frac{\HA_{\sf w_0^1,a_1,\dots,a_k}(x)}{\sqrt{x(x-c)}}\Bigg](N).\N
 \end{eqnarray}
\end{theorem}
\begin{proof}
 For the left hand side of \eqref{eq:Pattern5}, using \eqref{eq:RepresentBinomial} and \eqref{eq:MellinMult}, we obtain
 \[
  \frac{4^N}{\pi}\Mvec\left[\sum_{e_0,\dots,e_k=0}^1\frac{\prod_{i=0}^ka_i^{e_i}}{\sqrt{x}}\int_x^1dt \frac{\HA_{\sf b_0^{e_0,e_1},b_1^{e_1,e_2},\dots,b_k^{e_k}}(t)}{(t-c)\sqrt{(t-a_0)^{e_0}(t-x)}}\right](N).
 \]
 Applying \eqref{eq:Convolution5} to the integral yields
 \[
  \frac{4^N}{\pi}\Mvec\left[\sum_{e_0,\dots,e_k=0}^1\frac{\prod_{i=0}^ka_i^{e_i}}{\sqrt{x(x-c)}}\int_x^1dt \frac{1}{(t-a_0)^{e_0}\sqrt{t-c}}\int_t^1du \frac{\sqrt{(u-a_0)^{e_0}}f_{\sf b_0^{e_0,e_1}}(u)\HA_{\sf b_1^{e_1,e_2},\dots,b_k^{e_k}}(u)}{\sqrt{u-t}}\right](N),
 \]
 which equals
 \[
  4^N\Mvec\left[\sum_{e_0,\dots,e_k=0}^1\frac{\prod_{i=0}^ka_i^{e_i}}{\sqrt{x(x-c)}}\sum_{i_1=0}^{e_1}\dots\sum_{i_k=0}^{e_k}\left(\prod_{j=1}^k\frac{(1-e_j)a_j+2i_j-e_j}{a_j}\right)\HA_{\sf w_0^{e_0},i_1a_1,\dots,i_ka_k}(x)\right](N)
 \]
 by virtue of \eqref{eq:ConvolutionPattern5}. Noting that $\sum_{e=0}^1a^e\sum_{i=0}^e\frac{(1-e)a+2i-e}{a}f(i)=f(1)$, we can simplify this expression to obtain the right hand side of \eqref{eq:Pattern5}. Finally, we prove \eqref{eq:ConvolutionPattern5} by induction on $k$. For $k=1$ we obtain
 \[
  \int_x^1dt\frac{\HA_{\sf b_1^{e_1}}(t)}{t\sqrt{(t-a_1)^{e_1}(t-x)}} = \frac{1}{\sqrt{x}}\int_x^1dt\frac{1}{(t-a_1)^{e_1}\sqrt{t}}\int_t^1du\frac{1}{u\sqrt{(1-u)(u-t)}}
 \]
 from \eqref{eq:Convolution5} and we compute
 \[
  \int_x^1dt\frac{1}{t\sqrt{(1-t)(t-x)}} = \frac{\pi}{\sqrt{x}}.
 \]
 Altogether, for the two cases $e_1=0$ and $e_1=1$ this yields
 \[
  \int_x^1dt\frac{\HA_{\sf b_1^0}(t)}{t\sqrt{t-x}} = \frac{\pi}{\sqrt{x}}\HA_{\sf 0}(x) \qquad\text{and}\qquad \int_x^1dt\frac{\HA_{\sf b_1^1}(t)}{t\sqrt{(t-a_1)(t-x)}} = \frac{\pi}{a_1\sqrt{x}}\left(\HA_{\sf a_1}(x)-\HA_{\sf 0}(x)\right),
 \]
 which can be written uniformly as
 \[
  \int_x^1dt\frac{\HA_{\sf b_1^{e_1}}(t)}{t\sqrt{(t-a_1)^{e_1}(t-x)}} = \frac{\pi}{\sqrt{x}}\sum_{i_1=0}^{e_1}\frac{(1-e_1)a_1+2i_1-e_1}{a_1}\HA_{\sf i_1a_1}(x)
 \]
 based on $\frac{1}{(t-a)^et}=\sum_{i=0}^e\frac{(1-e)a+2i-e}{a(t-ia)}$, where $e \in \{0,1\}$. Similarly, for $k>1$ applying \eqref{eq:Convolution5} and the induction hypothesis results in
 \[
  \int_x^1dt\frac{\HA_{\sf b_1^{e_1,e_2},\dots,b_k^{e_k}}(t)}{t\sqrt{(t-a_1)^{e_1}(t-x)}} = \frac{\pi}{\sqrt{x}}\sum_{i_2=0}^{e_2}\dots\sum_{i_k=0}^{e_k}\left(\prod_{j=2}^k\frac{(1-e_j)a_j+2i_j-e_j}{a_j}\right)\int_x^1dt\frac{\HA_{\sf i_2a_2,\dots,i_ka_k}(t)}{(t-a_1)^{e_1}t},
 \]
 which can be identified with \eqref{eq:ConvolutionPattern5}.
\end{proof}

%----------------------------------------------------------------------------

\begin{example}
 In order to illustrate the use of the previous theorems we will now, step by
step, set up an integral representation for the binomial sum
 \[ 
\sum_{i=1}^N\frac{(-1)^i}{(2i+1)\binom{2i}{i}}\sum_{j=1}^i\binom{2j}{j}\frac{
S_2(j)}{j}.
 \]
 Starting from an integral representation of $S_2(N)$, see \eqref{eq:RepresentS2}, we immediately obtain
 \[
  \frac{S_2(N)}{N} = -\Mvec\left[\frac{\HA_{\sf 1,0}(x)}{x}\right](N)+\zeta_2\Mvec\left[\frac{1}{x}\right](N)
 \]
 by \eqref{eq:RepresentInverse} and \eqref{eq:MellinConv}. Multiplying this by $\binom{2N}{N}$ we arrive at two terms matching the structure of the left hand side of \eqref{eq:Pattern1}. Although the condition $a_0,\dots,a_k<0$ of Theorem~\ref{thm:Pattern1} is not satisfied, we still can derive the relevant formulae from it. For $k=0$ the limit $a_0\to0$ yields
 \[
  \binom{2N}{N}\Mvec\left[\frac{1}{x}\right](N)=\frac{4^N}{\pi}\Mvec\left[\frac{\HA_{\sf w_1}(x)}{x}\right](N).
 \]
 Similarly, sending $a_0\to0$ and $a_2\to0$, for $k=2$ we obtain
 \[
  \binom{2N}{N}\Mvec\left[\frac{1}{x}\int_x^1dt\frac{\HA_{\sf 0}(t)}{t-a_1}\right](N) = \frac{4^N}{\pi}\Mvec\left[\frac{1}{x}\int_x^1dt\int_t^1du\frac{\HA_{\sf w_1}(u)}{\sqrt{t(t-a_1)}\sqrt{u(u-a_1)}}\right](N)
 \]
 and analytic continuation of both integrands, avoiding $a_1\ge0$, in the limit $a_1\to1$ gives $-f_{\sf 1}(t)\HA_{\sf 0}(t)$ and $-f_{\sf w_1}(t)f_{\sf w_1}(u)\HA_{\sf w_1}(u)$, respectively. Altogether, we have
 \[
  \binom{2N}{N}\frac{S_2(N)}{N} = \frac{4^N}{\pi}\Mvec\left[\frac{-\HA_{\sf w_1,w_1,w_1}(x)+\zeta_2\HA_{\sf w_1}(x)}{x}\right](N)
 \]
 and by \eqref{eq:MellinSummation} we immediately arrive at
 \[
  \sum_{i=1}^N\binom{2i}{i}\frac{S_2(i)}{i} = \frac{4^N}{\pi}\Mvec\left[\frac{-\HA_{\sf w_1,w_1,w_1}(x)+\zeta_2\HA_{\sf w_1}(x)}{x-\frac{1}{4}}\right](N)+\frac{\HA_{\sf \frac{1}{4},w_1,w_1,w_1}(0)-\zeta_2\HA_{\sf \frac{1}{4},w_1}(0)}{\pi},
 \]
 see also \eqref{eq:RepresentSimpleBinomialSum}. We compute the constants $\HA_{\sf \frac{1}{4},w_1,w_1,w_1}(0)=-\frac{2\pi}{3}\zeta_3$ and $\HA_{\sf \frac{1}{4},w_1}(0)=0$ as Cauchy principal values. Next, we multiply by $\frac{1}{(2N+1)\binom{2N}{N}}$ and use \eqref{eq:RepresentInverseBinomial2} to obtain
 \begin{eqnarray*}
  \frac{1}{(2N+1)\binom{2N}{N}}\sum_{i=1}^N\binom{2i}{i}\frac{S_2(i)}{i} &=& \frac{4^N}{\pi}\frac{1}{(2N+1)\binom{2N}{N}}\Mvec\left[\frac{-\HA_{\sf w_1,w_1,w_1}(x)+\zeta_2\HA_{\sf w_1}(x)}{x-\frac{1}{4}}\right](N)\\*
  &&-4^{-N}\frac{\zeta_3}{3}\Mvec\left[\frac{1}{\sqrt{1-x}}\right](N).
 \end{eqnarray*}
 The first term matches the structure of the left hand side of \eqref{eq:Pattern2}, but again the relevant condition $a_0,\dots,a_k,c<0$ of Theorem~\ref{thm:Pattern2} is not satisfied. As above, the relevant formulae can be derived by analytic continuation. Taking the limits $a_0\to0$, $a_1\to1$, and $a_2\to0$ as described above we obtain
 \[
  \frac{1}{(2N+1)\binom{2N}{N}}\Mvec\left[\frac{-\HA_{\sf w_1,w_1,w_1}(x)+\zeta_2\HA_{\sf w_1}(x)}{x-c}\right](N) = \frac{\pi}{2{\cdot}4^N}\Mvec\left[\frac{1}{\sqrt{x-c}}\int_x^1dt\frac{-\HA_{\sf 1,0}(t)+\zeta_2}{t\sqrt{t-c}}\right](N).
 \]
 In view of \eqref{eq:modifiedMellin} we do the analytic continuation to $c=\frac{1}{4}$ on both sides of the real line, avoiding $c\ge0$. In order to arrive at
 \[
  \frac{1}{(2N+1)\binom{2N}{N}}\Mvec\left[\frac{-\HA_{\sf w_1,w_1,w_1}(x)+\zeta_2\HA_{\sf w_1}(x)}{x-\frac{1}{4}}\right](N) = \frac{\pi}{2{\cdot}4^N}\Mvec\left[\frac{1}{\sqrt{x-\frac{1}{4}}}\int_x^1dt\frac{-\HA_{\sf 1,0}(t)+\zeta_2}{t\sqrt{t-\frac{1}{4}}}\right](N)
 \]
 we also rely on the following identities for $x\in(0,1)$.
 \begin{eqnarray*}
  \lim_{c\to\frac{1}{4}\pm i0}\frac{-\HA_{\sf w_1,w_1,w_1}(x)+\zeta_2\HA_{\sf w_1}(x)}{x-c} &=& \lim_{y\to x\mp i0}\frac{-\HA_{\sf w_1,w_1,w_1}(y)+\zeta_2\HA_{\sf w_1}(y)}{y-\frac{1}{4}}\\
  \lim_{c\to\frac{1}{4}\pm i0}\frac{1}{\sqrt{x-c}}\int_x^1dt\frac{-\HA_{\sf 1,0}(t)+\zeta_2}{t\sqrt{t-c}} 
&=& \lim_{y\to x\mp i0}\frac{1}{\sqrt{y-\frac{1}{4}}}\int_y^1dt\frac{-\HA_{\sf 1,0}(t)+\zeta_2}{t\sqrt{t-\frac{1}{4}}}.
 \end{eqnarray*}
 Altogether, we have
 \[
  \frac{1}{(2N+1)\binom{2N}{N}}\sum_{i=1}^N\binom{2i}{i}\frac{S_2(i)}{i} = \frac{1}{2}\Mvec\left[\frac{-\HA_{\sf w_8,1,0}(x)+\zeta_2\HA_{\sf w_8}(x)}{\sqrt{x-\frac{1}{4}}}\right](N)-4^{-N}\frac{\zeta_3}{3}\Mvec\left[\frac{1}{\sqrt{1-x}}\right](N).
 \]
 Finally we apply \eqref{eq:MellinSummation} once more in order to obtain
 \begin{multline*}
  \sum_{i=1}^N\frac{(-1)^i}{(2i+1)\binom{2i}{i}}\sum_{j=1}^i\binom{2j}{j}\frac{S_2(j)}{j} = \frac{1}{2}(-1)^N\Mvec\left[\frac{x(-\HA_{\sf w_8,1,0}(x)+\zeta_2\HA_{\sf w_8}(x))}{(x+1)\sqrt{x-\frac{1}{4}}}\right](N)\\*
  -\frac{1}{2}\Mvec\left[\frac{x(-\HA_{\sf w_8,1,0}(x)+\zeta_2\HA_{\sf w_8}(x))}{(x+1)\sqrt{x-\frac{1}{4}}}\right](0)
  -\frac{\zeta_3}{3}(-4)^{-N}\Mvec\left[\frac{x}{(x+4)\sqrt{1-x}}\right](N)\\*
  +\frac{\zeta_3}{3}\Mvec\left[\frac{x}{(x+4)\sqrt{1-x}}\right](0),
 \end{multline*}
 cf.~\eqref{eq:RepresentDoubleBinomialSum}. We may also express the constants as
 \begin{eqnarray*}
  \Mvec\left[\frac{x}{(x+4)\sqrt{1-x}}\right](0) &=& 2+\frac{8}{\sqrt{5}}\ln\left(\frac{\sqrt{5}-1}{2}\right)\\
  \Mvec\left[\frac{x(-\HA_{\sf w_8,1,0}(x)+\zeta_2\HA_{\sf w_8}(x))}{(x+1)\sqrt{x-\frac{1}{4}}}\right](0)&=& \frac{4}{3}\zeta_3\left(1-\frac{8}{\sqrt{5}}\ln\left(\frac{\sqrt{5}-1}{2}\right)\right)\\*
  &&+\frac{2}{\sqrt{5}}\int_0^1dx\arccos\left(\frac{3-2x}{2+2x}\right)\frac{\HA_{\sf 1,0}(x)-\zeta_2}{x\sqrt{x-\frac{1}{4}}}.
 \end{eqnarray*}
\end{example}
%----------------------------------------------------------------------------
%----------------------------------------------------------------------------
\section{Mellin Representations of Nested Finite Binomial Sums}
\label{sec:5}
%----------------------------------------------------------------------------

\vspace{1mm}
\noindent
In the following we will present a larger class of binomial and inverse binomial sums
weighted by (generalized) harmonic sums up to depth {\sf d = 2} and present their
Mellin representations based on iterated integrals over the letters of the alphabet given in 
Section~\ref{sec:3}.

Starting from \eqref{eq:RepresentBinomial} and \eqref{eq:RepresentInverseBinomial} it is immediate from 
\eqref{eq:MellinMult} and \eqref{eq:RepresentInverse} to determine the Mellin representation of binomial 
coefficients with an arbitrary power of $N$ in the denominator:
%----------------------------------------------------------------------------
\begin{eqnarray}
\frac{1}{N^k} \binom{2N}{N} &=& \frac{4^N}{\pi}\Mvec\left[\frac{1}{x} \HA_{\sf  
\underbrace{\mbox{\scriptsize 0, \ldots ,0}}_{\mbox{\scriptsize
$k-1$}}, w_1}(x) \right](N), \quad k \in \mathbb{N}, k \geq 1,
\\
\frac{1}{\displaystyle N^k \binom{2N}{N}} &=& \frac{1}{4^N}\Mvec\left[\frac{1}{x}\HA_{\sf \underbrace{\mbox{\scriptsize 0, \ldots ,0}}_{\mbox{\scriptsize
$k-1$}}, w_3}(x)\right](N), \quad k \in \mathbb{N}, k \geq 2.
\end{eqnarray}
%----------------------------------------------------------------------------
Based on the properties of the binomial coefficient we can write $\frac{1}{(2N+1)\binom{2N}{N}}$ and 
$\frac{1}{(N+1)\binom{2N}{N}}$ in terms of shifts of $\frac{1}{N\binom{2N}{N}}$ and $\frac{1}{N^2\binom{2N}{N}}$. 
Thereby we obtain the following integral representations~:
%----------------------------------------------------------------------------
\begin{eqnarray}
%-------
% \Gamma(2N) &=& \frac{4^N}{2\sqrt{\pi}} \Gamma(N) \Gamma\left(N + \frac{1}{2}\right)\\
% %-------
% \binom{2N}{N} &=& 
% \frac{4^N}{\pi} \int_0^1dx x^N \frac{1}{\sqrt{x(1-x)}}\\
% %-------
% \frac{1}{\displaystyle \binom{2N}{N}} &=& \frac{1}{4^N}\left(N + \frac{1}{2}\right) 
% \int_0^1 dx 
% \frac{x^N}{\sqrt{1-x}}\\
% %-------
% \frac{1}{\displaystyle N \binom{2N}{N}} &=& \frac{1}{4^N} \int_0^1 dx
% \frac{x^N}{x~\sqrt{1-x}}\\
%-------
\frac{1}{\displaystyle (2N+1) \binom{2N}{N}} &=& \frac{1}{2~4^N}\int_0^1 dx x^N \frac{1}{\sqrt{1-x}}\label{eq:RepresentInverseBinomial2}
\\
%-------
\frac{1}{\displaystyle (N+1) \binom{2N}{N}} &=& \frac{1}{4^N} \int_0^1 dx x^N \left(\frac{1}{\sqrt{1-x}}
-\frac{1}{2}\HA_{\sf w_3}(x)\right)~.
\end{eqnarray}
%-----------------------------------------------------------------------------
For the harmonic \cite{Vermaseren:1998uu,Blumlein:1998if} and S-sums \cite{Moch:2001zr,Ablinger:2013cf} appearing in the 
subsequent representations we explicitly give the integral representations in terms of $\HA$, cf. also \cite{Ablinger14}.
One obtains \cite{Vermaseren:1998uu,Blumlein:1998if,Ablinger:2012qm,Ablinger:2013cf}
%-----------------------------------------------------------------------------
\begin{eqnarray}
S_1(N) &=& \int_0^1dx \frac{x^N-1}{x-1}\\
%-------
S_2(N) &=& \int_0^1dx x^N\frac{\HA_{\sf 0}(x)}{x-1}+\zeta_2\label{eq:RepresentS2}\\
%-------
S_{-2}(N) &=& (-1)^N\int_0^1dx x^N\frac{\HA_{\sf 0}(x)}{x+1}-\frac{\zeta_2}{2}\\
%-------
S_{1,1}(N) &=& -\int_0^1dx \frac{x^N-1}{x-1}\HA_{\sf 1}(x)\\
%-------
S_{1,2}\left(\frac{1}{2},1;N\right) &=& \frac{1}{2^N}\int_0^1dx x^N\frac{\zeta_2-\HA_{\sf 1,0}(x)}{x-2}+\frac{5}{8}\zeta_3\\
%-------
S_{1,2}\left(\frac{1}{2},-1;N\right) &=& \left(-\frac{1}{2}\right)^N\int_0^1dx x^N\frac{\HA_{\sf -1,0}(x)}{x+2}-\frac{\zeta_2}{2^{N+1}}\int_0^1dx \frac{x^N}{x-2}-\frac{13}{24}\zeta_3\\
%-------
S_{1,2}\left(-\frac{1}{2},1;N\right) &=& \left(-\frac{1}{2}\right)^N\int_0^1dx x^N\frac{\zeta_2-\HA_{\sf 1,0}(x)}{x+2}
\N\\*&&
+\frac{1}{3}\ln^3\left(\tfrac{3}{2}\right)+\ln\left(\tfrac{3}{2}\right)\Li_2\left(-\tfrac{1}{2}\right)
-\Li_3\left(-\tfrac{1}{2}\right)-2\Li_3\left(\tfrac{1}{3}\right)~.
%-------
% \frac{1}{\displaystyle N^2 \binom{2N}{N}} &=& - \frac{1}{4^N} \int_0^1 dx
% x^N \frac{1}{x} 
% \ln\left(\frac{1-\sqrt{1-x}}{1+\sqrt{1-x}}\right)\\
% %-------
% \sum_{i = 1}^N \binom{2i}{i} (-2)^i &=& -\frac{2}{3} + \frac{(-8)^N}{\pi}
% \int_0^1 dx x^N \frac{8x}{1+8x} \frac{1}{\sqrt{x(1-x)}}\\
% %-------
% \sum_{i = 1}^N \frac{1}{\displaystyle i \binom{2i}{i}} \sum_{j=1}^i \binom{2j}{j} 
% (-2)^j &=& 
% -\frac{\pi}{9 \sqrt{3}} - \ln(3)
% - \frac{2}{3} \frac{1}{4^N} \int_0^1 dx x^N 
% \frac{1}{(x-4) \sqrt{1-x}}
% \nonumber\\ &&
% - (-2)^{N+1} \int_0^1 dx x^N \frac{1}{2x+1} \left[1 - \frac{1}{3\sqrt{1+8x}}\right]
% \\
% %-------
% \sum_{i = 1}^N \frac{1}{\displaystyle i^2 \binom{2i}{i}} \sum_{j=1}^i \binom{2j}{j} 
% (-2)^j &=& - \frac{2}{3} \zeta_2 - \ln^2(2) - \frac{2}{9} 
% \Li_2\left(-\frac{1}{8}\right) 
% \nonumber\\ &&
% + \frac{2}{3} \frac{1}{4^N} \int_0^1 dx x^N 
% \frac{1}{x-4} \ln\left(\frac{1-\sqrt{1-x}}{1+\sqrt{1-x}}\right)\nonumber\\
% && -\frac{4}{3} (-2)^N \int_0^1 dx x^N \frac{1}{2x+1} \left[\ln(x) 
% +\ln\left(1+\sqrt{1+8x}\right) - 2 \ln(2)\right]
% \nonumber\\
\end{eqnarray}
%-----------------------------------------------------------------------------
Here $\Li_n(x)$ denotes the polylogarithm \cite{LEWIN:81}
%-----------------------------------------------------------------------------
\begin{eqnarray}
\Li_n(x) = {\rm H}_{\sf
\underbrace{\mbox{\scriptsize 0, \ldots ,0}}_{\mbox{\scriptsize
$n-1$}}, 1}(x)~. 
\end{eqnarray}
%-----------------------------------------------------------------------------
In the following we will list the integral representations of different type
of (inverse) binomial sums. We start with some very 
simple examples, which do not involve harmonic or S-sums.
%-----------------------------------------------------------------------------
\begin{eqnarray}
%-------
\sum_{i=1}^N \frac{1}{\displaystyle i \binom{2i}{i}} &=& \int_0^1 
dx \frac{(\frac{x}{4})^N-1}{x-4}\frac{1}{\sqrt{1-x}}
\\
%-------
\sum_{i=1}^N \frac{1}{\displaystyle i^2 \binom{2i}{i}} &=& \int_0^1 
dx \frac{(\frac{x}{4})^N-1}{x-4}\HA_{\sf w_3}(x)
\\
%-------
\sum_{i = 1}^N \binom{2i}{i} (-1)^i &=& \frac{1}{\pi} \int_0^1 dx \frac{(-4x)^N-1}{x+\frac{1}{4}} \sqrt{\frac{x}{1-x}}\\
%-------
\sum_{i = 1}^N \frac{1}{i}\binom{2i}{i} (-1)^i &=& \frac{1}{\pi} \int_0^1 dx \frac{(-4x)^N-1}{x+\frac{1}{4}}\HA_{\sf w_1}(x)
\\
%-------
% \sum_{i=1}^N \frac{4^i}{\displaystyle \binom{2i}{i} i^3} &=&
%  \int_0^1 dx 
%  \frac{x^N-1}{x-1}\Biggl\{
%  \Li_2\left[\frac{1}{2}\left(1-\sqrt{1-x}\right)\right]
% -\Li_2\left[\frac{1}{2}\left(1+\sqrt{1-x}\right)\right]
% \nonumber\\ &&
% + \frac{1}{2}\ln\left(\frac{x}{4}\right)\left[
%     \ln\left(1-\sqrt{1-x}\right) 
%    -\ln\left(1+\sqrt{1-x}\right)\right] \Biggr\}
% \\
%-------
\sum_{i = 1}^N \frac{1}{i^2}\binom{2i}{i} (-1)^i &=& \frac{1}{\pi} \int_0^1 dx \frac{(-4x)^N-1}{x+\frac{1}{4}}\HA_{\sf 0,w_1}(x)
\\
%-------
\sum_{i = 1}^N \frac{1}{i^3}\binom{2i}{i} (-1)^i &=& \frac{1}{\pi} \int_0^1 dx \frac{(-4x)^N-1}{x+\frac{1}{4}}\HA_{\sf 0,0,w_1}(x)
\\
%-------
\sum_{i = 1}^N \binom{2i}{i} (-2)^i &=& \frac{1}{\pi}
\int_0^1 dx \frac{(-8x)^N-1}{x+\frac{1}{8}}\sqrt{\frac{x}{1-x}}\\
%-------
\sum_{i = 1}^N \frac{1}{\displaystyle i \binom{2i}{i}} \sum_{j=1}^i \binom{2j}{j} 
(-2)^j &=&
\int_0^1 dx \frac{(-2x)^N-1}{x+\frac{1}{2}}\left(1-\frac{1}{6\sqrt{2}\sqrt{x+\frac{1}{8}}}\right)
\N\\*&&
-\frac{2}{3}\int_0^1 dx \frac{(\frac{x}{4})^N-1}{x-4}\frac{1}{\sqrt{1-x}}
\\
%-------
\sum_{i = 1}^N \frac{1}{\displaystyle i^2 \binom{2i}{i}} \sum_{j=1}^i \binom{2j}{j} 
(-2)^j &=& 
-\int_0^1 dx \frac{(-2x)^N-1}{x+\frac{1}{2}}\left(\ln(x)+\frac{\HA_{\sf w_{28}}(x)}{6\sqrt{2}}\right)
\N\\*&&
-\frac{2}{3}\int_0^1 dx \frac{(\frac{x}{4})^N-1}{x-4}\HA_{\sf w_3}(x)~.
\end{eqnarray}
%-----------------------------------------------------------------------------
The following single sums have summands that fit the pattern \eqref{eq:SumPattern1}, so Theorem~\ref{thm:Pattern1} applies.
%-----------------------------------------------------------------------------
\begin{eqnarray}
\sum_{i = 1}^N \frac{1}{i^2}\binom{2i}{i}  S_1(i) &=& -\frac{1}{\pi}\int_0^1 dx \frac{4^Nx^N-1}{x-\frac{1}{4}}(\HA_{\sf 0,w_1,1}(x)+2\ln(2)\HA_{\sf 0,w_1}(x))
\\
%-------
\sum_{i = 1}^N \binom{2i}{i}  S_2(i) &=& -\frac{1}{\pi}\int_0^1 dx \frac{4^N x^N-1}{x-\frac{1}{4}}\sqrt{\frac{x}{1-x}}\left(\tfrac{1}{2}\HA_{\sf w_1}(x)^2-\zeta_2\right)
\\
%-------
\sum_{i = 1}^N \frac{1}{i}\binom{2i}{i}  S_{2}(i) &=& -\frac{1}{\pi}\int_0^1 dx 
\frac{4^Nx^N-1}{x-\frac{1}{4}}\left(\tfrac{1}{6}
\HA_{\sf w_1}(x)^3-\zeta_2\HA_{\sf w_1}(x)\right)\label{eq:RepresentSimpleBinomialSum}
\\
%-------
\sum_{i = 1}^N \frac{(-1)^i}{i}\binom{2i}{i}  S_{2}(i) &=& -\frac{1}{\pi}\int_0^1 dx 
\frac{(-4)^Nx^N-1}{x+\frac{1}{4}}\left[\tfrac{1}{6}\HA_{\sf w_1}(x)^3-\zeta_2\HA_{\sf 
w_1}(x)\right]
\\
%-------
\sum_{i = 1}^N \frac{1}{i^2}\binom{2i}{i}  S_2(i) &=& -\frac{1}{\pi}\int_0^1 dx \frac{4^N x^N-1}{x-\frac{1}{4}}\left(\HA_{\sf 0,w_1,w_1,w_1}(x)-\zeta_2\HA_{\sf 0,w_1}(x)\right)
\\
%-------
\sum_{i = 1}^N \binom{2i}{i}  S_{-2}(i) &=& \frac{1}{\pi}\Biggl[\int_0^1 dx \frac{(-4)^Nx^N-1}{x+\frac{1}{4}}\sqrt{\frac{x}{1+x}}\HA_{\sf w_2,w_1}(x)
\N\\*&&
- \frac{\zeta_2}{2}\int_0^1 dx \frac{4^Nx^N-1}{x-\frac{1}{4}}\sqrt{\frac{x}{1-x}}\Biggr]
\\
%-------
\sum_{i = 1}^N \frac{1}{i}\binom{2i}{i}  S_{-2}(i) &=& 
\frac{1}{\pi}\Biggl[\int_0^1 dx \frac{(-4)^Nx^N-1}{x+\frac{1}{4}}\HA_{\sf 
w_2,w_2,w_1}(x) 
\nonumber\\* &&
- \frac{\zeta_2}{2}\int_0^1 dx \frac{4^Nx^N-1}{x-\frac{1}{4}}\HA_{\sf w_1}(x)\Biggr]
\\
%-------
\sum_{i = 1}^N \frac{(-1)^i}{i}\binom{2i}{i}  S_{-2}(i) &=& 
\frac{1}{\pi}\Biggl[\int_0^1 dx \frac{4^Nx^N-1}{x-\frac{1}{4}}\HA_{\sf w_2,w_2,w_1}(x) 
\nonumber\\* &&
- \frac{\zeta_2}{2}\int_0^1 dx \frac{(-4)^Nx^N-1}{x+\frac{1}{4}}\HA_{\sf w_1}(x)\Biggr]
\\
%-------
\sum_{i = 1}^N \frac{1}{i^2}\binom{2i}{i}  S_{-2}(i) &=& \frac{1}{\pi}\Biggl[\int_0^1 dx \frac{(-4)^Nx^N-1}{x+\frac{1}{4}}\HA_{\sf 0,w_2,w_2,w_1}(x)
\N\\*&&
- \frac{\zeta_2}{2}\int_0^1 dx \frac{4^Nx^N-1}{x-\frac{1}{4}}\HA_{\sf 0,w_1}(x)\Biggr]
\\
%-------
\sum_{i = 1}^N \frac{1}{i}\binom{2i}{i}  S_{1,1}(i) 
&=& \frac{1}{\pi}\int_0^1 dx \frac{4^Nx^N-1}{x-\frac{1}{4}}\Biggl[\HA_{\sf w_1,1,1}(x)
+2\ln(2)\HA_{\sf w_1,1}(x)
\nonumber\\* &&
+[2\ln^2(2)-\zeta_2]\HA_{\sf w_1}(x)\Biggr]~.
\end{eqnarray}
%-------------------------------------------------------------------------------------------------------------------------
We also include some examples with S-sums that are not harmonic sums, Theorem~\ref{thm:Pattern1} applies here as well.
%-------------------------------------------------------------------------------------------------------------------------
\begin{eqnarray}
\lefteqn{\sum_{i = 1}^N \binom{2i}{i}  (-2)^i S_{1,2}\left(\frac{1}{2},1;i\right)
\ =\ 
\frac{1}{\pi} \Biggl\{ \int_0^1 dx \frac{(-4)^N x^N -1}{x + \tfrac{1}{4}}\sqrt{\frac{x}{2-x}}\left[
\HA_{\sf w_6,w_1,w_1}(x) - \zeta_2 \HA_{\sf w_6}(x)\right]
}\N\\*&&
+ \frac{5}{8} \zeta_3 \int_0^1 dx \frac{(-8)^N x^N-1}{x+ \frac{1}{8}} \sqrt{\frac{x}{1-x}}\Biggr\}
\label{eq:RepresentSimpleBinomialSum1}\\
%-------
\lefteqn{\sum_{i = 1}^N \binom{2i}{i}  (-2)^i S_{1,2}\left(\frac{1}{2},-1;i\right) 
\ =\ 
\frac{1}{\pi}\Biggl[\int_0^1 dx \frac{4^Nx^N-1}{x-\frac{1}{4}}\sqrt{\frac{x}{2+x}}\HA_{\sf w_5,w_2,w_1}(x)
}\N\\*&&
+\frac{\zeta_2}{2}\int_0^1 dx \frac{(-4)^Nx^N-1}{x+\frac{1}{4}}\sqrt{\frac{x}{2-x}}\HA_{\sf w_6}(x)
%\N\\&&
-\frac{13}{24}\zeta_3\int_0^1 dx \frac{(-8)^Nx^N-1}{x+\frac{1}{8}}\sqrt{\frac{x}{1-x}}\Biggr]
\quad\quad\quad\quad\label{eq:RepresentSimpleBinomialSum2}\\
%-------
\lefteqn{\sum_{i = 1}^N \binom{2i}{i}  (-2)^i S_{1,2}\left(-\frac{1}{2},1;i\right) 
\ =\ 
-\frac{1}{\pi}\Biggl\{\int_0^1 dx \frac{4^N x^N-1}{x-\frac{1}{4}}\sqrt{\frac{x}{2+x}}
\left[\HA_{\sf w_7,w_1,w_1}(x) - \zeta_2 \HA_{\sf w_7}(x)\right]
}\N\\*&&
+\left[2\Li_3(\tfrac{1}{3})+\Li_3(-\tfrac{1}{2})-\ln(\tfrac{3}{2})\Li_2(-\tfrac{1}{2})-\tfrac{1}{3}\ln(\tfrac{3}{2})^3\right]
%\N\\&&*
\int_0^1 dx \frac{(-8)^N x^N -1}{x +\tfrac{1}{8}} \sqrt{\frac{x}{1-x}}\Biggr\}.
\end{eqnarray}
%-------------------------------------------------------------------------------------------------------------------------
Next we list the two inverse binomial sums which occurred in Ref.~\cite{Ablinger:2014uka}. 
Their summands obey the pattern \eqref{eq:SumPattern3}.
%-------------------------------------------------------------------------------------------------------------------------
\begin{eqnarray}
% \frac{4^N}{\displaystyle N^3 \binom{2N}{N}} &=& \Mvec\left[\frac{1}{x}\HA_{\sf 0,w_3}(x)\right](N)
% \\
%-------
\sum_{i = 1}^N \frac{4^i}{i^3}\frac{1}{\displaystyle \binom{2i}{i}} &=& \int_0^1 dx 
\frac{x^N-1}{x-1} \HA_{\sf 0,w_3}(x)
\\
%-------
\sum_{i = 1}^N \frac{4^i}{i^2}\frac{1}{\displaystyle \binom{2i}{i}} S_1(i) &=& 
\int_0^1 dx \frac{x^N-1}{x-1}\left[\HA_{\sf 0,w_3}(x)-\HA_{\sf w_3,0}(x)
-\HA_{\sf w_3,1}(x)-2\ln(2)\HA_{\sf w_3}(x)\right].\quad
\end{eqnarray}
%----------------------------------------------------------------------------
The following more involved variants of the above also fit the pattern \eqref{eq:SumPattern3}.
%----------------------------------------------------------------------------
\begin{eqnarray}
\lefteqn{\sum_{j=1}^N \frac{1}{\displaystyle j \binom{2j}{j}}
S_{1,2}\left(\frac{1}{2},1\right)(j)
\ =\ 
\int_0^1 dx \frac{(\frac{x}{8})^N-1}{x-8} \Bigg[\left(\zeta_2\HA_{\sf w_3}(x)+\HA_{\sf 0,0,w_3}(x)-\HA_{\sf w_3,w_3,w_3}(x)\right)
}\N\\*&&
-2\frac{\zeta_2\HA_{\sf w_{11}}(x)+\HA_{\sf w_{10},0,w_3}(x)-\HA_{\sf w_{11},w_3,w_3}(x)}{\sqrt{2-x}}\Bigg]
%\N\\&&
+\frac{5}{8}\zeta_3\int_0^1 dx \frac{(\frac{x}{4})^N-1}{x-4} \frac{1}{\sqrt{1-x}}
\\
%--------------------------------
\lefteqn{\sum_{j=1}^N \frac{1}{\displaystyle j^2 \binom{2j}{j}}
S_{1,2}\left(\frac{1}{2},1\right)(j)
\ =\ 
\int_0^1 dx \frac{(\frac{x}{8})^N-1}{x-8} \Big[\zeta_2\HA_{\sf 0,w_3}(x)+\HA_{\sf 0,0,0,w_3}(x)-\HA_{\sf 0,w_3,w_3,w_3}(x)
}\N\\*&&
-2\left(\zeta_2\HA_{\sf w_{10},w_{11}}(x)+\HA_{\sf w_{10},w_{10},0,w_3}(x)-\HA_{\sf w_{10},w_{11},w_3,w_3}(x)\right)\Big]
%\N\\&&
+\frac{5}{8}\zeta_3\int_0^1 dx \frac{(\frac{x}{4})^N-1}{x-4} \HA_{\sf w_3}(x)
\quad\quad\\
%--------------------------------
\lefteqn{\sum_{j=1}^N \frac{1}{\displaystyle j \binom{2j}{j}}
S_{1,2}\left(\frac{1}{2},-1\right)(j) 
\ =\ 
\int_0^1 dx \frac{(-\frac{x}{8})^N-1}{x+8}\Bigg[2\frac{\HA_{\sf w_{15},w_4,w_3}(x)-\HA_{\sf w_{16},0,w_3}(x)}{\sqrt{2+x}}
}\N\\*&&
+\HA_{\sf 0,0,w_3}(x)-\HA_{\sf w_4,w_4,w_3}(x)\Bigg]+\frac{\zeta_2}{2}\int_0^1 dx \frac{(\frac{x}{8})^N-1}{x-8}\Bigg[2\frac{\HA_{\sf w_{11}}(x)}{\sqrt{2-x}}-\HA_{\sf w_3}(x)\Bigg]
\N\\*&&
-\frac{13}{24}\zeta_3\int_0^1 dx \frac{(\frac{x}{4})^N-1}{x-4}\frac{1}{\sqrt{1-x}}
\\
%--------------------------------
\lefteqn{\sum_{j=1}^N \frac{1}{\displaystyle j^2 \binom{2j}{j}}
S_{1,2}\left(\frac{1}{2},-1\right)(j) 
\ =\ 
\int_0^1 dx \frac{(-\frac{x}{8})^N-1}{x+8}\Big[2\HA_{\sf w_{16},w_{15},w_4,w_3}(x)-2\HA_{\sf w_{16},w_{16},0,w_3}(x)
}\N\\*&&
+\HA_{\sf 0,0,0,w_3}(x)-\HA_{\sf 0,w_4,w_4,w_3}(x)\Big]+\frac{\zeta_2}{2}\int_0^1 dx \frac{(\frac{x}{8})^N-1}{x-8}\Big[2\HA_{\sf w_{10},w_{11}}(x)-\HA_{\sf 0,w_3}(x)\Big]
\N\\*&&
-\frac{13}{24}\zeta_3\int_0^1 dx \frac{(\frac{x}{4})^N-1}{x-4}\HA_{\sf w_3}(x).
\end{eqnarray}
%----------------------------------------------------------------------------
From here on we list double sums for which the inner sums already have been listed before. We start with a multitude of 
instances to which Theorem~\ref{thm:Pattern2} applies.
%----------------------------------------------------------------------------
\begin{eqnarray}
%--------------------------------
\lefteqn{\sum_{i=1}^N \frac{1}{\displaystyle (2i+1) \binom{2i}{i}} \sum_{j=1}^i 
\binom{2j}{j} \frac{(-1)^j}{j^3} 
\ =\ 
\frac{1}{2}\int_0^1 dx \frac{(-x)^N-1}{x+1}\frac{x}{\sqrt{x+\frac{1}{4}}}\HA_{\sf w_{14},0,0}(x)
}\N\\*&&
-\frac{\HA_{\sf -\frac{1}{4},0,0,w_1}(0)}{2\pi}\int_0^1 dx \frac{(\frac{x}{4})^N-1}{x-4}\frac{x}{\sqrt{1-x}}
\\
%--------------------------------
\lefteqn{\sum_{i=1}^N \frac{1}{\displaystyle (i+1) \binom{2i}{i}} \sum_{j=1}^i 
\binom{2j}{j} \frac{(-1)^j}{j^3} 
\ =\ 
\int_0^1 dx \frac{(-x)^N-1}{x+1}\left(\frac{x}{\sqrt{x+\frac{1}{4}}}\HA_{\sf w_{14},0,0}(x)
-\frac{x}{2}
\HA_{\sf w_{14},w_{14},0,0}(x)\right)
}\N\\*&&
-\frac{\HA_{\sf -\frac{1}{4},0,0,w_1}(0)}{\pi}\int_0^1 dx \frac{(\frac{x}{4})^N-1}{x-4}\left(\frac{x}{\sqrt{1-x}}-\frac{x}{2}\HA_{\sf w_3}(x)\right)
\\
%--------------------------------
\lefteqn{\sum_{i=1}^N \frac{(-2)^i}{\displaystyle (2i+1) \binom{2i}{i}} \sum_{j=1}^i 
\binom{2j}{j} \frac{(-1)^j}{j^3} 
\ =\ 
\frac{1}{2}\int_0^1 dx \frac{(2x)^N-1}{x-\frac{1}{2}}\frac{x}{\sqrt{x+\frac{1}{4}}}\HA_{\sf w_{14},0,0}(x)
}\N\\*&&
-\frac{\HA_{\sf -\frac{1}{4},0,0,w_1}(0)}{2\pi}\int_0^1 dx \frac{(-\frac{x}{2})^N-1}{x+2}\frac{x}{\sqrt{1-x}}
\\
%--------------------------------
\lefteqn{\sum_{i=1}^N \frac{(-2)^i}{\displaystyle (i+1) \binom{2i}{i}} \sum_{j=1}^i 
\binom{2j}{j} \frac{(-1)^j}{j^3} 
\ =\ 
\int_0^1 dx \frac{(2x)^N-1}{x-\frac{1}{2}}\left(\frac{x}{\sqrt{x+\frac{1}{4}}}\HA_{\sf w_{14},0,0}(x)-\frac{x}{2}\HA_{\sf w_{14},w_{14},0,0}(x)\right)
}\N\\*&&
-\frac{\HA_{\sf -\frac{1}{4},0,0,w_1}(0)}{\pi}\int_0^1 dx \frac{(-\frac{x}{2})^N-1}{x+2}\left(\frac{x}{\sqrt{1-x}}-\frac{x}{2}\HA_{\sf w_3}(x)\right)
\\
%--------------------------------
\lefteqn{\sum_{i=1}^N \frac{1}{\displaystyle (2i+1) \binom{2i}{i}} \sum_{j=1}^i 
\binom{2j}{j} \frac{1}{j} 
S_{1}(j)
\ =\ 
-\frac{1}{2}\int_0^1 dx \frac{x^N-1}{x-1}\frac{x}{\sqrt{x-\frac{1}{4}}}\HA_{\sf w_8,1}(x)
}\N\\*&&
-\frac{\zeta_2}{3}\int_0^1 dx \frac{(\frac{x}{4})^N-1}{x-4}\frac{x}{\sqrt{1-x}}
\\
%--------------------------------
\lefteqn{\sum_{i=1}^N \frac{1}{\displaystyle (i+1) \binom{2i}{i}} \sum_{j=1}^i 
\binom{2j}{j} \frac{1}{j} 
S_{1}(j)
\ =\ 
\int_0^1 dx \frac{x^N-1}{x-1}\left(\frac{x}{2}\HA_{\sf w_8,w_8,1}(x)-\frac{x}{\sqrt{x-\frac{1}{4}}}\HA_{\sf w_8,1}(x)\right)
}\N\\*&&
+\frac{2}{3}\zeta_2\int_0^1 dx \frac{(\frac{x}{4})^N-1}{x-4}\left(\frac{x}{2}\HA_{\sf w_3}(x)-\frac{x}{\sqrt{1-x}}\right)
\\
%--------------------------------
\lefteqn{\sum_{i=1}^N \frac{1}{\displaystyle (2i+1) \binom{2i}{i}} \sum_{j=1}^i \binom{2j}{j} \frac{1}{j^2} S_{1}(j)
\ =\ 
-\frac{1}{2}\int_0^1 dx \frac{x^N-1}{x-1}\frac{x}{\sqrt{x-\frac{1}{4}}}\HA_{\sf w_8,0,1}(x)
}\N\\*&&
+\left(\frac{\HA_{\sf \frac{1}{4},0,w_1,1}(0)}{2\pi}-\ln(2)\zeta_2\right)\int_0^1 dx \frac{(\frac{x}{4})^N-1}{x-4}\frac{x}{\sqrt{1-x}}
\\
%--------------------------------
\lefteqn{\sum_{i=1}^N \frac{1}{\displaystyle (i+1) \binom{2i}{i}} \sum_{j=1}^i \binom{2j}{j} \frac{1}{j^2} S_{1}(j)
\ =\ 
\int_0^1 dx \frac{x^N-1}{x-1}\left(\frac{x}{2}\HA_{\sf w_8,w_8,0,1}(x)-\frac{x}{\sqrt{x-\frac{1}{4}}}\HA_{\sf w_8,0,1}(x)\right)
}\N\\*&&
+\left(\frac{\HA_{\sf \frac{1}{4},0,w_1,1}(0)}{\pi}-2\ln(2)\zeta_2\right)
%\\&&
\int_0^1 dx \frac{(\frac{x}{4})^N-1}{x-4}\left(\frac{x}{\sqrt{1-x}}-\frac{x}{2}\HA_{\sf w_3}(x)\right)
\\
%--------------------------------
\lefteqn{\sum_{i=1}^N \frac{1}{\displaystyle (2i+1) \binom{2i}{i}} \sum_{j=1}^i 
\binom{2j}{j} \frac{1}{j} 
S_{2}(j)
\ =\ 
\frac{1}{2}\int_0^1 dx \frac{x^N-1}{x-1}\frac{x}{\sqrt{x-\frac{1}{4}}}\left(\zeta_2\HA_{\sf w_8}(x)-\HA_{\sf w_8,1,0}(x)\right)
}\N\\*&&
-\frac{\zeta_3}{3}\int_0^1 dx \frac{(\frac{x}{4})^N-1}{x-4}\frac{x}{\sqrt{1-x}}
\\
%--------------------------------
\lefteqn{\sum_{i=1}^N \frac{1}{\displaystyle (i+1) \binom{2i}{i}} \sum_{j=1}^i 
\binom{2j}{j} \frac{1}{j} 
S_{2}(j)
\ =\ 
\int_0^1 dx \frac{x^N-1}{x-1}\Bigl[\frac{x}{2}\left(\HA_{\sf w_8,w_8,1,0}(x)-\zeta_2\HA_{\sf w_8,w_8}(x)\right)
}\N\\*&&
-\frac{x}{\sqrt{x-\frac{1}{4}}}\left(\HA_{\sf w_8,1,0}(x)-\zeta_2\HA_{\sf w_8}(x)\right)\Bigr]
%\\&&
+\frac{2}{3}\zeta_3\int_0^1 dx \frac{(\frac{x}{4})^N-1}{x-4}\left(\frac{x}{2}\HA_{\sf w_3}(x)-\frac{x}{\sqrt{1-x}}\right)
\\
%--------------------------------
\lefteqn{\sum_{i=1}^N \frac{(-1)^i}{\displaystyle (2i+1) \binom{2i}{i}} \sum_{j=1}^i 
\binom{2j}{j} \frac{1}{j} 
S_{2}(j)
\ =\ 
\frac{1}{2}\int_0^1 dx \frac{(-x)^N-1}{x+1}\frac{x}{\sqrt{x-\frac{1}{4}}}\left(\zeta_2\HA_{\sf w_8}(x)-\HA_{\sf w_8,1,0}(x)\right)
}\N\\*&&
-\frac{\zeta_3}{3}\int_0^1 dx \frac{(-\frac{x}{4})^N-1}{x+4}\frac{x}{\sqrt{1-x}}\label{eq:RepresentDoubleBinomialSum}
\\
%--------------------------------
\lefteqn{\sum_{i=1}^N \frac{(-1)^i}{\displaystyle (i+1) \binom{2i}{i}} \sum_{j=1}^i 
\binom{2j}{j} \frac{1}{j} 
S_{2}(j)
\ =\ 
\int_0^1 dx \frac{(-x)^N-1}{x+1}\Bigl[\frac{x}{2}\left(\HA_{\sf w_8,w_8,1,0}(x)-\zeta_2\HA_{\sf w_8,w_8}(x)\right)
}\N\\*&&
-\frac{x}{\sqrt{x-\frac{1}{4}}}\left(\HA_{\sf w_8,1,0}(x)-\zeta_2\HA_{\sf w_8}(x)\right)\Bigr]
%\\&&
+\frac{2}{3}\zeta_3\int_0^1 dx \frac{(-\frac{x}{4})^N-1}{x+4}\left(\frac{x}{2}\HA_{\sf w_3}(x)-\frac{x}{\sqrt{1-x}}\right)
\quad\quad\quad\quad\\
%--------------------------------
\lefteqn{\sum_{i=1}^N \frac{1}{\displaystyle (2i+1) \binom{2i}{i}} \sum_{j=1}^i \binom{2j}{j} \frac{(-1)^j}{j} S_{2}(j)
\ =\ 
\frac{1}{2}\int_0^1 dx \frac{(-x)^N-1}{x+1}\frac{x}{\sqrt{x+\frac{1}{4}}}\left(\zeta_2\HA_{\sf w_{14}}(x)-\HA_{\sf w_{14},1,0}(x)\right)
}\N\\*&&
+\left(\frac{2}{3}\ln\left(\tfrac{\sqrt{5}-1}{2}\right)^3-\frac{4}{5}\ln\left(\tfrac{\sqrt{5}-1}{2}\right)\zeta_2-\frac{4}{5}\zeta_3\right)
%\N\\&&
\int_0^1 dx \frac{(\frac{x}{4})^N-1}{x-4}\frac{x}{\sqrt{1-x}}
\\
%--------------------------------
\lefteqn{\sum_{i=1}^N \frac{1}{\displaystyle (i+1) \binom{2i}{i}} \sum_{j=1}^i \binom{2j}{j} \frac{(-1)^j}{j} S_{2}(j)
\ =\ 
\int_0^1 dx \frac{(-x)^N-1}{x+1}\Biggl[\frac{x}{2}\left(\HA_{\sf w_{14},w_{14},1,0}(x)-\zeta_2\HA_{\sf w_{14},w_{14}}(x)\right)
}\N\\*&&
-\frac{x}{\sqrt{x+\frac{1}{4}}}\left(\HA_{\sf w_{14},1,0}(x)-\zeta_2\HA_{\sf w_{14}}(x)\right)\Biggr]
\N\\*&&
+\left(\frac{4}{3}\ln\left(\tfrac{\sqrt{5}-1}{2}\right)^3-\frac{8}{5}\ln\left(\tfrac{\sqrt{5}-1}{2}\right)\zeta_2-\frac{8}{5}\zeta_3\right)
%\N\\&&
\int_0^1 dx \frac{(\frac{x}{4})^N-1}{x-4}\left(\frac{x}{\sqrt{1-x}}-\frac{x}{2}\HA_{\sf w_3}(x)\right)
\\
%--------------------------------
\lefteqn{\sum_{i=1}^N \frac{1}{\displaystyle (2i+1) \binom{2i}{i}} \sum_{j=1}^i 
\binom{2j}{j} \frac{1}{j} 
S_{-2}(j)
\ =\ 
\frac{1}{2}\int_0^1 dx \frac{(-x)^N-1}{x+1}\frac{x}{\sqrt{x+\frac{1}{4}}}\HA_{\sf w_{14},-1,0}(x)
}\N\\*&&
-\frac{\zeta_2}{4}\int_0^1 dx \frac{x^N-1}{x-1}\frac{x}{\sqrt{x-\frac{1}{4}}}\HA_{\sf w_8}(x)
%\\&&
-\frac{\HA_{\sf -\frac{1}{4},w_2,w_2,w_1}(0)}{2\pi}\int_0^1 dx \frac{(\frac{x}{4})^N-1}{x-4}\frac{x}{\sqrt{1-x}}
\\
%--------------------------------
\lefteqn{\sum_{i=1}^N \frac{1}{\displaystyle (i+1) \binom{2i}{i}} \sum_{j=1}^i 
\binom{2j}{j} \frac{1}{j} 
S_{-2}(j)
\ =\ 
\int_0^1 dx \frac{(-x)^N-1}{x+1}\Biggl(\frac{x}{\sqrt{x+\frac{1}{4}}}\HA_{\sf w_{14},-1,0}(x)
}\N\\*&&
-\frac{x}{2}\HA_{\sf w_{14},w_{14},-1,0}(x)\Biggr)
-\frac{\zeta_2}{2}\int_0^1 dx \frac{x^N-1}{x-1}\left(\frac{x}{\sqrt{x-\frac{1}{4}}}\HA_{\sf w_8}(x)-\frac{x}{2}\HA_{\sf w_8,w_8}(x)\right)
\N\\*&&
-\frac{\HA_{\sf -\frac{1}{4},w_2,w_2,w_1}(0)}{\pi}\int_0^1 dx \frac{(\frac{x}{4})^N-1}{x-4}\left(\frac{x}{\sqrt{1-x}}-\frac{x}{2}\HA_{\sf w_3}(x)\right)
\\
%--------------------------------
\lefteqn{\sum_{i=1}^N \frac{(-2)^i}{\displaystyle (2i+1) \binom{2i}{i}} \sum_{j=1}^i 
\binom{2j}{j} \frac{1}{j} 
S_{-2}(j)
\ =\ 
\frac{1}{2}\int_0^1 dx \frac{(2x)^N-1}{x-\frac{1}{2}}\frac{x}{\sqrt{x+\frac{1}{4}}}\HA_{\sf w_{14},-1,0}(x)
}\N\\*&&
-\frac{\zeta_2}{4}\int_0^1 dx \frac{(-2x)^N-1}{x+\frac{1}{2}}\frac{x}{\sqrt{x-\frac{1}{4}}}\HA_{\sf w_8}(x)
%\\&&
-\frac{\HA_{\sf -\frac{1}{4},w_2,w_2,w_1}(0)}{2\pi}\int_0^1 dx \frac{(-\frac{x}{2})^N-1}{x+2}\frac{x}{\sqrt{1-x}}
\\
%--------------------------------
\lefteqn{\sum_{i=1}^N \frac{(-2)^i}{\displaystyle (i+1) \binom{2i}{i}} \sum_{j=1}^i 
\binom{2j}{j} \frac{1}{j} 
S_{-2}(j)
\ =\ 
\int_0^1 dx \frac{(2x)^N-1}{x-\frac{1}{2}}\Biggl(\frac{x}{\sqrt{x+\frac{1}{4}}}\HA_{\sf w_{14},-1,0}(x)
}\N\\*&&
-\frac{x}{2}\HA_{\sf w_{14},w_{14},-1,0}(x)\Biggr)
-\frac{\zeta_2}{2}\int_0^1 dx \frac{(-2x)^N-1}{x+\frac{1}{2}}\left(\frac{x}{\sqrt{x-\frac{1}{4}}}\HA_{\sf w_8}(x)-\frac{x}{2}\HA_{\sf w_8,w_8}(x)\right)
\N\\*&&
-\frac{\HA_{\sf -\frac{1}{4},w_2,w_2,w_1}(0)}{\pi}\int_0^1 dx \frac{(-\frac{x}{2})^N-1}{x+2}\left(\frac{x}{\sqrt{1-x}}
-\frac{x}{2}\HA_{\sf w_3}(x)\right)
\\
%--------------------------------
\lefteqn{\sum_{i=1}^N \frac{1}{\displaystyle (2i+1) \binom{2i}{i}} \sum_{j=1}^i 
\binom{2j}{j} \frac{(-1)^j}{j} 
S_{-2}(j)
\ =\ 
\frac{1}{2}\int_0^1 dx \frac{x^N-1}{x-1}\frac{x}{\sqrt{x-\frac{1}{4}}}\HA_{\sf w_8,-1,0}(x)
}\N\\*&&
-\frac{\zeta_2}{4}\int_0^1 dx \frac{(-x)^N-1}{x+1}\frac{x}{\sqrt{x+\frac{1}{4}}}\HA_{\sf w_{14}}(x)
\N\\*&&
-\left(\frac{\zeta_2}{2}\ln\left(\tfrac{\sqrt{5}-1}{2}\right)+\frac{\HA_{\sf \frac{1}{4},w_2,w_2,w_1}(0)}{2\pi}\right)\int_0^1 dx \frac{(\frac{x}{4})^N-1}{x-4}\frac{x}{\sqrt{1-x}}
\\
%--------------------------------
\lefteqn{\sum_{i=1}^N \frac{1}{\displaystyle (i+1) \binom{2i}{i}} \sum_{j=1}^i 
\binom{2j}{j} \frac{(-1)^j}{j} 
S_{-2}(j)
\ =\ 
\int_0^1 dx \frac{x^N-1}{x-1}\Biggl(\frac{x}{\sqrt{x-\frac{1}{4}}}\HA_{\sf w_8,-1,0}(x)
}\N\\*&&
-\frac{x}{2}\HA_{\sf w_8,w_8,-1,0}(x)\Biggr)
-\frac{\zeta_2}{2}\int_0^1 dx \frac{(-x)^N-1}{x+1}\left(\frac{x}{\sqrt{x+\frac{1}{4}}}\HA_{\sf w_{14}}(x)-\frac{x}{2}\HA_{\sf w_{14},w_{14}}(x)\right)
\N\\*&&
-\left(\zeta_2\ln\left(\tfrac{\sqrt{5}-1}{2}\right)+\frac{\HA_{\sf \frac{1}{4},w_2,w_2,w_1}(0)}{\pi}\right)
%\\&&
\int_0^1 dx \frac{(\frac{x}{4})^N-1}{x-4}\left(\frac{x}{\sqrt{1-x}}-\frac{x}{2}\HA_{\sf w_3}(x)\right)
\\
%--------------------------------
\lefteqn{\sum_{i=1}^N \frac{1}{\displaystyle (2i+1) \binom{2i}{i}} \sum_{j=1}^i 
\binom{2j}{j} \frac{1}{j} 
S_{1,1}(j)
\ =\ 
\frac{1}{2}\int_0^1 dx \frac{x^N-1}{x-1}\frac{x}{\sqrt{x-\frac{1}{4}}}\HA_{\sf w_8,1,1}(x)
}\N\\*&&
-\left(\frac{\HA_{\sf \frac{1}{4},w_1,1,1}(0)}{2\pi}-\frac{2\ln(2)\zeta_2}{3}\right)\int_0^1 dx \frac{(\frac{x}{4})^N-1}{x-4}\frac{x}{\sqrt{1-x}}
\\
%--------------------------------
\lefteqn{\sum_{i=1}^N \frac{1}{\displaystyle (i+1) \binom{2i}{i}} \sum_{j=1}^i 
\binom{2j}{j} \frac{1}{j} 
S_{1,1}(j)
\ =\ 
\int_0^1 dx \frac{x^N-1}{x-1}\left(\frac{x}{\sqrt{x-\frac{1}{4}}}\HA_{\sf w_8,1,1}(x)-\frac{x}{2}\HA_{\sf w_8,w_8,1,1}(x)\right)
}\N\\*&&
-\left(\frac{\HA_{\sf \frac{1}{4},w_1,1,1}(0)}{\pi}-\frac{4\ln(2)\zeta_2}{3}\right)
%\\&&
\int_0^1 dx
\frac{(\frac{x}{4})^N-1}{x-4}\left(\frac{x}{\sqrt{1-x}}-\frac{x}{2}\HA_{\sf
w_3}(x)\right).
\end{eqnarray}
%--------------------------------
The summands of the next set of sums obey the pattern \eqref{eq:SumPattern5}. For shorter notation some of the constants involved in the integral representations are represented as infinite inverse binomial sums.
%--------------------------------
\begin{eqnarray}
\lefteqn{\sum_{i=1}^N \binom{2i}{i} (-2)^i \sum_{j=1}^i \frac{1}{\displaystyle j \binom{2j}{j}}
S_{1,2}\left(\frac{1}{2},1\right)(j) 
=
\int_0^1 dx \frac{(-x)^N-1}{x+1}\sqrt{\frac{x}{8-x}}\Bigl[\HA_{\sf w_{12},1,0}(x)-2\HA_{\sf w_{13},1,0}(x)
}\N\\*&&
-\zeta_2\left(\HA_{\sf w_{12}}(x)-2\HA_{\sf w_{13}}(x)\right)\Bigr]
%\N\\&&
-\frac{5\zeta_3}{8\sqrt{3}}\int_0^1 dx \frac{(-2x)^N-1}{x+\frac{1}{2}}\sqrt{\frac{x}{4-x}}
\N\\*&&
+c_1\int_0^1 dx \frac{(-8x)^N-1}{x+\frac{1}{8}}\sqrt{\frac{x}{1-x}}
\\
\lefteqn{c_1=\frac{1}{\pi}\sum_{j=1}^\infty \frac{1}{\displaystyle j \binom{2j}{j}}
S_{1,2}\left(\frac{1}{2},1\right)(j) \approx 0.10184720\dots
}\\
%--------------------------------
\lefteqn{\sum_{i=1}^N \binom{2i}{i} (-2)^i \sum_{j=1}^i \frac{1}{\displaystyle j^2 \binom{2j}{j}}
S_{1,2}\left(\frac{1}{2},1\right)(j) 
=
\int_0^1 dx \frac{(-x)^N-1}{x+1}\sqrt{\frac{x}{8-x}}(\zeta_2\HA_{\sf w_{12},2}(x)
}\N\\*&&
-\HA_{\sf w_{12},2,1,0}(x))
-\frac{5\zeta_3}{8}\int_0^1 dx \frac{(-2x)^N-1}{x+\frac{1}{2}}\sqrt{\frac{x}{4-x}}\HA_{\sf w_{19}}(x)
+c_2\int_0^1 dx \frac{(-8x)^N-1}{x+\frac{1}{8}}\sqrt{\frac{x}{1-x}}
\N\\
\\
\lefteqn{c_2=\frac{1}{\pi}\sum_{j=1}^\infty \frac{1}{\displaystyle j^2 \binom{2j}{j}}
S_{1,2}\left(\frac{1}{2},1\right)(j) \approx 0.08979755\dots
}\\
%--------------------------------
\lefteqn{\sum_{i=1}^N \binom{2i}{i} (-2)^i \sum_{j=1}^i \frac{1}{\displaystyle j \binom{2j}{j}}
S_{1,2}\left(\frac{1}{2},-1\right)(j) 
=
\int_0^1 dx \frac{x^N-1}{x-1}\sqrt{\frac{x}{8+x}}\left(\HA_{\sf w_{17},-1,0}(x) \right.
}\N\\*&&
\left. -2\HA_{\sf w_{18},-1,0}(x)\right)
+\frac{\zeta_2}{2}\int_0^1 dx \frac{(-x)^N-1}{x+1}\sqrt{\frac{x}{8-x}}\left(\HA_{\sf w_{12}}(x)-2\HA_{\sf w_{13}}(x)\right)
\N\\&&
+\frac{13}{24\sqrt{3}}\zeta_3\int_0^1 dx \frac{(-2x)^N-1}{x+\frac{1}{2}}\sqrt{\frac{x}{4-x}}
+c_3\int_0^1 dx \frac{(-8x)^N-1}{x+\frac{1}{8}}\sqrt{\frac{x}{1-x}}
\\
\lefteqn{c_3=\frac{1}{\pi}\sum_{j=1}^\infty \frac{1}{\displaystyle j \binom{2j}{j}}
S_{1,2}\left(\frac{1}{2},-1\right)(j) \approx -0.09960950\dots
}\\
%--------------------------------
\lefteqn{
\sum_{i=1}^N \binom{2i}{i} (-2)^i \sum_{j=1}^i \frac{1}{\displaystyle j^2 \binom{2j}{j}}
S_{1,2}\left(\frac{1}{2},-1\right)(j) 
=
\int_0^1 dx \frac{x^N-1}{x-1}\sqrt{\frac{x}{8+x}}\HA_{\sf w_{17},-2,-1,0}(x)
}\N\\*&&
-\frac{\zeta_2}{2}\int_0^1 dx \frac{(-x)^N-1}{x+1}\sqrt{\frac{x}{8-x}}\HA_{\sf w_{12},2}(x)
%\N\\&&
+\frac{13}{24}\zeta_3\int_0^1 dx \frac{(-2x)^N-1}{x+\frac{1}{2}}\sqrt{\frac{x}{4-x}}\HA_{\sf w_{19}}(x)
\N\\*&&
+c_4\int_0^1 dx \frac{(-8x)^N-1}{x+\frac{1}{8}}\sqrt{\frac{x}{1-x}}
\\
\lefteqn{c_4=\frac{1}{\pi}\sum_{j=1}^\infty \frac{1}{\displaystyle j^2 \binom{2j}{j}}
S_{1,2}\left(\frac{1}{2},-1\right)(j) \approx -0.08878871\dots
}
\end{eqnarray}
%--------------------------------
Finally, we also treat some sums which are not strictly nested. The summands obey the pattern \eqref{eq:SumPattern1} 
combined with \eqref{eq:SumPattern4} or \eqref{eq:SumPattern4a}. In order to save some space the right hand sides of the following identities still contain a binomial sum, for which we refer to the integral representations derived in Eq.~\eqref{eq:RepresentSimpleBinomialSum1} and \eqref{eq:RepresentSimpleBinomialSum2}, respectively.
%--------------------------------
\begin{eqnarray}
\lefteqn{\sum_{i=1}^N \binom{2i}{i} (-2)^i S_{1,2}\left(\frac{1}{2},1\right)(i) \sum_{j=1}^i \frac{1}{\displaystyle j \binom{2j}{j}}
\ =\ 
-\int_0^1 dx \frac{(-x)^N-1}{x+1}\sqrt{\frac{x}{8-x}}\bigg[\HA_{\sf w_{25},w_{26},0}(x)
}\N\\*&&
+2\HA_{\sf w_{13},1,0}(x)
+\frac{1}{\sqrt{3}}\HA_{\sf w_{25},w_{19},w_{19}}(x)-\zeta_2\bigg(2\HA_{\sf w_{13}}(x)
+\frac{\HA_{\sf w_{25}}(x)}{\sqrt{3}}\bigg)\bigg]
\N\\&&
-\frac{5}{8\sqrt{3}}\zeta_3 \int_0^1 dx \frac{(-2x)^N-1}{x+\frac{1}{2}}\sqrt{\frac{x}{4-x}}
+\frac{\pi}{3\sqrt{3}}\sum_{i=1}^N \binom{2i}{i} (-2)^i S_{1,2}\left(\frac{1}{2},1\right)(i)
\\
%--------------------------------
\lefteqn{\sum_{i=1}^N \binom{2i}{i} (-2)^i S_{1,2}\left(\frac{1}{2},1\right)(i) \sum_{j=1}^i \frac{1}{\displaystyle j^2 \binom{2j}{j}}
\ =\ 
-\int_0^1 dx \frac{(-x)^N-1}{x+1}\sqrt{\frac{x}{8-x}}\Big[\HA_{\sf w_{12},2,1,0}(x)
}\N\\*&&
+\HA_{\sf w_{12},0,1,0}(x)
+\HA_{\sf w_{25},w_{19},1,0}(x)+\HA_{\sf w_{25},w_{19},0,0}(x)+\HA_{\sf w_{25},w_{19},w_{19},w_{19}}(x)
%\N\\&&
-\zeta_2\big(\HA_{\sf w_{25},w_{19}}(x)+\HA_{\sf w_{12},2}(x)
\N\\*&&
+\HA_{\sf w_{12},0}(x)\big)\Big]
-\frac{5}{8}\zeta_3 \int_0^1 dx \frac{(-2x)^N-1}{x+\frac{1}{2}}\sqrt{\frac{x}{4-x}}\HA_{\sf w_{19}}(x)
%\N\\&&
+\frac{\zeta_2}{3}\sum_{i=1}^N \binom{2i}{i} (-2)^i S_{1,2}\left(\frac{1}{2},1\right)(i)
\N\\
\\
%--------------------------------
\lefteqn{\sum_{i=1}^N \binom{2i}{i} (-2)^i S_{1,2}\left(\frac{1}{2},-1\right)(i) \sum_{j=1}^i \frac{1}{\displaystyle j \binom{2j}{j}}
\ =\ 
-\int_0^1 dx \frac{x^N-1}{x-1}\sqrt{\frac{x}{8+x}}\Big[\frac{1}{\sqrt{3}}\HA_{\sf w_{21},w_{20},w_{19}}(x)
}\N\\*&&
+\HA_{\sf w_{21},w_{23},0}(x)+2\HA_{\sf w_{18},-1,0}(x)\Big]
%\N\\&&
-\frac{\zeta_2}{2} \int_0^1 dx \frac{(-x)^N-1}{x+1}\sqrt{\frac{x}{8-x}}\Bigg[2\HA_{\sf w_{13}}(x)+\frac{\HA_{\sf w_{25}}(x)}{\sqrt{3}}\Bigg]
\N\\*&&
+\frac{13}{24\sqrt{3}}\zeta_3 \int_0^1 dx \frac{(-2x)^N-1}{x+\frac{1}{2}}\sqrt{\frac{x}{4-x}}
%\\&&
+\frac{\pi}{3\sqrt{3}}\sum_{i=1}^N \binom{2i}{i} (-2)^i S_{1,2}\left(\frac{1}{2},-1\right)(i)
\\
%--------------------------------
\lefteqn{\sum_{i=1}^N \binom{2i}{i} (-2)^i S_{1,2}\left(\frac{1}{2},-1\right)(i) \sum_{j=1}^i \frac{1}{\displaystyle j^2 \binom{2j}{j}}
\ =\ 
\int_0^1 dx \frac{x^N-1}{x-1}\sqrt{\frac{x}{8+x}}\Big[\HA_{\sf w_{17},-2,-1,0}(x)
}\N\\*&&
-\HA_{\sf w_{17},0,-1,0}(x)
+\HA_{\sf w_{21},w_{20},-1,0}(x)-\HA_{\sf w_{21},w_{20},0,0}(x)-\HA_{\sf w_{21},w_{20},w_{19},w_{19}}(x)\Big]
%\N\\&&
-\frac{\zeta_2}{2} \int_0^1 dx \frac{(-x)^N-1}{x+1}
\N\\*&&
\times \sqrt{\frac{x}{8-x}}\Big[\HA_{\sf w_{25},w_{19}}(x)
+\HA_{\sf w_{12},2}(x)+\HA_{\sf w_{12},0}(x)\Big]
%\\&&
+\frac{13}{24}\zeta_3 \int_0^1 dx \frac{(-2x)^N-1}{x+\frac{1}{2}}\sqrt{\frac{x}{4-x}}\HA_{\sf w_{19}}(x)
\N\\*&&
+\frac{\zeta_2}{3}\sum_{i=1}^N \binom{2i}{i} (-2)^i S_{1,2}\left(\frac{1}{2},-1\right)(i)~.
\end{eqnarray}
%--------------------------------------------------------------------------------------------------------------
%%%%%%%%%%%%%%%%%%%%%%%%%%%%%%%%%%%%%%%%%%%%%%%%%%%%%%%%%%%%%%%%%%%%%%%
\section{The Mellin Transform of \boldmath $D$-finite Functions}
\label{sec:6}
%%%%%%%%%%%%%%%%%%%%%%%%%%%%%%%%%%%%%%%%%%%%%%%%%%%%%%%%%%%%%%%%%%%%%%%

\vspace*{1mm}
\noindent
In the preceding Sections we described the formalism relating the nested
finite (inverse) binomial sums to Mellin transforms
of iterated integrals over also root-valued letters. In physical applications also the converse way is of interest. The 
corresponding methods are the subject of the present Section. We consider the
Mellin transform of \textit{D-finite} functions. 
Therefore let $\mathbb K$ be a field of characteristic 0. A function $f=f(x)$ is called \textit{D-finite} if there exist 
polynomials $p_d(x),p_{d-1}(x),\ldots,p_0(x)\in \mathbb K[x]$  (not all $p_i$
being $0$) such that the following  $D$-finite differential equation holds:
%----------------------------------------------------------------------------
\begin{equation}
 p_d(x)f^{(d)}(x)+\cdots+p_1(x)f'(x)+p_0(x)f(x)=0.
\end{equation}
%----------------------------------------------------------------------------
We emphasize that the class of $D$-finite functions is rather large due to its
closure properties. Namely, if we are given two such differential
equations that contain $D$-finite functions $f(x)$ and $g(x)$ as solutions, one
can compute $D$-finite differential equations that contain $f(x)+g(x)$,
$f(x)g(x)$ or $\int_0^x f(y)dy$ as solutions. In other words any composition
of these operations over known $D$-finite functions $f_i(x)$ is again a
$D$-finite function $h(x)$. In particular, if for the inner building blocks
$f_i(x)$ the $D$-finite differential equations are given, also the $D$-finite
differential equation of $h(x)$ can be computed.\\
Of special importance in this Section is the connection to recurrence
relations. A sequence $(f_n)_{n\geq0}$ with $f_n\in\mathbb K$ is called
\textit{P-finite} (or \textit{P-recursive}) if there exist polynomials 
$p_d(n),p_{d-1}(n),\ldots,p_0(n)\in \mathbb K[n]$ (not all $p_i$ being $0$) such
that a $P$-finite recurrence
%----------------------------------------------------------------------------
\begin{equation}
 p_d(n)f_{n+d}+\cdots+p_1(n)f_{n+1}+p_0(n)f_n=0
\end{equation}
%----------------------------------------------------------------------------
holds for all $n\in\mathbb N$ (from a certain point on).
In the following we utilize the fact that $D$-finite functions are
precisely the generating functions of $P$-finite sequences: 
if $f(x)$ is $D$-finite, then the coefficients 
$f_n$ of the formal power series expansion
%----------------------------------------------------------------------------
\begin{equation}\label{Equ:GeneratingFu}
f(x) = \sum_{n=0}^{\infty} f_n x^n
\end{equation}
%----------------------------------------------------------------------------
form a $P$-finite sequence. Conversely, for a given $P$-finite sequence
$(f_n)_{n\geq0}$, the function defined by the above sum (\ie its 
generating function) is $D$-finite (this is true in the sense of formal power series, even if the sum has a zero radius of 
convergence). Note that given a $D$-finite differential equation for a $D$-finite function $f(x)$ it is straightforward to 
construct a 
$P$-finite recurrence for the coefficients of its power series expansion. For a
recent overview of this holonomic machinery and further literature we
refer to~\cite{KauersPaule:2011}.

\noindent
Subsequently, we deal with the following problem:\\
\textbf{Given} a $D$-finite function $f(x)$.\\
\textbf{Find} an expression $F(n)$ given as a linear combination of indefinite
nested sums such that for all $n\in\mathbb N$ (from a certain point on) we have
%----------------------------------------------------------------------------
\begin{equation}\label{Equ:MellRep}
\Mvec[f(x)](n)=F(n).
\end{equation}
We want to present three different but similar methods to solve the problem
above which are implemented in Ablinger's
Mathematica package
\texttt{HarmonicSums}~\cite{Ablinger:2010kw,Ablinger:2011te,Ablinger:2013cf,
Ablinger:2013hcp}. All of these methods rely on the $D$-finite machinery sketched
above. In addition the symbolic
summation package \texttt{Sigma}~\cite{SIG1,SIG2} is used which is based on an
algorithmic difference field
theory~\cite{Karr:81,Bron:00,Schneider:01,Schneider:05a,Schneider:10b,
Schneider:07d,Schneider:08c,Schneider:10a,Schneider:10c,Schneider:13b}. Here
one of the key ideas is to derive a recurrence relation that contains the Mellin
transform as solution and to execute \SigmaP's recurrence solver that finds all
solutions that can be expressed in
terms of indefinite nested sums and
products~\cite{Petkov:92,Abramov:94,Singer:99,Blumlein:2009tj}; similarly to the
differential case presented in Section~\ref{sec:4}, \eqref{Equ:NestedIntegrals}, these solutions
are also called d'Alembertian solutions.
Besides that we might end up at an expression of the Mellin transform that is
given in terms of definite multi-sums. In this case these sums are transformed
to expressions in terms of indefinite nested sums and products with the help
of the package {\tt
EvaluateMultiSums}~\cite{Ablinger:2010pb,Blumlein:2012hg,Schneider:2013zna},
which is based on \SigmaP. In order to deal with infinite summations the
package {\tt HarmonicSums} is used in addition that can
deal, e.g., with asymptotic expansions of the
arising special functions.

\vspace*{2mm}
\noindent
\textbf{Method 1:} Let $f(x)$ be a $D$-finite function. If we succeed in finding $(f_i)_{i\geq 0}$ such that 
$f(x)=\sum_{i=0}^\infty f_i x^i$ we have
%----------------------------------------------------------------------------
\begin{equation}\label{eq-method1}
\int_0^1 x^n \sum_{i=0}^\infty f_i x^i dx =\sum_{i=0}^\infty f_i\int_0^1 x^{n+i}dx=\sum_{i=0}^\infty f_i \frac{1}{n+i+1}.
\end{equation}
%----------------------------------------------------------------------------
Note that the equation above is true in the sense of formal power series. If we consider analytic functions, we have to make sure 
that the change of the integral and sum is valid. In order to compute the Mellin transform $\Mvec[f(x)](n)$ we can now proceed as 
follows: 
%----------------------------------------------------------------------------
\begin{enumerate}
 \item Compute a $D$-finite differential equation for $f(x).$
 \item Use the differential equation to compute a $P$-finite recurrence for
$(f_i)_{i\geq 0}$ where $f_i$ are the coefficients in~\eqref{Equ:GeneratingFu}.
 \item Compute initial values for the recurrence.
 \item Solve the recurrence (by using {\tt Sigma}) to get a closed
form representation for 
$(f_i)_{i\geq 
0}.$ \label{RecSol}
 \item $F(n)=\sum_{i=0}^\infty f_i \frac{1}{n+i+1}$ is the Mellin transform of
$f(x)$.
 \item Transform $F(n)$ in terms of indefinite nested sums and products (using
{\tt Sigma}, 
{\tt HarmonicSums} and 
{\tt EvaluateMultiSums}).
\end{enumerate}
%----------------------------------------------------------------------------

\noindent Note that $\SigmaP$ finds all solutions that can be expressed in
terms of indefinite nested sums and products. Hence as long as
such solutions suffice to solve the recurrence in item
$\ref{RecSol}$ we can proceed.
In particular, if we succeed in item $\ref{RecSol}$ to find a closed form for
$(f_i)_{i\geq 0}$, we can represent $f(x)$ by $\sum_{i=0}^\infty f_i x^i$ and
use Eq.~($\ref{eq-method1}$) to get the Mellin transform of $f(x).$

\begin{example}We want to compute the Mellin transform of
%----------------------------------------------------------------------------
 $$f(x):=\int_0^x \frac{\sqrt{1-\tau}}{1+\tau} \, d\tau.$$
 We find that
$$(-3+x) f(x)'+2 (-1+x) (1+x) f(x)''=0$$
%----------------------------------------------------------------------------
which leads to the recurrence
%----------------------------------------------------------------------------
$$(i+2 (-1+i) i) f_i-3 (1+i) f_{i+1}-2 (1+i) (2+i) f_{i+2}=0$$
%----------------------------------------------------------------------------
for the coefficients $f_i$ in the formal power series of $f(x)$. Initial values
can be computed easily and solving the recurrence leads to
%----------------------------------------------------------------------------
$$f(x)=\sum_{i=1}^{\infty } x^i \big((-1)^i \big(-\frac{1}{i}-\frac{\prod_{p=1}^i \frac{1-2 p}{2 p}}{i (-1+2 i)}\big)+\frac{(-1)^i \sum_{j=1}^i \frac{\prod_{p=1}^i \frac{1-2 p}{2 p}}{-1+2 j}}{i}\big).$$
%----------------------------------------------------------------------------
Hence
%----------------------------------------------------------------------------
$$\Mvec[f(x)](n)=\sum_{\text{j}_1=1}^{\infty}\frac{(-1)^{\text{j}_1}
\big(1-\prod_{\mathit{p}_1=1}^{\text{j}_1}\frac{1-2\mathit{p}_1}{2\mathit{p}_1}
-\sum_{\text{j}_2=1}^{\text{j}_1}\frac{\prod_{\mathit{p}_1=1}^{\text{j}_2}\frac{
1-2\mathit{p}_1}{2\mathit{p}_1}}{-1+2\text{j}_2}-2\text{j}_1+2\big(\sum_{\text{j
}_2=1}^{\text{j}_1}\frac{\prod_{\mathit{p}_1=1}^{\text{j}_2}\frac{1-2\mathit{p}
_1}{2\mathit{p}_1}}{-1+2\text{j}_2}\big)\text{j}_1\big)}{\text{j}
_1\big(1+n+\text{j}_1\big)\big(-1+2\text{j}_1\big)}.$$
%----------------------------------------------------------------------------
The package {\tt HarmonicSums}  offers the following 
command to apply this method:
\end{example}

%----------------------------------------------------------------------------
\begin{mma}
{
\In \text{\bf GeneralMellin[GL[$\left\{\frac{(1 - \textnormal{VarGL})^{\frac{1}{2}}}{1 + \textnormal{VarGL}}\right\},x$],$x$,Method$\rightarrow 1$]}\\
\Out {\sum_{\text{j}_1=1}^{\infty}\frac{(-1)^{\text{j}_1}\big(1-\prod_{\mathit{p}_1=1}^{\text{j}_1}\frac{1-2\mathit{p}_1}{2\mathit{p}_1}-\sum_{\text{j}_2=1}^{\text{j}_1}\frac{\prod_{\mathit{p}_1=1}^{\text{j}_2}\frac{1-2\mathit{p}_1}{2\mathit{p}_1}}{-1+2\text{j}_2}-2\text{j}_1+2\big(\sum_{\text{j}_2=1}^{\text{j}_1}\frac{\prod_{\mathit{p}_1=1}^{\text{j}_2}\frac{1-2\mathit{p}_1}{2\mathit{p}_1}}{-1+2\text{j}_2}\big)\text{j}_1\big)}{\text{j}_1\big(1+n+\text{j}_1\big)\big(-1+2\text{j}_1\big)}}\\
}
\end{mma}
%----------------------------------------------------------------------------
Before we present another method to compute the Mellin transform of a $D$-finite
function we state the following proposition.

%----------------------------------------------------------------------------
\begin{prop}
 If the Mellin transform of a $D$-finite function is defined \ie the integral $\int_0^1x^nf(x)dx$ exist, then it is 
$P$-finite.
\end{prop}
%----------------------------------------------------------------------------

%----------------------------------------------------------------------------
\begin{proof}
Let $f(x)$ be a $D$-finite function such that the integral $\int_0^1x^nf(x)dx$ exists. Using the properties of the Mellin 
transform we can easily check that
\begin{eqnarray*}
\Mvec[x^m f^{(p)}(x)](n)&=&\frac{(-1)^p (n+m)!}{(n+m-p)!}\Mvec[f(x)](n+m-p)\\&&+\sum_{i=0}^{p-1}\frac{(-1)^i 
(n+m)!}{(n+m-i)!}f^{(p-1-i)}(1).
\end{eqnarray*}
Finally, we apply the Mellin transform to the $D$-finite differential equation of
$f(x)$ using the relation above, and we get a 
$P$-finite recurrence for $\Mvec[f(x)](n).$
\end{proof}
%----------------------------------------------------------------------------

\noindent
Now, a second method to compute the Mellin transform is obvious:

\vspace*{2mm}
\noindent
\textbf{Method 2:} Let $f(x)$ be a $D$-finite function. In order to compute the Mellin transform 
$\Mvec[f(x)](n),$ we can proceed as follows:
%----------------------------------------------------------------------------
\begin{enumerate}
 \item Compute a $D$-finite differential equation for $f(x).$
 \item Use the proposition above to compute a $P$-finite recurrence for $\Mvec[f(x)](n).$
 \item Compute initial values for the recurrence.
 \item Solve the recurrence (by using \SigmaP) to get a closed form
representation for $\Mvec[f(x)](n).$
\end{enumerate}
%----------------------------------------------------------------------------

%----------------------------------------------------------------------------
\begin{example}We want to compute the Mellin transform of
 $$f(x):=\int_0^x \frac{\sqrt{1-\tau}}{1+\tau} \, d\tau.$$
 We find that
$$(-3+x) f(x)'+2 (-1+x) (1+x) f(x)''=0$$
which leads to the recurrence
%----------------------------------------------------------------------------
\begin{eqnarray*}
6\int_0^1 \frac{\sqrt{1-\tau}}{1+\tau} \, d\tau&=&-2 (n-1) n \Mvec[f(x)](n-2)+3 n \Mvec[f(x)](n-1)\\&&+(n+1) (2 n+3) 
\Mvec[f(x)](n).
\end{eqnarray*}
%----------------------------------------------------------------------------
Initial values can be computed easily and solving the recurrence leads to
%----------------------------------------------------------------------------
\begin{eqnarray*}
\Mvec[f(x)](n)&=&(-1)^n \bigg(4 \frac{\prod_{\text{i}_1=1}^n -\frac{2 \text{i}_1}{-1+2 \text{i}_1}}{(2 n+1) (2 
n+3)}+\frac{\int_0^1 
\frac{\sqrt{1-\tau}}{1+\tau} \, d\tau-2}{n+1}\bigg)\\
	    &&-\frac{4 (-1)^n \sum_{\text{i}_1=1}^n \frac{\prod_{\text{i}_2=1}^{\text{i}_1} -\frac{2 \text{i}_2}{-1+2 \text{i}_2}}{1+2 \text{i}_1}}{n+1}+\frac{\int_0^1 \frac{\sqrt{1-\tau}}{1+\tau} \, d\tau}{n+1}.
\end{eqnarray*}
%----------------------------------------------------------------------------
{\tt HarmonicSums} offers the following command to apply 
this method:
\end{example}
%----------------------------------------------------------------------------

%----------------------------------------------------------------------------
\begin{mma}
{
\In \text{\bf GeneralMellin[GL[$\left\{\frac{(1 - \textnormal{VarGL})^{\frac{1}{2}}}{1 + \textnormal{VarGL}}\right\},x$],$x$,Method$\rightarrow 2$]}\\
\Out
{\frac{\text{GL}\big(\big\{\frac{\sqrt{1-\text{VarGL}}}{\text{VarGL}+1}\big\},
1\big)}{n+1}+(-1)^n \bigg(4 \frac{\prod_{\text{i}_1=1}^n -\frac{2
\text{i}_1}{-1+2 \text{i}_1}}{(2 n+1) (2
n+3)}+\frac{\GL{\big\{\frac{\sqrt{1-\text{VarGL}}}{\text{VarGL}+1}\big\}}{1}-2}{
n+1}\bigg)\\{-\frac{4 (-1)^n \sum_{\text{i}_1=1}^n
\frac{\prod_{\text{i}_2=1}^{\text{i}_1} -\frac{2 \text{i}_2}{-1+2
\text{i}_2}}{1+2 \text{i}_1}}{n+1}}.}\\
}
\end{mma}
%----------------------------------------------------------------------------

\vspace*{2mm}
\noindent
\textbf{Method 3:} We consider the generating function $g(y)$ of the Mellin transform of $f(x)$ \ie 
$g(y):=\sum_{n=0}^\infty y^n 
F(n)=\sum_{n=0}^\infty y^n \int_0^1 x^n f(x)dx.$ 
Changing integral and sum leads to 
%----------------------------------------------------------------------------
$$g(y)=\int_0^1 f(x)\sum_{n=0}^\infty y^n x^ndx=\int_0^1f(x)\frac{1}{1-xy}dx.$$
%----------------------------------------------------------------------------
Again if we consider analytic functions we have to make sure that the change of the integral and sum is valid.
Now we proceed as follows:
%----------------------------------------------------------------------------
\begin{enumerate}
 \item Compute a $D$-finite differential equation for $g(y).$
 \item Use the differential equation to compute a $P$-finite recurrence for $F(n).$
 \item Compute initial values of $F(n)$.
 \item Solve the recurrence (by using \SigmaP) to get a closed form
representation for $F(n).$
\end{enumerate}
%----------------------------------------------------------------------------

%----------------------------------------------------------------------------
\begin{example}We want to compute the Mellin transform $F(n)$ of
 $$f(x):=\int_0^x \frac{\sqrt{1-\tau _1}}{1+\tau _1} \, d\tau _1.$$
%----------------------------------------------------------------------------
Hence 
%----------------------------------------------------------------------------
$$g(y):=\int_0^1f(x)\frac{1}{1-xy}dx.$$
%----------------------------------------------------------------------------
 We find that
%----------------------------------------------------------------------------
\begin{eqnarray*}
&&2 (y-1)^2 (y+1) y^2 g^{(6)}(y)+(y-1) \left(47 y^2+7 y-22\right) y g^{(5)}(y)\\
&&+\left(355 y^3-244 y^2-113 y+50\right) g^{(4)}(y) +4 (3 y-2) (85 y+14) g^{(3)}(y)\\
&&+12 (85 y-24) g''(y)+240 g'(y)=0
\end{eqnarray*}
%----------------------------------------------------------------------------
which leads to a recurrence which is satisfied by the Mellin transform $F(n)$ of
$f(x)$:
%----------------------------------------------------------------------------
\begin{eqnarray*}
&&\left(2 n^6+29 n^5+170 n^4+515 n^3+848 n^2+716 n+240\right) F({n+1})\\
&&+\left(-2 n^6-34 n^5-234 n^4-830 n^3-1588 n^2-1536 n-576\right) F({n+2})\\
&&+\left(-2 n^6-35 n^5-248 n^4-905 n^3-1778 n^2-1760 n-672\right) F({n+3})\\
&&+\left(2 n^6+40 n^5+320 n^4+1300 n^3+2798 n^2+2980 n+1200\right) F({n+4})=0.
\end{eqnarray*}
%----------------------------------------------------------------------------
Initial values of $F(n)$ can be computed easily and solving the recurrence leads
to the closed form
%----------------------------------------------------------------------------
\begin{eqnarray*}
F(n)=\Mvec[f(x)](n)&=&(-1)^n \bigg(\frac{\int_0^1 \frac{\sqrt{1-\tau}}{1+\tau}
\, d\tau+2}{n+1}-4 (n+2) \frac{\prod_{\text{i}_1=1}^n 
-\frac{2 \text{i}_1}{-1+2 \text{i}_1}}{(n+1) (2 n+1) (2 n+3)}\bigg)\\
&&+\frac{2 (-1)^n \sum_{\text{i}_1=1}^n \frac{\prod_{\text{i}_2=1}^{\text{i}_1}
-\frac{2 \text{i}_2}{-1+2 \text{i}_2}}{\text{i}_1}}{n+1}+\frac{\int_0^1
\frac{\sqrt{1-\tau}}{1+\tau} \, d\tau}{n+1}.
\end{eqnarray*}
%----------------------------------------------------------------------------
{\tt HarmonicSums}  offers the following command to apply 
this method:
\end{example}
%----------------------------------------------------------------------------
%----------------------------------------------------------------------------
\begin{mma}
{
\In \text{\bf GeneralMellin[GL[$\left\{\frac{(1 - \text{VarGL})^{\frac{1}{2}}}{1 + \text{VarGL}}\right\},x$],$x$,Method$
\rightarrow 3$]}\\
\Out {(-1)^n \bigg(\frac{\text{GL}\big[\big\{\frac{\sqrt{1-\text{VarGL}}}{\text{VarGL}+1}\big\},1\big]+2}{n+1}-4 (n+2) 
\frac{\prod_{\text{i}_1=1}^n -\frac{2 \text{i}_1}{-1+2 \text{i}_1}}{(n+1) (2 n+1) (2 n+3)}\bigg)+\frac{2 (-1)^n 
\sum_{\text{i}_1=1}^n \frac{\prod_{\text{i}_2=1}^{\text{i}_1} -\frac{2 \text{i}_2}{-1+2 \text{i}_2}}{\text{i}_1}}{n+1}
+\frac{\text{GL}\big[\big\{\frac{\sqrt{1-\text{VarGL}}}{\text{VarGL}+1}\big\},1\big]}{n+1}.}\\
}
\end{mma}
%----------------------------------------------------------------------------

So far we presented a general toolbox to express the Mellin transform 
of nested integrals which are $D$-finite to indefinite nested sums and
products whenever this is possible. 
The purpose of the following results is to provide direct rewrite rules to
compute the Mellin transform of certain generalized harmonic polylogarithms. Here we will exploit the
specific structure of our nested integrals instead of merely using their $D$-finiteness. In order to compute their Mellin
transforms we proceed recursively by applying the identities given in the
following lemmas. Moreover, the theorem below gives some sufficient conditions
on when the Mellin transform can be written in terms of nested (inverse)
binomial sums.
%----------------------------------------------------------------------------
\begin{lemma}\label{lem:Transform1}
For $c \in \mathbb{C}\setminus\{0\}$ we have that
%----------------------------------------------------------------------------
 \begin{eqnarray}
  \int_0^{1-\ep}dx x^Nf(x) &=& \frac{1}{N+1}\left((1-\ep)^{N+1}f(1-\ep)-\int_0^{1-\ep}dx x^{N+1}f^\prime(x)\right)\\*
  \int_0^1dx x^N\frac{f(x)}{x-c} &=& c^N\left(\int_0^1dx \frac{f(x)}{x-c} +
\sum_{i=1}^N\frac{1}{c^i}\int_0^1dx x^{i-1}f(x)\right).
 \end{eqnarray}
%----------------------------------------------------------------------------
\end{lemma}
%----------------------------------------------------------------------------

%----------------------------------------------------------------------------
\begin{lemma}\label{lem:Transform2}
 Let $a \in \mathbb{C}\setminus[0,\infty[$ and let $f(x)$ be differentiable on $]0,1[$. Then for all $N \in \mathbb{N}$ and small $\ep>0$ we have
 \begin{eqnarray}
  \int_0^{1-\ep}dx x^N\frac{f(x)}{\sqrt{x-a}} &=& \frac{(4a)^N}{(2N+1)\binom{2N}{N}}\Bigg(\int_0^{1-\ep}dx \frac{f(x)}{\sqrt{x-a}}\\*
  &&+2\sum_{i=1}^N\frac{\binom{2i}{i}}{(4a)^i}\bigg(\sqrt{1-a-\ep}(1-\ep)^if(1-\ep)-\int_0^{1-\ep}dx x^i\sqrt{x-a}f^\prime(x)\bigg)\Bigg).\N
 \end{eqnarray}
\end{lemma}
%----------------------------------------------------------------------------
%----------------------------------------------------------------------------
\begin{proof}
 It is easy to check that both sides of the equation satisfy the following
recurrence in $N$:
 \[
  (N+\tfrac{3}{2})y(N+1)-a(N+1)y(N) =
\sqrt{1-a-\ep}(1-\ep)^{N+1}f(1-\ep)-\int_0^{1-\ep}dx
x^{N+1}\sqrt{x-a}f^\prime(x).
 \]
 For $N=0$ both sides of the equation trivially agree. Altogether, this implies
that the equation holds for all $N \in \mathbb{N}$.
\end{proof}
%----------------------------------------------------------------------------

The following variants of the previous lemma are proven in a completely analogous way.

%----------------------------------------------------------------------------
\begin{lemma}\label{lem:Transform3}
For $a \in \mathbb{C}\setminus[0,\infty[$ we have that
 \begin{multline}
  \int_0^{1-\ep}dx x^N\frac{f(x)}{\sqrt{x(x-a)}}\ =\ \left(\frac{a}{4}\right)^N\binom{2N}{N}\Bigg(\int_0^{1-\ep}dx \frac{f(x)}{\sqrt{x(x-a)}}\\*
  +\sum_{i=1}^N\frac{(4/a)^i}{i\binom{2i}{i}}\bigg(\sqrt{1-a-\ep}(1-\ep)^{i-1/2}f(1-\ep)
  -\int_0^{1-\ep}dx x^i\sqrt{\frac{x-a}{x}}f^\prime(x)\bigg)\Bigg).
 \end{multline}
\end{lemma}
%----------------------------------------------------------------------------

%----------------------------------------------------------------------------
\begin{lemma}\label{lem:Transform3a}
For $a \in \mathbb{C}\setminus[0,\infty[$ we have that
 \begin{multline}
  \int_0^{1-\ep}dx x^N\sqrt{\frac{x}{x-a}}f(x)\ =\ \frac{1}{2}\left(\frac{a}{4}\right)^N\binom{2(N+1)}{N+1}\Bigg(\int_0^{1-\ep}dx \sqrt{\frac{x}{x-a}}f(x)\\*
  +\sum_{i=1}^N\frac{(4/a)^i}{(2i+1)\binom{2i}{i}}\bigg(\sqrt{1-a-\ep}(1-\ep)^{i+1/2}f(1-\ep)
  -\int_0^{1-\ep}dx x^i\sqrt{x(x-a)}f^\prime(x)\bigg)\Bigg).
 \end{multline}
\end{lemma}
%----------------------------------------------------------------------------
\noindent
Also the lemmata~\ref{lem:Transform2}, \ref{lem:Transform3} and \ref{lem:Transform3a} were derived using {\tt Singular} 
\cite{SINGULAR}.

Based on these formulae we now are in the position to prove the following criterion on the expressibility of the Mellin transform in terms of nested (inverse) binomial sums. Note that analytic continuation can be used to relax some of the restrictions 
on the position of the singularities.

%----------------------------------------------------------------------------
\begin{theorem}
 Let $r_0(x),\dots,r_k(x) \in \mathbb{C}(x)\setminus\{0\}$ without a pole at
$x=1$, let $p_0(x),\dots,p_k(x) \in \mathbb{C}[x]\setminus\{0\}$ with all 
their roots in $]{-}\infty,0]$, and set
$h_i(x):=\frac{r_i(x)}{\sqrt{p_i(x)}}$. Assume that each of the products 
$h_0(x)\cdots h_i(x)$ is of one of the forms $r(x)$, $\frac{r(x)}{\sqrt{x}}$, $\frac{r(x)}{\sqrt{x-a}}$,
or $\frac{r(x)}{\sqrt{x(x-a)}}$ for some $r(x)\in\mathbb C(x)\setminus\{0\}$
(not necessarily all of the same form). If the integral 
$\int_0^1dx h_0(x)\HA_{\sf h_1,\dots,h_k}(x)$ exists, then the Mellin transform 
$\Mvec[h_0(x)\HA_{\sf h_1,\dots,h_k}(x)](N)$ is expressible in terms of nested (inverse) binomial sums.
\end{theorem}
%----------------------------------------------------------------------------
%----------------------------------------------------------------------------
\begin{proof}
 Assume without loss of generality that $p_0(x)$ is monic and square-free and hence is equal to $1$, $x$, $x-a$, or 
$x(x-a)$ for some $a \in ]{-}\infty,0[$. Determine the full partial fraction decomposition of $r_0(x)$, then by linearity 
and Lemma~\ref{lem:Transform1} we can assume without loss of generality that $r_0(x)=1$ leaving the four cases $h_0(x) \in 
\{1,\frac{1}{\sqrt{x}},\frac{1}{\sqrt{x-a}},\frac{1}{\sqrt{x(x-a)}}\}$. We proceed by induction on $k$. If $k=0$, we trivially have 
$\Mvec[1](N)=\frac{1}{N+1}$ and $\Mvec\left[\frac{1}{\sqrt{x}}\right](N)=\frac{1}{N+1/2}$ as well as the following two identities obtained by Lemma~\ref{lem:Transform2} and Lemma~\ref{lem:Transform3} respectively.
%----------------------------------------------------------------------------
 \begin{eqnarray}
  \Mvec\left[\frac{1}{\sqrt{x-a}}\right](N) &=& \frac{(4a)^N}{(N+1/2)\binom{2N}{N}}\left(\sqrt{1-a}-\sqrt{-a}+\sqrt{1-a}\sum_{i=1}^N\frac{\binom{2i}{i}}{(4a)^i}\right)\\
  \Mvec\left[\frac{1}{\sqrt{x(x-a)}}\right](N) &=& \left(\frac{a}{4}\right)^N\binom{2N}{N}\left(2\arcsinh\left(\frac{1}{\sqrt{-a}}\right)+\sqrt{1-a}\sum_{i=1}^N\frac{(4/a)^i}{i\binom{2i}{i}}\right)
 \end{eqnarray}
%----------------------------------------------------------------------------
 If $k>0$, by the lemmas above we obtain the following identities for the four cases. Note that 
$\HA_{\sf h_1,\dots,h_k}(1)=0$ by the assumptions above. The following Mellin transforms have the representation
%----------------------------------------------------------------------------
 \begin{eqnarray}
  \Mvec\left[\HA_{\sf h_1,\dots,h_k}(x)\right](N) &=& \frac{1}{N+1}\Mvec\left[xh_1(x)\HA_{\sf h_2,\dots,h_k}(x)\right](N)\\
  \Mvec\left[\frac{\HA_{\sf h_1,\dots,h_k}(x)}{\sqrt{x}}\right](N) &=& \frac{1}{N+1/2}\Mvec\left[\sqrt{x}h_1(x)\HA_{\sf h_2,\dots,h_k}(x)\right](N)\\
  \Mvec\left[\frac{\HA_{\sf h_1,\dots,h_k}(x)}{\sqrt{x-a}}\right](N) &=& \frac{(4a)^N}{(2N+1)\binom{2N}{N}}\Bigg(\int_0^1dx \frac{\HA_{\sf h_1,\dots,h_k}(x)}{\sqrt{x-a}}+\\*
  &&+2\sum_{i=1}^N\frac{\binom{2i}{i}}{(4a)^i}\Mvec\left[\sqrt{x-a}h_1(x)\HA_{\sf h_2,\dots,h_k}(x)\right](i)\Bigg)\N\\
  \Mvec\left[\frac{\HA_{\sf h_1,\dots,h_k}(x)}{\sqrt{x(x-a)}}\right](N) &=& \left(\frac{a}{4}\right)^N\binom{2N}{N}\Bigg(\int_0^1dx \frac{\HA_{\sf h_1,\dots,h_k}(x)}{\sqrt{x(x-a)}}+\\*
  &&+\sum_{i=1}^N\frac{(4/a)^i}{i\binom{2i}{i}}\Mvec\left[\sqrt{\frac{x-a}{x}}h_1(x)\HA_{\sf h_2,\dots,h_k}(x)\right](i)\Bigg)\N
 \end{eqnarray}
%----------------------------------------------------------------------------
 Each pre-factor of $\HA_{\sf h_2,\dots,h_k}(x)$ in the Mellin transforms on the right hand sides is a rational multiple of $h_0(x)h_1(x)$. Hence it has one of the forms $r(x)$, $\frac{r(x)}{\sqrt{x}}$, $\frac{r(x)}{\sqrt{x-a}}$, or $\frac{r(x)}{\sqrt{x(x-a)}}$ as well. Now we apply the induction hypothesis to complete the proof.
\end{proof}
%----------------------------------------------------------------------------

The previous theorem shows that even three square root singularities in a single letter are allowed in certain cases. 
One of the simplest examples in terms of the letters
%----------------------------------------------------------------------------
\begin{equation}
h_{\sf 0}(x):=\frac{1}{\sqrt{x+1}}~~~~~\text{and}~~~~~h_{\sf 1}(x):=\frac{1}{\sqrt{x(1-x)(1+x)}}
\end{equation}
%----------------------------------------------------------------------------
is given by
%----------------------------------------------------------------------------
\begin{eqnarray}
 \Mvec\left[\frac{\HA_{\sf h_1}(x)}{\sqrt{x+1}}\right](N) &=& \frac{(-4)^N}{(2N+1)\binom{2N}{N}}\Bigg(\int_0^1dx\frac{\HA_{\sf h_1}(x)}{\sqrt{x+1}}+2\sum_{i=1}^N\frac{\binom{2i}{i}}{(-4)^i}\Mvec\left[\frac{1}{\sqrt{x(1-x)}}\right](i)\Bigg)\N\\
 &=&\frac{(-4)^N}{(2N+1)\binom{2N}{N}}\left(2\pi-
\frac{\Gamma^2\left(\tfrac{1}{4}\right)}{\sqrt{2 \pi}}
+2\pi\sum_{i=1}^N\binom{2i}{i}^2\left(-\frac{1}{16}\right)^i\right).
\label{eq:ELL1}
\end{eqnarray}
%----------------------------------------------------------------------------
Due to the emergence of the term $\binom{2i}{i}^2$ in (\ref{eq:ELL1}) the
corresponding iterated integral $\HA_{\sf h_1}(x)$ of depth  {\sf d = 1} is 
not an elementary function anymore, but it can be written in terms of an incomplete elliptic integral of the first kind,
%----------------------------------------------------------------------------
\begin{equation}
 \HA_{\sf h_1}(x) = \sqrt{2}\left((1-i)K\left(\frac{1}{\sqrt{2}}\right)+iF\left(\frac{\pi}{2}
+i\arcsinh\left(\sqrt{x}\right),\frac{1}{\sqrt{2}}\right)\right).
\end{equation}
%----------------------------------------------------------------------------
An analogous behaviour has been observed in case of inverse binomial sums in Ref.~\cite{Ablinger:2013eba}, Eq.~(27), 
in a similar context before, noting that Mellin convolutions can be converted into iterated integrals, as has been 
described in Section~\ref{sec:4}. 

The direct rewrite rules above are as well implemented in {\tt HarmonicSums}. The command \texttt{MellinG} applies them:
\begin{mma}
{
\In \text{\bf MellinG[GL[$\left\{\frac{(1 + \text{VarGL})^{\frac{1}{2}}}{\text{VarGL}}\right\},x$],$x$,$n$]}\\
\Out {-\frac{1}{n+1} \big(
        \frac{\sqrt{2}}{n+1}
        +\frac{(-1)^n 2^{2 n+2}}{(n+1) (2 n+3) \binom{2 n+2}{n+1}} \
\big(
                \sqrt{2} \big(
                        \frac{\big(
                                -\frac{1}{4}\big)^{1
                        +n
                        } (2 (1
                        +n
                        ))!}{((1
                        +n
                        )!)^2} (1)
                        +
                        \sum_{\tau_1=1}^n \frac{\big(
                                -\frac{1}{4}\big)^{\tau_1} \big(
                                2 \tau_1\big)!}{\big(
                                \tau_1!\big)^2}
                \big)
+\sqrt{2}-1\big)
\big)
+\frac{\text{GL}\big[
        \big\{\frac{\sqrt{\text{VarGL}+1}}{\text{VarGL}}\big\},1\big]}{n+1}}\\
}
\end{mma}
%%%%%%%%%%%%%%%%%%%%%%%%%%%%%%%%%%%%%%%%%%%%%%%%%%%%%%%%%%%%%%%%%%%%%%%
\section{Generating Functions and Infinite Nested Binomial Sums}
\label{sec:7}
%%%%%%%%%%%%%%%%%%%%%%%%%%%%%%%%%%%%%%%%%%%%%%%%%%%%%%%%%%%%%%%%%%%%%%%

\vspace*{1mm}
\noindent
In the previous Section we dealt with finite nested binomial sums. We would
like to turn now to the infinite versions of (inverse) 
binomial sums. Infinite sums of this kind have been considered previously in
Refs.~\cite{Ogreid:1997bx,Fleischer:1998nb,Kalmykov:2000qe,Davydychev:2000na,Borwein:2001,
Davydychev:2003mv,Weinzierl:2004bn,Kalmykov:2007dk}.\footnote{In
Ref.~\cite{Weinzierl:2004bn} also infinite sums are considered, which 
result of the $\Gamma$-function around rational numbers.} They can be seen as limiting cases of the finite nested sums 
considered so far, introducing a new parameter $x$ in addition. Starting from their Mellin representation we may send 
$N\to\infty$ to obtain an integral 
representation for the infinite sum, provided the sum converges. On the other hand infinite sums can also be 
considered as specializations of generating functions. So, if we are given an
integral representation of the generating 
function of a sequence then we can obtain an integral representation for the
infinite sum over that sequence. These two approaches to infinite sums can be summarized by the
following formula:
%----------------------------------------------------------------------------
\[
 \sum_{i=1}^\infty f(i) = \lim_{N\to\infty}\sum_{i=1}^N f(i) =
\lim_{x\to1}\sum_{i=1}^\infty x^if(i).
\]
%----------------------------------------------------------------------------
Generating functions are also important in their own right. In the following we will consider two ways to compute them.

%----------------------------------------------------------------------------
\subsection{Generating Functions from Sum Expressions}
%----------------------------------------------------------------------------
We consider generating functions of sequences that are given by nested sums. In analogy to the 
formulae we developed in the previous Sections for the computation of convolution integrals and Mellin transforms, we aim 
at an appropriate set of identities which can be applied recursively in order to express generating functions in terms of iterated integrals. First we summarize a few well known properties, which will be useful in our context.
%----------------------------------------------------------------------------
\begin{lemma}
 \begin{eqnarray}
  \sum_{n=1}^\infty\frac{x^n}{n}f(n) &=& \int_0^xdt\frac{1}{t}\sum_{n=1}^\infty t^nf(n)\label{eq:GenFunReciprocal}\\
  \sum_{n=1}^\infty x^n\sum_{i=1}^nf(i) &=& \frac{1}{1-x}\sum_{n=1}^\infty x^nf(n)\label{eq:GenFunSum}\\
  \sum_{n=1}^\infty\frac{x^n}{n+1}f(n) &=& \frac{1}{x}\int_0^xdt\sum_{n=1}^\infty t^nf(n)\\
  &=& \sum_{n=1}^\infty\frac{x^n}{n}f(n)-\frac{1}{x}\int_0^xdt\sum_{n=1}^\infty \frac{t^n}{n}f(n)
 \end{eqnarray}
\end{lemma}
%----------------------------------------------------------------------------

\noindent
Next, we give some identities specifically useful to expressions involving binomial coefficients. 
Related formulae can also be found in the Appendix of \cite{Fleischer:1998nb}, which do not 
explicitly express the results as iterated integrals, however.
%----------------------------------------------------------------------------
\begin{lemma}
 \begin{equation}\label{eq:GenFunBinomialSum}
  \sum_{n=1}^\infty x^n\binom{2n}{n}\sum_{i=1}^nf(i) = \frac{1}{4\sqrt{\frac{1}{4}-x}}\int_0^xdt\frac{1}
{t\sqrt{\frac{1}{4}-t}}\sum_{n=1}^\infty t^nn\binom{2n}{n}f(n).
 \end{equation}
\end{lemma}
\begin{proof}
 Both sides of the equation satisfy the following initial value problem for
$y(x)$, which has a unique solution near $x=0$:
 \begin{eqnarray*}
  y^\prime(x)-\frac{2}{1-4x}y(x) &=& \frac{1}{x(1-4x)}\sum_{n=1}^\infty x^nn\binom{2n}{n}f(n)\\
  y(0) &=& 0.
 \end{eqnarray*}
%pos
\end{proof}
\noindent Similarly, the following lemmas are obtained.
%----------------------------------------------------------------------------
\begin{lemma}
 \begin{eqnarray}\label{eq:GenFunInverseBinomialSum1a}
  \sum_{n=1}^\infty \frac{x^n}{n\binom{2n}{n}}\sum_{i=1}^nf(i) &=& \sqrt{\frac{x}{4-x}}\int_0^xdt\frac{1}{\sqrt{t(4-t)}}\sum_{n=0}^\infty \frac{t^n}{\binom{2n}{n}}f(n+1)\\
  &=& \sum_{n=1}^\infty \frac{x^n}{n\binom{2n}{n}}f(n)+\sqrt{\frac{x}{4-x}}\int_0^xdt\frac{1}{\sqrt{t(4-t)}}
\sum_{n=1}^\infty \frac{t^n}{\binom{2n}{n}}f(n). \label{eq:GenFunInverseBinomialSum1b}
 \end{eqnarray}
\end{lemma}
%----------------------------------------------------------------------------
\begin{lemma}
 \begin{equation}\label{eq:GenFunInverseBinomialSum2}
  \sum_{n=1}^\infty\frac{x^n}{(2n+1)\binom{2n}{n}}\sum_{i=1}^nf(i) = \frac{2}{\sqrt{x(4-x)}}\int_0^xdt\frac{1}
{\sqrt{t(4-t)}}\sum_{n=1}^\infty\frac{t^n}{\binom{2n}{n}}f(n).
 \end{equation}
\end{lemma}
%----------------------------------------------------------------------------
Note that for $f(n):=\delta_{n,1}$ we obtain the following generating functions as special cases from the formulae above.
%----------------------------------------------------------------------------
\begin{eqnarray}
 \sum_{n=1}^\infty x^n\binom{2n}{n} &=& \frac{1}{2\sqrt{\frac{1}{4}-x}}-1\label{eq:GenFunBinomial}\\
 \sum_{n=1}^\infty \frac{x^n}{n\binom{2n}{n}} &=& \sqrt{\frac{x}{4-x}}\int_0^xdt\frac{1}{\sqrt{t(4-t)}}
= 2 \sqrt{\frac{x}{4-x}} \arcsin\left(\frac{\sqrt{x}}{2}\right)
\label{eq:GenFunInverseBinomial1}\\
 \sum_{n=1}^\infty\frac{x^n}{(2n+1)\binom{2n}{n}} &=& \frac{2}{\sqrt{x(4-x)}}\int_0^xdt\frac{1}
{\sqrt{t(4-t)}}-1 = \frac{4}{\sqrt{x(4-x)}} \arcsin\left(\frac{\sqrt{x}}{2}\right) -1~.
\label{eq:GenFunInverseBinomial2}
\nonumber\\
\end{eqnarray}
%----------------------------------------------------------------------------
Let us illustrate the use of the formulae above by a simple example.
%----------------------------------------------------------------------------
\begin{example}
 Consider the generating function
%----------------------------------------------------------------------------
 \begin{equation}\label{eq:GenFunExample1}
  \sum_{n=1}^\infty\frac{x^n}{n\binom{2n}{n}}S_2(n).
 \end{equation}
%----------------------------------------------------------------------------
 Applying \eqref{eq:GenFunInverseBinomialSum1b} and then \eqref{eq:GenFunReciprocal} we obtain
%----------------------------------------------------------------------------
 \begin{eqnarray*}
  \sum_{n=1}^\infty\frac{x^n}{n\binom{2n}{n}}S_2(n) &=& \sum_{n=1}^\infty \frac{x^n}{n^3\binom{2n}{n}}+\sqrt{\frac{x}{4-x}}\int_0^xdt\frac{1}{\sqrt{t(4-t)}}\sum_{n=1}^\infty \frac{t^n}{n^2\binom{2n}{n}}\\
  &=&\int_0^xdt\frac{1}{t}\int_0^tdu\frac{1}{u}\sum_{n=1}^\infty\frac{u^n}{n\binom{2n}{n}}+\sqrt{\frac{x}{4-x}}\int_0^xdt\frac{1}{\sqrt{t(4-t)}}\int_0^tdu\frac{1}{u}\sum_{n=1}^\infty\frac{u^n}{n\binom{2n}{n}}.
 \end{eqnarray*}
%----------------------------------------------------------------------------
 Now, by virtue of \eqref{eq:GenFunInverseBinomial1} we obtain the result
%----------------------------------------------------------------------------
 \begin{eqnarray}
  \sum_{n=1}^\infty\frac{x^n}{n\binom{2n}{n}}S_2(n) &=& \int_0^xdt\frac{1}{t}\int_0^tdu\frac{1}{\sqrt{u(4-u)}}\int_0^udv\frac{1}{\sqrt{v(4-v)}}
\N\\&&
  +\sqrt{\frac{x}{4-x}}\int_0^xdt\frac{1}{\sqrt{t(4-t)}}\int_0^tdu\frac{1}{\sqrt{u(4-u)}}\int_0^udv\frac{1}{\sqrt{v(4-v)}}\label{eq:GenFunExample1Result}\\
  &=&\mathrm{H}_{\sf 0,w_{19},w_{19}}(x)+\sqrt{\frac{x}{4-x}}\mathrm{H}_{\sf w_{19},w_{19},w_{19}}(x).
 \end{eqnarray}
%----------------------------------------------------------------------------
\end{example}

\noindent
The sum (\ref{eq:GenFunExample1}) has been calculated e.g. in Ref.~\cite{Davydychev:2003mv} 
in terms of LogSine-functions \cite{LEWIN:81}.

These direct rewrite rules are as well implemented {\tt HarmonicSums}. The command \texttt{BSToGL} applies them:
\begin{mma}
{
\In \text{\bf BSToGL[}\sum_{\text{o}_1=1}^{\infty } \frac{x^{{o}_1} \binom{2 {o}_1}{{o}_1}}{\big(-1+2 {o}_1\big)^2}\text{\bf ,x]}\\
\Out {\sqrt{x} 64 \text{GL}\big[
        \big\{\sqrt{1-4 \text{VarGL}} \sqrt{\text{VarGL}}\big\},x\big]
+\sqrt{1-4 x}
+4 (1-8 x) \sqrt{1-4 x} x
-1}\\
}
\end{mma}
%----------------------------------------------------------------------------
%----------------------------------------------------------------------------
\subsection{Generating Functions from Integral Representations}
%----------------------------------------------------------------------------
Now we consider generating functions of sequences that are given by Mellin
transforms of nested integrals:
%----------------------------------------------------------------------------
\begin{equation}\label{eq:GenFunMellin}
 \sum_{n=k}^\infty x^n\Mvec[f(t)](n) = \int_0^1dt\frac{(tx)^k}{1-tx}f(t). 
\end{equation}
%----------------------------------------------------------------------------
We again aim at a recursive procedure for rewriting expressions of the above
form into iterated integrals. For most of the situations encountered in the
context of the integral representations considered in Section~\ref{sec:2}, the
following theorems provide the necessary formulae to do so.
%----------------------------------------------------------------------------
\begin{theorem}\label{thm:GenFunMellin}
 Let $a<0$, $c<0$, $0<x<1$, and let $k,e_0,e_1$ with $\frac{e_0}{2}<k+1$ and $\frac{e_0+e_1}{2}<k+1$. Furthermore, let $f(x)$ be bounded on $(0,1)$. Then
%  \begin{multline}
%   \int_0^1dt\frac{(tx)^k}{1-tx}\frac{tf(t)}{(t-c)t^{e_0/2}(t-a)^{e_1/2}}\ =\\* c_1\frac{x^k}{1-cx}+\frac{x^{\frac{e_0+e_1}{2}-1}}{(1-cx)(1-ax)^{e_1/2}}\int_0^xdt\frac{t^{1-\frac{e_0+e_1}{2}}}{(1-at)^{1-e_1/2}}\Bigg(c_2t^k+(1-a)^{1-e_1/2}f(1)\frac{t^k}{1-t}\\*
%   -\int_0^1du\frac{(tu)^k}{1-tu}u^{1-e_0/2}(u-a)^{1-e_1/2}f^\prime(u)\Bigg)
%  \end{multline}
%  \begin{eqnarray}
%   c_1 &=& \int_0^1dt\,t^k\frac{tf(t)}{(t-c)t^{e_0/2}(t-a)^{e_1/2}}\\
%   c_2 &=& a\left(k+1-\frac{e_0}{2}\right)\int_0^1dt\,t^k\frac{f(t)}{t^{e_0/2}(t-a)^{e_1/2}}
%  \end{eqnarray}
 \begin{multline}
  \int_0^1dt\frac{(tx)^k}{1-tx}\frac{tf(t)}{(t-c)t^{e_0/2}(t-a)^{e_1/2}}\ =\\*
  \frac{x^{\frac{e_0+e_1}{2}-1}}{(1-cx)(1-ax)^{e_1/2}}\int_0^xdt\frac{t^{1-\frac{e_0+e_1}{2}}}{(1-at)^{1-e_1/2}}\Bigg(\bar{c}_1t^{k-1}+\bar{c}_2t^k+(1-a)^{1-e_1/2}f(1)\frac{t^k}{1-t}\\*
  -\int_0^1du\frac{(tu)^k}{1-tu}u^{1-e_0/2}(u-a)^{1-e_1/2}f^\prime(u)\Bigg), \label{eq:GenFunMellinMain}
 \end{multline}
%----------------------------------------------------------------------------
 where
%----------------------------------------------------------------------------
 \begin{eqnarray}
  \bar{c}_1 &=& \left(k+1-\frac{e_0+e_1}{2}\right)\int_0^1dt\,t^k\frac{tf(t)}{(t-c)t^{e_0/2}(t-a)^{e_1/2}}\label{eq:GenFunMellinMain1}\\
  \bar{c}_2 &=& -ac\left(k+1-\frac{e_0}{2}\right)\int_0^1dt\,t^k\frac{f(t)}{(t-c)t^{e_0/2}(t-a)^{e_1/2}}.\label{eq:GenFunMellinMain2}
 \end{eqnarray}
%----------------------------------------------------------------------------
\end{theorem}
%----------------------------------------------------------------------------
%----------------------------------------------------------------------------
\begin{theorem}
 Let $a<0$, $c<0$, $0<x<1$, and let $k,e_0,e_1$ with $\frac{e_0}{2}<k+1$ and $\frac{e_0+e_1}{2}=k+1$. Furthermore, let $f(x)$ be bounded on $(0,1)$. Then we have
 \begin{multline}
  \int_0^1dt\frac{(tx)^k}{1-tx}\frac{tf(t)}{(t-c)t^{e_0/2}(t-a)^{e_1/2}}\ =\\*
  \frac{x^k}{(1-cx)(1-ax)^{e_1/2}}\Bigg(\bar{c}_1+\int_0^xdt\frac{1}{(1-at)^{1-e_1/2}}\Bigg(\bar{c}_2+(1-a)^{1-e_1/2}f(1)\frac{1}{1-t}\\*
  -\frac{1}{t^k}\int_0^1du\frac{(tu)^k}{1-tu}u^{1-e_0/2}(u-a)^{1-e_1/2}f^\prime(u)\Bigg)\Bigg),
 \end{multline}
%----------------------------------------------------------------------------
 where
%----------------------------------------------------------------------------
 \begin{eqnarray}
  \bar{c}_1 &=& \int_0^1dt\,t^k\frac{tf(t)}{(t-c)t^{e_0/2}(t-a)^{e_1/2}}\\
  \bar{c}_2 &=& -ac\frac{e_1}{2}\int_0^1dt\,t^k\frac{f(t)}{(t-c)t^{e_0/2}(t-a)^{e_1/2}}.
 \end{eqnarray}
\end{theorem}
%----------------------------------------------------------------------------
In fact, the last theorem can also be derived as a limiting case of the first. The scope of the theorems above may be further extended by considering limits or by analytic continuation in the parameters. As an example we give two formulae which can be derived as limiting cases of Theorem~\ref{thm:GenFunMellin}. For the first we set $e_0=e_1=1$ and let $c\to0$, for the second we set $e_0=e_1=0$ and let $a\to0$.
%----------------------------------------------------------------------------
\begin{cor}
For $a<0$, $k>0$, $0<x<1$ we have that
 \begin{eqnarray}
  \int_0^1dt\frac{(tx)^k}{1-tx}\frac{f(t)}{\sqrt{t(t-a)}} &=& \frac{1}{\sqrt{1-ax}}\Bigg(c_1\int_0^xdt\frac{t^{k-1}}{\sqrt{1-at}}+\sqrt{1-a}f(1)\int_0^xdt\frac{t^k}{(1-t)\sqrt{1-at}}\N\\
  &&-\int_0^xdt\frac{1}{\sqrt{1-at}}\int_0^1du\frac{(tu)^k}{1-tu}\sqrt{u(u-a)}f^\prime(u)\Bigg)\\
  c_1 &=& k\int_0^1dt\,t^k\frac{f(t)}{\sqrt{t(t-a)}}.
 \end{eqnarray}
\end{cor}
\begin{cor}\label{cor:GenFunMellinC}
 For $c<0$, $k>-1$, $0<x<1$ we have that
 \begin{eqnarray}
  \int_0^1dt\frac{(tx)^k}{1-tx}\frac{tf(t)}{t-c} &=& \frac{1}{x(1-cx)}\int_0^xdt\Bigg(c_1t^k+f(1)\frac{t^{k+1}}{1-t}-\int_0^1du\frac{(tu)^{k+1}}{1-tu}uf^\prime(u)\Bigg)\\
  c_1 &=& (k+1)\int_0^1dt\,t^k\frac{tf(t)}{t-c}.
 \end{eqnarray}
\end{cor}

\begin{example}
 Consider again the generating function \eqref{eq:GenFunExample1}, this time the
coefficients are given in terms of a Mellin transform:
 \begin{equation}
  \sum_{k=1}^\infty\frac{x^n}{4^n}\Mvec\left[\frac{\HA_{\sf
0,w_3}(t)}{t}+\frac{\zeta_2-\HA_{\sf w_3,w_3}(t)}{t\sqrt{1-t}}\right](n).
 \end{equation}
 We treat the two terms inside the Mellin transform separately. By
Eq.~\eqref{eq:GenFunMellin} the first gives rise to the following integral, to
which we apply the main formula
\eqref{eq:GenFunMellinMain}-\eqref{eq:GenFunMellinMain2} with $k=1$, $e_0=2$,
$e_1=0$ and $a\to0$, $c\to0$. This yields
 \[
  \int_0^1dt\frac{t\frac{x}{4}}{1-t\frac{x}{4}}\frac{\HA_{\sf 0,w_3}(t)}{t} =
\int_0^{x/4}dt\left(\int_0^1du\,\HA_{\sf 0,w_3}(u)+\HA_{\sf
0,w_3}(1)\frac{t}{1-t}+\int_0^1du\frac{tu}{1-tu}\HA_{\sf w_3}(u)\right).
 \]
 The two constants appearing have the values $\int_0^1dx\,\HA_{\sf 0,w_3}(x)=2$
and $\HA_{\sf 0,w_3}(1)=0$. The remaining inner integral allows one to apply our
main formula again. More precisely, we use Corollary~\ref{cor:GenFunMellinC}
with $k=1$ and $c\to0$ to obtain
 \[
  \int_0^1dt\frac{tx}{1-tx}\HA_{\sf w_3}(t) = \frac{1}{x}\int_0^xdt\left(2t\int_0^1du\,u\HA_{\sf w_3}(u)+\HA_{\sf w_3}(1)\frac{t^2}{1-t}+\int_0^1du\frac{(tu)^2}{1-tu}\frac{1}{\sqrt{1-u}}\right).
 \]
 Again, we evaluate the constants, $\int_0^1du\,u\HA_{\sf w_3}(u)=\frac{2}{3}$ and $\HA_{\sf w_3}(1)=0$, and apply Theorem~\ref{thm:GenFunMellin} to the inner integral. With $k=2$, $e_0=0$, $e_1=1$, and $c\to0$ analytic continuation to $a=1$, avoiding $a\ge0$, yields
 \[
  \int_0^1dt\frac{(tx)^2}{1-tx}\frac{1}{\sqrt{1-t}} = \frac{1}{\sqrt{x}\sqrt{1-x}}\int_0^xdt\frac{\sqrt{t}}{\sqrt{1-t}}\frac{5}{2}t\int_0^1du\frac{u^2}{\sqrt{1-u}}.
 \]
 Finally, we just need to compute $\int_0^1dx\frac{x^2}{\sqrt{1-x}}=\frac{16}{15}$. Altogether, we obtain the following 
 representation of the integral
 \begin{eqnarray*}
  \int_0^1dt\frac{t\frac{x}{4}}{1-t\frac{x}{4}}\frac{\HA_{\sf 0,w_3}(t)}{t} &=& \int_0^{x/4}dt\left(2+\frac{1}{t}\int_0^tdu\left(\frac{4}{3}u+\frac{1}{\sqrt{u}\sqrt{1-u}}\int_0^udv\frac{\frac{8}{3}v^2}{\sqrt{v}\sqrt{1-v}}\right)\right)\\
  &=& \int_0^{x/4}dt\frac{1}{t}\int_0^tdu\frac{1}{\sqrt{u}\sqrt{1-u}}\int_0^udv\frac{1}{\sqrt{v}\sqrt{1-v}}.
 \end{eqnarray*}
 Dealing with the second term of the Mellin transform we similarly apply
Theorem~\ref{thm:GenFunMellin} three times in order to arrive at the following
identities:
 \begin{eqnarray*}
  \int_0^1dt\frac{t\frac{x}{4}}{1-t\frac{x}{4}}\frac{\zeta_2-\HA_{\sf w_3,w_3}(t)}{t\sqrt{1-t}} &=& \frac{\sqrt{\frac{x}{4}}}{\sqrt{1-\frac{x}{4}}}\int_0^{x/4}dt\frac{1}{\sqrt{t}\sqrt{1-t}}\int_0^1du\frac{tu}{1-tu}\frac{\HA_{\sf w_3}(u)}{u}\\
  \int_0^1dt\frac{tx}{1-tx}\frac{\HA_{\sf w_3}(t)}{t} &=& \int_0^xdt\left(2+\int_0^1du\frac{tu}{1-tu}\frac{1}{\sqrt{1-u}}\right)\\
  \int_0^1dt\frac{tx}{1-tx}\frac{1}{\sqrt{1-t}} &=&
\frac{1}{\sqrt{x}\sqrt{1-x}}\int_0^xdt\frac{2t}{\sqrt{t}\sqrt{1-t}}.
 \end{eqnarray*}
 Combining them gives
 \[
  \int_0^1dt\frac{t\frac{x}{4}}{1-t\frac{x}{4}}\frac{\zeta_2-\HA_{\sf w_3,w_3}(t)}{t\sqrt{1-t}} = \frac{\sqrt{\frac{x}{4}}}{\sqrt{1-\frac{x}{4}}}\int_0^{x/4}dt\frac{1}{\sqrt{t}\sqrt{1-t}}\int_0^tdu\frac{1}{\sqrt{u}\sqrt{1-u}}\int_0^udv\frac{1}{\sqrt{v}\sqrt{1-v}}.
 \]
 Altogether, we obtain the following representation in terms of iterated
integrals in agreement with Eq.~\eqref{eq:GenFunExample1Result}:
 \begin{equation}
  \sum_{k=1}^\infty\frac{x^n}{4^n}\Mvec\left[\frac{\HA_{\sf
0,w_3}(t)}{t}+\frac{\zeta_2-\HA_{\sf w_3,w_3}(t)}{t\sqrt{1-t}}\right](n) =
\mathrm{H}_{\sf
0,w_3,w_3}(\tfrac{x}{4})+\frac{\sqrt{x}}{\sqrt{4-x}}\mathrm{H}_{\sf
w_3,w_3,w_3}(\tfrac{x}{4}).
 \end{equation}
\end{example}
%----------------------------------------------------------------------------

A change of indices on the expense of a different argument is possible in case of the iterated 
integrals, cf. \cite{Ablinger:2014yaa,Ablinger14}. This may sometimes be of advantage.  Let us consider the 
simple case of Eq.~(\ref{eq:GenFunInverseBinomial1}). 
The corresponding integral is
%----------------------------------------------------------------------------------------------------------------------------------
\begin{eqnarray}
\label{eq:I1}
I_1 = \int_0^x dt \frac{1}{\sqrt{t(4 - t)}}~.
\end{eqnarray}
%----------------------------------------------------------------------------------------------------------------------------------
One changes variables by setting
%----------------------------------------------------------------------------------------------------------------------------------
\begin{eqnarray}
y := \sqrt{t(4 - t)}
\label{eq:VARTR1}
\end{eqnarray}
%----------------------------------------------------------------------------------------------------------------------------------
and obtains
%----------------------------------------------------------------------------------------------------------------------------------
\begin{eqnarray}
I_1 = \int_0^{\sqrt{x(4-x)}} dy \frac{1}{\sqrt{4-y^2}} = 
\arcsin\left(\frac{\sqrt{x(4-x)}}{2}\right)~.
\end{eqnarray}
%----------------------------------------------------------------------------------------------------------------------------------
The letter $f_{\sf 19}$ occurring in (\ref{eq:I1}) now re-appears in the argument, i.e. moves to the argument of the 
iterated integral.
In case of the single infinite binomial and inverse binomial sums variable transforms of the kind
(\ref{eq:VARTR1}) lead to polylogarithmic expressions in a new variable. The corresponding variable transforms
depend on the set of iterated letters. To trade letters of the alphabet in iterated integrals against the 
argument of the integral has been of use in several physics examples, see e.g. 
Refs.~\cite{Davydychev:2003mv,Weinzierl:2004bn,Ablinger:2014yaa,Ablinger14}.  
%----------------------------------------------------------------------------------------------------------------------------
\subsection{Generating Functions for Iterated Integrals}
\label{sec:7.3}
%----------------------------------------------------------------------------------------------------------------------------

\vspace*{1mm}
\noindent
In various applications generating function representations of iterated integrals need to be calculated. They can 
be determined using the package {\tt HarmonicSums} by the following command~:
\begin{mma}
{
\In \text{\bf GLToS[GL[$\left\{(1 - \text{VarGL})^{\frac{1}{2}}\right\},x$]]}\\
\Out {\sum_{\text{o}_1=2}^{\infty } -\frac{2 x^{\text{o}_1} \prod_{\text{i}_1=1}^{\text{o}_1} \frac{-1+2 \text{i}_1}{2 \text{i}_1}}{\big(
        -3+2 \text{o}_1
\big)
\big(-1+2 \text{o}_1\big)}
+x}\\

\In \text{\bf GLToS[GL[$\left\{\text{VarGL}^{\frac{1}{2}},(4 - \text{VarGL})^{\frac{1}{2}}\right\},x$]]}\\
\Out {-\frac{1}{\sqrt{x}} 2048 
\sum_{\text{o}_1=6}^{\infty } \frac{2^{-4 \text{o}_1} x^{\text{o}_1} \
\binom{2 \text{o}_1}{\text{o}_1} \big(
        -1+\text{o}_1\big) \text{o}_1}{\big(
        1-2 \text{o}_1\big)^2 \big(
        -7+2 \text{o}_1
\big)
\big(-5+2 \text{o}_1
\big)
\big(-3+2 \text{o}_1\big)}
+\frac{4 x^{5/2}}{5}
-\frac{x^{7/2}}{28}
-\frac{x^{9/2}}{864}}\\
}
\end{mma}
%----------------------------------------------------------------------------------------------------------------------
\subsection{Asymptotic Expansions of Mellin Transforms}
\label{sec:7.4}
%----------------------------------------------------------------------------------------------------------------------

\vspace*{1mm}
\noindent
The asymptotic expansion of the Mellin transforms of iterated integrals, dealt with in the present paper, in the variable 
$n$
is of importance to perform the analytic continuation of these quantities, and therefore of the nested (inverse)  binomial 
sums, to the 
complex plane, see also Ref.~\cite{Ablinger:2014yaa,Blumlein:2009ta}. We note in particular that different of the 
individual sums occurring in the respective Feynman integrals do diverge exponentially as $n \rightarrow \infty$, but 
in their combination for a single or several Feynman diagrams regular results are obtained, with an asymptotic  growth of 
at most like a power of $\ln(n)$.
This property has to be checked in all physical applications requesting to perform the asymptotic expansion of the
respective nested binomial sums or the associated Mellin transforms of the iterated integrals . 
The package {\tt HarmonicSums} allows to compute these asymptotic expansions using the command~:
$\texttt{GLExpansion}[\text{GL}[a,x],x,n,ord]$. It computes the asymptotic expansion of 
$\Mvec\left[{\text{GL}[a,x]}\right](n)$ 
in $n$ up to order $ord.$
\begin{mma}
{
\In \text{\bf GLExpansion[GL[$\left\{(1 + \text{VarGL})^{\frac{1}{2}}\right\},x$],$x$,$n$,5]}\\
\Out {\frac{4225}{96 \sqrt{2} n^5}-\frac{2}{3 n^5}-\frac{469}{24 \sqrt{2} \
n^4}+\frac{2}{3 n^4}+\frac{55}{6 \sqrt{2} n^3}-\frac{2}{3 \
n^3}-\frac{7 \sqrt{2}}{3 n^2}+\frac{2}{3 n^2}+\frac{4 \sqrt{2}}{3 n}-\
\frac{2}{3 n}}\\
\In \text{\bf GLExpansion[GL[$\left\{\frac{(1 - \text{VarGL})^{\frac{1}{2}}}{1+\text{VarGL}}\right\},x$],$x$,$n$,5]}\\
\Out {\big(
        \frac{1}{n}-\frac{1}{n^2}+\frac{1}{n^3}-\frac{1}{n^4}+\frac{1}\
{n^5}\big) \text{GL}\big[
        \big\{\frac{\sqrt{1-\text{VarGL}}}{\text{VarGL}+1}\big\},1\big]
+\big(
        -\frac{1149}{512} \big(
                \frac{1}{n}\big)^{9/2}+\frac{29}{32} \big(
                \frac{1}{n}\big)^{7/2}-\frac{1}{4} \big(
                \frac{1}{n}\big)^{5/2}\big) \sqrt{\pi }}\\
}
\end{mma}
%--------------------------------------------------------------------------------------------------------------------------------
\section{Conclusions}
\label{sec:8}
%--------------------------------------------------------------------------------------------------------------------------------

\vspace*{1mm}
\noindent
Nested finite (inverse) binomial sums, weighted by (generalized, cyclotomic)
harmonic sums emerge in
Feynman diagram calculations containing massive lines starting with 2-loop
order. These sums are
related by the Mellin transform over the interval $[0,1]$ to iterated integrals
containing a
larger alphabet of root-valued letters. We have worked out algorithms to
transform given sums into
Mellin transforms of the associated integrals and vice versa. Both the
iterated integrals $\HA_{\sf \vec{a}}$
at argument $x=0$ and the nested sums in the limit $N \rightarrow \infty$ define
special constants
which enlarge the sets of constants associated with the harmonic, cyclotomic,
and generalized
harmonic sums and polylogarithms. At the side of the iterated integrals the
shuffle relations
hold.\footnote{Due to the emerging hypergeometric weights at the side of the
sums, corresponding
quasi-shuffle relations are more involved and will be dealt with elsewhere.}

As has been shown, the Mellin representation of the (inverse) binomial sums can
be found based on
the consecutive convolution of a few building blocks only. Moreover we proved a
series of theorems for
classes of functions with general parameters. We gave explicit representations
for the nested
binomial sums emerging in classes of massive 3-loop integrals containing local
operator
insertions, cf. \cite{Ablinger:2014yaa}. Using Mellin transforms of $D$-finite
functions we also provided
means in transforming iterated integrals of root-valued letters into their
associated nested sums.
With the help of the given theorems wider classes of sums and iterated integrals
emerging in loop calculations than
the presented examples can be dealt with. The knowledge of the Mellin transforms
is instrumental for
the asymptotic expansion of the nested sums, and therefore for their analytic
continuation.

Using generating functions, we also addressed the case of infinite
(inverse) binomial sums. The calculation of Feynman diagrams of growing
complexity implied by higher loop orders,
more scales and an increasing number of legs leads to a growing variety of
special mathematical functions and
special numbers forming the basis of these physical quantities. They can be
found and classified in a constructive way
in computing the diagrams with systematic algorithms in difference and
differential fields by intense and
large scale applications of computer algebra.

\newpage
\appendix
%---------------------------------------------------------------------------------------------------------------------------
\section{Constants and their Integral Representations}
\label{sec:A}
%---------------------------------------------------------------------------------------------------------------------------
In this Appendix we list a series of constants evaluating the iterated integrals $\HA_{\sf \vec{a}}$  at the argument $x = 0$.
We set $\HA_{\sf \vec{a}}(0):=\lim\limits_{\ep\to0}\tfrac{1}{2}\left(\HA_{\sf \vec{a}}(i\ep)+\HA_{\sf \vec{a}}(-i\ep)\right)$ to 
avoid poles and branch cuts on the path of integration. Previously known special numbers as the multiple zeta values, cyclotomic
harmonic sums and generalized harmonic sums for $N \rightarrow \infty$ \cite{Blumlein:2009cf,Ablinger:2011te,Ablinger:2013cf}
express part of  these constants. We denote the Catalan constant by $C$ and
zeta-values by $\zeta_k:=\zeta(k)$. 
%----------------------------------------------------------------------------
\begin{eqnarray}
\HA_{\sf w_2,w_1}(0) &=& 4C
\\
%-------
\HA_{\sf w_2,w_2,w_1}(0) &=& \frac{\pi}{2}\zeta_2
\\
%-------
\HA_{\sf w_1,1}(0) &=& -2\pi\ln(2)
\\
%-------
\HA_{\sf w_1,1,1}(0) &=& \pi\left(2\ln(2)^2+\zeta_2\right)
\\
%-------
\HA_{\sf 1,w_1}(0) &=& 2\pi\ln(2)
\\
%-------
\HA_{\sf 1,w_1,0}(0) &=& \pi\left(4\ln(2)^2-\zeta_2\right)
\\
%-------
\HA_{\sf \frac{1}{4},w_1}(0) &=& 0
\\
%-------
\HA_{\sf \frac{1}{4},w_1,1}(0) &=& -\frac{2\pi}{3}\zeta_2
\\
%-------
\HA_{\sf \frac{1}{4},w_1,1,1}(0) &=& -\pi\left(\frac{2}{9}\zeta_3-\frac{4}{3}\ln(2)\zeta_2+\frac{2\pi}{9\sqrt{3}}\left(4\zeta_2-\psi^\prime(\tfrac{1}{3})\right)\right)
\\
%-------
\HA_{\sf \frac{1}{4},w_1,w_1,w_1}(0) &=& -\frac{2\pi}{3}\zeta_3
\\
%-------
\HA_{\sf -\frac{1}{4},w_1}(0) &=& -2\pi\ln\left(\tfrac{\sqrt{5}-1}{2}\right)
\\
%-------
\HA_{\sf -\frac{1}{4},0,0,w_1}(0) &=& 2\pi\left(2\Li_3\left(\tfrac{\sqrt{5}-1}{2}\right)-\frac{6}{5}\zeta_3-\frac{6}{5}\ln\left(\tfrac{\sqrt{5}-1}{2}\right)\zeta_2+\frac{2}{3}\ln\left(\tfrac{\sqrt{5}-1}{2}\right)^3\right)
\\
%-------
\HA_{\sf \frac{1}{4},0,w_1}(0) &=& -\pi\zeta_2
\\
%-------
\HA_{\sf -\frac{1}{4},w_1,w_1,w_1}(0) &=& -2\pi\left(\frac{4}{5}\zeta_3+\frac{9}{5}\ln\left(\tfrac{\sqrt{5}-1}{2}\right)\zeta_2-\frac{2}{3}\ln\left(\tfrac{\sqrt{5}-1}{2}\right)^3\right)
\\
%-------
\HA_{\sf 1,w_3}(0) &=& 3\zeta_2
\\
%-------
\HA_{\sf 4,w_3}(0) &=& \frac{\zeta_2}{3}
\\
%-------
\HA_{\sf 8,w_3}(0) &=& 2\arccot(\sqrt{7})^2
\\
%-------
\HA_{\sf 8,0,0,w_3}(0) &=& \frac{1}{4}\mbox{}_5F_4\left({1,1,1,1,1 \atop \frac{3}{2},2,2,2} \middle| \frac{1}{8} \right)
\\
%-------
\HA_{\sf -2,1,0}(0) &=& \frac{1}{3}\ln(\tfrac{3}{2})^3+\ln(\tfrac{3}{2})\zeta_2+\ln(\tfrac{3}{2})\Li_2(-\tfrac{1}{2})-\Li_3(-\tfrac{1}{2})-2\Li_3(\tfrac{1}{3})
\\
%-------
\HA_{\sf w_{27},w_{19}}(0) &=& -\frac{1}{3}\ln(2)^2+\frac{4}{9}\zeta_2-\frac{2}{3}\Li_2(-\tfrac{1}{2})
\\
%-------
\HA_{\sf -\frac{1}{2},0}(0) &=& \frac{1}{2}\ln(2)^2+\zeta_2+\Li_2(-\tfrac{1}{2})
\\
%-------
\HA_{\sf \frac{1}{2},0,0,0}(0) &=& \frac{1}{24}\ln(2)^4-\ln(2)^2\zeta_2-\frac{4}{5}\zeta_2^2+\Li_4(\tfrac{1}{2})
\\
%-------
\HA_{\sf w_{29},w_8}(0) &=& \frac{4}{9}\left(\psi^\prime(\tfrac{1}{3})-4\zeta_2\right)
\\
%-------
\HA_{\sf -\frac{1}{2},w_{28}}(0) &=& 
\sqrt{2}\left(\frac{2}{3}\zeta_2-2\Li_2(-\tfrac{1}{2})-\ln(2)^2\right)~.
\end{eqnarray}
%----------------------------------------------------------------------------------------------------------------------------
Here $\psi(z)$ denotes the digamma function. We note that one may express the  value of the generalized 
hypergeometric function, occurring above and in later cases, by
%----------------------------------------------------------------------------------------------------------------------------
\begin{eqnarray}
_5F_4\left({1,1,1,1,\frac{3}{2} \atop 2,2,2,2} \middle|-4\right) &=& 2 
\cdot {_4F_3}\left({\frac{1}{2},\frac{1}{2},\frac{1}{2},\frac{1}{2} \atop 
\frac{3}{2},\frac{3}{2},\frac{3}{2}} \middle| -\frac{1}{4}\right) - \zeta_3
\\
_4F_3\left({\frac{1}{2},\frac{1}{2},\frac{1}{2},\frac{1}{2} 
\atop \frac{3}{2},\frac{3}{2},\frac{3}{2}} \middle| -\frac{1}{4}\right) 
&=&
\frac{1}{8} \int_0^1~dt~\left[1-\frac{1}{\sqrt{1+4t}}\right]
\frac{\ln^2(t)}{t} + 
\frac{1}{2} \zeta_3~.
\end{eqnarray}
%----------------------------------------------------------------------------------------------------------------------------
For an evaluation of the latter integral in terms of polylogarithms see Ref.~\cite{Ablinger:2014yaa}, Eq.~(5.35).
In the following we give integral representations of many constants suitable for numerical evaluation to high precision.
First we list those constants for which we have representations by integrals of depth {\sf d = 1}.
%----------------------------------------------------------------------------------------------------------------------------
\begin{eqnarray}
 \lefteqn{\Mvec\left[\frac{x\HA_{\sf w_8,w_8}(x)}{x-1}\right](0)\ =\ -\HA_{\sf 1,w_8,w_8}(0)\ =\ }\nonumber\\
 &&8\int_0^{1/4}dx\frac{\arctanh(\sqrt{1-4x})^2-\frac{2}{3}\zeta_2}{1-x}-8\int_{1/4}^1dx\frac{\left(\frac{\pi}{3}-\arctan(\sqrt{4x-1})\right)^2}{1-x}\\
 \lefteqn{\Mvec\left[\frac{x\HA_{\sf w_{14},w_{14}}(x)}{x+1}\right](0)\ =\ -\HA_{\sf -1,w_{14},w_{14}}(0)-4\ln\left(\tfrac{\sqrt{5}-1}{2}\right)\ =\ }\nonumber\\
 &&-8\int_0^1dx\frac{\left(\arccoth(\sqrt{4x+1})+\ln\left(\tfrac{\sqrt{5}-1}{2}\right)\right)^2}{x+1}-4\ln\left(\tfrac{\sqrt{5}-1}{2}\right)\\
 \lefteqn{\Mvec\left[\frac{\HA_{\sf -\frac{1}{2},0,0,0}(x)}{x+1}\right](0)\ =\ \HA_{\sf -1,-\frac{1}{2},0,0,0}(0)\ =\ 
 -\frac{1}{6}\int_0^1dx\frac{\ln(x+1)\ln(x)^3}{x+\frac{1}{2}}}\\
 \lefteqn{\HA_{\sf \frac{1}{4},0,w_1,1}(0)\ =\ }\nonumber\\
 &&-\int_0^{1/4}dx\frac{\text{Li}_2(4x)\ln(1-x)}{\sqrt{x(1-x)}}+\int_0^{3/4}dx\frac{\left(\ln(4(1-x))\ln(3-4x)+\text{Li}_2(4x-3)
-\zeta_2\right)\ln(x)}{\sqrt{x(1-x)}}\nonumber\\
\\
 \lefteqn{\HA_{\sf \frac{1}{4},w_2,w_2,w_1}(0)\ =\ }\nonumber\\
 &&4\int_0^{1/4}\frac{\ln(1-4x)}{\sqrt{x(x+1)}}\int_{\sqrt{x}}^1dy\frac{\arccos(y)}{\sqrt{1+y^2}}+\int_{1/4}^1dx\frac{\ln(4x-1)}{\sqrt{x(x+1)}}\int_x^1du\frac{\arccos(2u-1)}{\sqrt{u(1+u)}}\\
 \lefteqn{\HA_{\sf -\frac{1}{4},w_2,w_2,w_1}(0)\ =\ 
4\int_0^1\frac{\ln(4x+1)}{\sqrt{x(x+1)}}\int_{\sqrt{x}}^1dy\frac{\arccos(y)}{\sqrt{1+y^2}}}
\end{eqnarray}
\begin{eqnarray}
 S_{1,3}\left(-\frac{1}{2},-\frac{1}{2};\infty\right) &=& 8\text{Li}_4(\tfrac{1}{2})
+8\text{Li}_4(-\tfrac{1}{2})+\int_0^1dt\frac{\text{Li}_3(\tfrac{t}{4})}{t+2}\nonumber\\
 &=& 8\text{Li}_4(-\tfrac{1}{2})+\frac{19}{3}\text{Li}_4(\tfrac{1}{2})+3\text{Li}_4(\tfrac{2}{3})+\frac{3}{2}
\text{Li}_4(\tfrac{3}{4})+\ln(2)\text{Li}_3(-\tfrac{1}{2})\nonumber\\
 &&-\ln(3)\text{Li}_3(-\tfrac{1}{2})-\text{Li}_2(-\tfrac{1}{2})^2-\zeta_2\text{Li}_2(-\tfrac{1}{2})+\ln(2)^2
\text{Li}_2(-\tfrac{1}{2})\nonumber\\
 &&-\frac{79}{60}\zeta_2^2-\frac{10}{3}\ln(2)^2\zeta_2+5\ln(2)\ln(3)\zeta_2-\frac{3}{2}\ln(3)^2\zeta_2
+\frac{77}{36}\ln(2)^4\nonumber\\
 &&-\frac{19}{6}\ln(2)^3\ln(3)+\frac{3}{2}\ln(2)^2\ln(3)^2-\frac{\ln(2)\ln(3)^3}{2}+\frac{\ln(3)^4}{8}
\\
%\end{eqnarray}
%\begin{eqnarray}
 \Mvec\left[\frac{\sqrt{x}\HA_{\sf w_{13}}(x)}{(x+1)\sqrt{8-x}}\right](0) &=& -\frac{1}{2\sqrt{3}}\int_0^1dx\frac{\sqrt{x}\left(\arccoth\left(\frac{x+4}{\sqrt{3}\sqrt{x(8-x)}}\right)+\ln\left(\frac{5-\sqrt{21}}{2}\right)\right)}{(x+1)\sqrt{8-x}}\\
 \Mvec\left[\frac{x\HA_{\sf w_{14}}(x)}{(x+1)\sqrt{x+\frac{1}{4}}}\right](0) &=& 4\int_0^1dx\frac{x\left(\arccoth(\sqrt{4x+1})+\ln\left(\frac{\sqrt{5}-1}{2}\right)\right)}{(x+1)\sqrt{x+\frac{1}{4}}}\\
 \Mvec\left[\frac{\sqrt{x}\HA_{\sf w_{25}}(x)}{(x+1)\sqrt{8-x}}\right](0) &=& \int_0^1dx\frac{\sqrt{x}\left(\arccosh\left(3-\frac{x}{2}\right)+\ln\left(\frac{5-\sqrt{21}}{2}\right)\right)}{(x+1)\sqrt{8-x}}\\
 \Mvec\left[\frac{\sqrt{x}\HA_{\sf w_{12},0}(x)}{(x+1)\sqrt{8-x}}\right](0) &=&  -\int_0^1dx\frac{\ln(x)\left(\arccos\left(1-\frac{x}{4}\right)-\frac{1}{3}\arccos\left(\frac{4-5x}{4(x+1)}\right)\right)}{\sqrt{x(8-x)}}\\
 \Mvec\left[\frac{\sqrt{x}\HA_{\sf w_{12},2}(x)}{(x+1)\sqrt{8-x}}\right](0) &=&  \int_0^1dx\frac{\ln(2-x)\left(\arccos\left(1-\frac{x}{4}\right)-\frac{1}{3}\arccos\left(\frac{4-5x}{4(x+1)}\right)\right)}{\sqrt{x(8-x)}}\\
 \Mvec\left[\frac{\sqrt{x}\HA_{\sf w_{25},w_{19}}(x)}{(x+1)\sqrt{8-x}}\right](0) &=&  
\int_0^1dx\frac{\left(\arccos\left(1-\frac{x}{4}\right)
-\frac{1}{3}\arccos\left(\frac{4-5x}{4(x+1)}\right)\right)\left(\arccos\left(\frac{x}{2}-1\right)
-\frac{2\pi}{3}\right)}{\sqrt{(4-x)(8-x)}}
\nonumber\\
\\
 \Mvec\left[\frac{\HA_{\sf w_1,1,1}(x)}{x-\frac{1}{4}}\right](0) &=& \int_0^1dx\frac{\ln\left((3-2x)|2x-1|\right)\ln(x(2-x))^2}{\sqrt{x(2-x)}}\\
 \Mvec\left[\frac{\sqrt{x}\HA_{\sf w_{13},1,0}(x)}{(x+1)\sqrt{8-x}}\right](0) &=&  \int_0^1dx\frac{\left(\arccos\left(1-\frac{x}{4}\right)-\frac{1}{3}\arccos\left(\frac{4-5x}{4(x+1)}\right)\right)\text{Li}_2(1-x)}{(2-x)\sqrt{x(8-x)}}\\
 \Mvec\left[\frac{x\HA_{\sf w_{14},0,0}(x)}{(x+1)\sqrt{x+\frac{1}{4}}}\right](0) &=& 8\ _4F_3\left(\frac{1}{2},\frac{1}{2},\frac{1}{2},\frac{1}{2} \atop \frac{3}{2},\frac{3}{2},\frac{3}{2} \middle| -\frac{1}{4}\right)-4\zeta_3\nonumber\\
 &&-\frac{1}{\sqrt{3}}\int_0^1dx\frac{\ln(x)^2\left(\arccos\left(\frac{1-2x}{2(x+1)}\right)
-\frac{\pi}{3}\right)}{x\sqrt{x+\frac{1}{4}}}\\
 \Mvec\left[\frac{x\HA_{\sf w_{14},1,0}(x)}{(x+1)\sqrt{x+\frac{1}{4}}}\right](0) &=& \int_0^1dx\frac{\left(\sqrt{4x+1}
-1-\frac{2}{\sqrt{3}}\left(\arccos\left(\frac{1-2x}{2(x+1)}\right)
-\frac{\pi}{3}\right)\right)\text{Li}_2(1-x)}{x\sqrt{x+\frac{1}{4}}}
\nonumber\\
\end{eqnarray}
\begin{eqnarray}
 \lefteqn{\Mvec\left[\frac{x\HA_{\sf w_{14},-1,0}(x)}{(x+1)\sqrt{x+\frac{1}{4}}}\right](0)\ =}\nonumber\\
 &&\int_0^1dx\frac{\left(\sqrt{4x+1}-1-\frac{2}{\sqrt{3}}\left(\arccos\left(\frac{1-2x}{2(x+1)}\right)-\frac{\pi}{3}\right)\right)\left(\text{Li}_2(-x)+\ln(x)\ln(1+x)+\frac{\zeta_2}{2}\right)}{x\sqrt{x+\frac{1}{4}}}\\
 \lefteqn{\Mvec\left[\frac{\sqrt{x}\HA_{\sf w_{18},-1,0}(x)}{(x-1)\sqrt{8+x}}\right](0)\ =}\nonumber\\
 &&\int_0^1dx\frac{\left(\arccosh\left(1+\frac{x}{4}\right)-\frac{1}{3}\arccosh\left(\frac{4+5x}{4(1-x)}\right)\right)\left(\text{Li}_2(-x)+\ln(x)\ln(1+x)+\frac{\zeta_2}{2}\right)}{(2+x)\sqrt{x(8+x)}}\\
 \lefteqn{\Mvec\left[\frac{x\HA_{\sf w_8,0,1}(x)}{(x-1)\sqrt{x-\frac{1}{4}}}\right](0)\ =\ \int_{\frac{1}{4}}^1dx\frac{\left(\sqrt{4x-1}-\frac{2}{\sqrt{3}}\arccosh\left(\frac{2x+1}{2(1-x)}\right)\right)\left(\text{Li}_2(x)-\zeta_2\right)}{x\sqrt{x-\frac{1}{4}}}}\nonumber\\
 &&+4\int_0^1dx\frac{\left(\sqrt{x(2-x)}-1-\frac{2}{\sqrt{3}}\left(\arccos\left(\frac{x^2-2x+3}{(3-x)(x+1)}\right)-\frac{\pi}{3}\right)\right)\left(\text{Li}_2\left(\frac{(1-x)^2}{4}\right)-\zeta_2\right)}{(1-x)\sqrt{x(2-x)}}\\
 \lefteqn{\Mvec\left[\frac{x\HA_{\sf w_8,1,0}(x)}{(x-1)\sqrt{x-\frac{1}{4}}}\right](0)\ =\ \int_{\frac{1}{4}}^1dx\frac{\left(\sqrt{4x-1}-\frac{2}{\sqrt{3}}\arccosh\left(\frac{2x+1}{2(1-x)}\right)\right)\text{Li}_2(1-x)}{x\sqrt{x-\frac{1}{4}}}}\nonumber\\
 &&+4\int_0^1dx\frac{\left(\sqrt{x(2-x)}-1-\frac{2}{\sqrt{3}}\left(\arccos\left(\frac{x^2-2x+3}{(3-x)(x+1)}\right)-\frac{\pi}{3}\right)\right)\text{Li}_2\left(1-\frac{(1-x)^2}{4}\right)}{(1-x)\sqrt{x(2-x)}}\\
 \lefteqn{\Mvec\left[\frac{x\HA_{\sf w_8,1,1}(x)}{(x-1)\sqrt{x-\frac{1}{4}}}\right](0)\ =\ \frac{1}{2}\int_{\frac{1}{4}}^1dx\frac{\left(\sqrt{4x-1}-\frac{2}{\sqrt{3}}\arccosh\left(\frac{2x+1}{2(1-x)}\right)\right)\ln(1-x)^2}{x\sqrt{x-\frac{1}{4}}}}\nonumber\\
 &&+2\int_0^1dx\frac{\left(\sqrt{x(2-x)}-1-\frac{2}{\sqrt{3}}\left(\arccos\left(\frac{x^2-2x+3}{(3-x)(x+1)}\right)-\frac{\pi}{3}\right)\right)\ln\left(1-\frac{(1-x)^2}{4}\right)^2}{(1-x)\sqrt{x(2-x)}}\\
 \lefteqn{\Mvec\left[\frac{x\HA_{\sf w_8,-1,0}(x)}{(x-1)\sqrt{x-\frac{1}{4}}}\right](0)\ =\ 2\ln(2)\zeta_2-3\zeta_3-4\ln(2)\text{Li}_2(-\tfrac{1}{4})-4\text{Li}_3(-\tfrac{1}{4})}\nonumber\\
 &&-\frac{2}{\sqrt{3}}\int_{\frac{1}{4}}^1dx\frac{\arccosh\left(\frac{2x+1}{2(1-x)}\right)\left(\text{Li}_2(-x)+\ln(x)\ln(1+x)+\frac{\zeta_2}{2}\right)}{x\sqrt{x-\frac{1}{4}}}\nonumber\\
 &&+4\int_0^1dx\frac{\left(\sqrt{x(2-x)}-1\right)\left(\text{Li}_2\left(-\frac{(1-x)^2}{4}\right)+\ln\left(\frac{(1-x)^2}{4}\right)\ln\left(1+\frac{(1-x)^2}{4}\right)+\frac{\zeta_2}{2}\right)}{(1-x)\sqrt{x(2-x)}}\\
 &&-\frac{8}{\sqrt{3}}\int_0^1dx\frac{\left(\arccos\left(\frac{x^2-2x+3}{(3-x)(x+1)}\right)-\frac{\pi}{3}\right)\left(\text{Li}_2\left(-\frac{(1-x)^2}{4}\right)+\ln\left(\frac{(1-x)^2}{4}\right)\ln\left(1+\frac{(1-x)^2}{4}\right)+\frac{\zeta_2}{2}\right)}{(1-x)\sqrt{x(2-x)}}\nonumber\\
 \lefteqn{\Mvec\left[\frac{\sqrt{x}\HA_{\sf w_{25},w_{19},w_{19}}(x)}{(x+1)\sqrt{8-x}}\right](0)\ =}\nonumber\\
 &&\frac{1}{2}\int_0^1dx\frac{\left(\arccos\left(1-\frac{x}{4}\right)-\frac{1}{3}\arccos\left(\frac{4-5x}{4(x+1)}\right)\right)\left(\arccos\left(\frac{x}{2}-1\right)-\frac{2\pi}{3}\right)^2}{\sqrt{(4-x)(8-x)}}\\
 \lefteqn{\Mvec\left[\frac{\sqrt{x}\HA_{\sf w_{12},0,1,0}(x)}{(x+1)\sqrt{8-x}}\right](0)\ =}\nonumber\\
 &&\int_0^1dx\frac{\left(\arccos\left(1-\frac{x}{4}\right)-\frac{1}{3}\arccos\left(\frac{4-5x}{4(x+1)}\right)\right)\left(2(\text{Li}_3(x)-\zeta_3)-\ln(x)(\text{Li}_2(x)+\zeta_2)\right)}{\sqrt{x(8-x)}}\\
 \lefteqn{\Mvec\left[\frac{\sqrt{x}\HA_{\sf w_{17},0,-1,0}(x)}{(x-1)\sqrt{8+x}}\right](0)\ =}\nonumber\\
 &&\int_0^1dx\frac{\left(\arccosh\left(1+\frac{x}{4}\right)-\frac{1}{3}\arccosh\left(\frac{4+5x}{4(1-x)}\right)\right)\left(\ln(x)\left(\text{Li}_2(-x)-\frac{\zeta_2}{2}\right)-2\text{Li}_3(-x)-\frac{3}{2}\zeta_3\right)}{\sqrt{x(8+x)}}
\nonumber\\ &&
\\
 \lefteqn{\Mvec\left[\frac{\sqrt{x}\HA_{\sf w_{25},w_{19},w_{19},w_{19}}(x)}{(x+1)\sqrt{8-x}}\right](0)\ =}\nonumber\\
 &&\frac{1}{6}\int_0^1dx\frac{\left(\arccos\left(1-\frac{x}{4}\right)-\frac{1}{3}\arccos\left(\frac{4-5x}{4(x+1)}\right)\right)\left(\arccos\left(\frac{x}{2}-1\right)-\frac{2\pi}{3}\right)^3}{\sqrt{(4-x)(8-x)}}
\end{eqnarray}
%---------------------------------------------------------------------------------------------------------------------------------
Further constants are given by integrals of depth {\sf d = 2}.
%---------------------------------------------------------------------------------------------------------------------------------
\begin{eqnarray}
 \lefteqn{\Mvec\left[\frac{\sqrt{x}\HA_{\sf w_{21},w_{20},w_{19}}(x)}{(x-1)\sqrt{8+x}}\right](0)\ =}\nonumber\\
 &&2\int_0^1dx\frac{\arccosh\left(1+\frac{x}{4}\right)-\frac{1}{3}\arccosh\left(\frac{4+5x}{4(1-x)}\right)}{\sqrt{(4+x)(8+x)}}\int_{\sqrt{x}}^1du\frac{\arccos\left(\frac{u^2}{2}-1\right)-\frac{2\pi}{3}}{\sqrt{4+u^2}}\\
 \lefteqn{\Mvec\left[\frac{\sqrt{x}\HA_{\sf w_{21},w_{23},0}(x)}{(x-1)\sqrt{8+x}}\right](0)\ =}\nonumber\\
 &&-4\int_0^1dx\frac{\arccosh\left(1+\frac{x}{4}\right)-\frac{1}{3}\arccosh\left(\frac{4+5x}{4(1-x)}\right)}{\sqrt{(4+x)(8+x)}}\int_{\sqrt{x}}^1du\frac{\ln(u)}{(1+u^2)\sqrt{4+u^2}}\\
 \lefteqn{\Mvec\left[\frac{\sqrt{x}\HA_{\sf w_{25},w_{26},0}(x)}{(x+1)\sqrt{8-x}}\right](0)\ =}\nonumber\\
 &&-4\int_0^1dx\frac{\arccos\left(1-\frac{x}{4}\right)-\frac{1}{3}\arccos\left(\frac{4-5x}{4(x+1)}\right)}{\sqrt{(4-x)(8-x)}}\int_{\sqrt{x}}^1du\frac{\ln(u)}{(1-u^2)\sqrt{4-u^2}}\\
 \lefteqn{\Mvec\left[\frac{\sqrt{x}\HA_{\sf w_{12},2,1,0}(x)}{(x+1)\sqrt{8-x}}\right](0)\ =}\nonumber\\
 &&\int_0^1dx\frac{\arccos\left(1-\frac{x}{4}\right)-\frac{1}{3}\arccos\left(\frac{4-5x}{4(x+1)}\right)}{\sqrt{x(8-x)}}\int_x^1dt\frac{\text{Li}_2(1-t)}{2-t}\\
 \lefteqn{\Mvec\left[\frac{x\HA_{\sf w_{14},w_{14},0,0}(x)}{x+1}\right](0)\ =}\nonumber\\
 &&\frac{1}{6}\int_0^1dx\frac{x-\ln(x+1)}{x\sqrt{x+\frac{1}{4}}}\int_{\ln(x)^3}^0du\frac{1}{\sqrt{\exp(-(-u)^{1/3})+\frac{1}{4}}}\\
 \lefteqn{\Mvec\left[\frac{x\HA_{\sf w_{14},w_{14},1,0}(x)}{x+1}\right](0)\ =}\nonumber\\
 &&\int_0^1dx\frac{x-\ln(x+1)}{x\sqrt{x+\frac{1}{4}}}\int_{\ln(x)}^0du\frac{\text{Li}_2(1-e^u)}{\sqrt{e^u+\frac{1}{4}}}\\
 \lefteqn{\Mvec\left[\frac{x\HA_{\sf w_{14},w_{14},-1,0}(x)}{x+1}\right](0)\ =}\nonumber\\
 &&\int_0^1dx\frac{x-\ln(x+1)}{x\sqrt{x+\frac{1}{4}}}\int_{\ln(x)}^0du\frac{\text{Li}_2(-e^u)+u\ln(1+e^u)+\frac{\zeta_2}{2}}{\sqrt{e^u+\frac{1}{4}}}\\
 \lefteqn{\Mvec\left[\frac{\sqrt{x}\HA_{\sf w_{17},-2,-1,0}(x)}{(x-1)\sqrt{8+x}}\right](0)\ =}\nonumber\\
 &&\int_0^1dx\frac{\arccosh\left(1+\frac{x}{4}\right)-\frac{1}{3}\arccosh\left(\frac{4+5x}{4(1-x)}\right)}{\sqrt{x(8+x)}}\int_x^1dt\frac{\text{Li}_2(-t)+\ln(t)\ln(1+t)+\frac{\zeta_2}{2}}{2+t}\\
 \lefteqn{\Mvec\left[\frac{\sqrt{x}\HA_{\sf w_{21},w_{20},0,0}(x)}{(x-1)\sqrt{8+x}}\right](0)\ =}\nonumber\\
 &&4\int_0^1dx\frac{\arccosh\left(1+\frac{x}{4}\right)-\frac{1}{3}\arccosh\left(\frac{4+5x}{4(1-x)}\right)}{\sqrt{(4+x)(8+x)}}\int_{\sqrt{x}}^1du\frac{\ln(u)^2}{\sqrt{4+u^2}}\\
 \lefteqn{\Mvec\left[\frac{\sqrt{x}\HA_{\sf w_{21},w_{20},-1,0}(x)}{(x-1)\sqrt{8+x}}\right](0)\ =}\nonumber\\
&&\int_0^1dx\frac{\arccosh\left(1+\frac{x}{4}\right)-\frac{1}{3}\arccosh\left(\frac{4+5x}{4(1-x)}\right)}{\sqrt{(4+x)(8+x)}}\int_{\sqrt{x}}^1du\frac{2\text{Li}_2(-u^2)+4\ln(u)\ln(1+u^2)+\zeta_2}{\sqrt{4+u^2}}\nonumber\\
\\ 
 \lefteqn{\Mvec\left[\frac{\sqrt{x}\HA_{\sf w_{21},w_{20},w_{19},w_{19}}(x)}{(x-1)\sqrt{8+x}}\right](0)\ =}\nonumber\\
 &&\int_0^1dx\frac{\arccosh\left(1+\frac{x}{4}\right)-\frac{1}{3}\arccosh\left(\frac{4+5x}{4(1-x)}\right)}{\sqrt{(4+x)(8+x)}}\int_{\sqrt{x}}^1du\frac{\left(\arccos\left(\frac{u^2}{2}-1\right)-\frac{2\pi}{3}\right)^2}{\sqrt{4+u^2}}\\
 \lefteqn{\Mvec\left[\frac{\sqrt{x}\HA_{\sf w_{25},w_{19},0,0}(x)}{(x+1)\sqrt{8-x}}\right](0)\ =}\nonumber\\
 &&4\int_0^1dx\frac{\arccos\left(1-\frac{x}{4}\right)-\frac{1}{3}\arccos\left(\frac{4-5x}{4(x+1)}\right)}{\sqrt{(4-x)(8-x)}}\int_{\sqrt{x}}^1du\frac{\ln(u)^2}{\sqrt{4-u^2}}\\
 \lefteqn{\Mvec\left[\frac{\sqrt{x}\HA_{\sf w_{25},w_{19},1,0}(x)}{(x+1)\sqrt{8-x}}\right](0)\ =}\nonumber\\
 &&2\int_0^1dx\frac{\arccos\left(1-\frac{x}{4}\right)-\frac{1}{3}\arccos\left(\frac{4-5x}{4(x+1)}\right)}{\sqrt{(4-x)(8-x)}}\int_{\sqrt{x}}^1du\frac{\text{Li}_2(1-u^2)}{\sqrt{4-u^2}}\\
 \lefteqn{\Mvec\left[\frac{x\HA_{\sf w_8,w_8,0,1}(x)}{x-1}\right](0)\ =}\nonumber\\
 &&16\int_0^{\sqrt{3}/2}dy\frac{1+4y^2+4\ln\left(\frac{3}{4}-y^2\right)}{1+4y^2}\int_y^{\sqrt{3}/2}du\frac{\text{Li}_2\left(\frac{1}{4}+u^2\right)-\zeta_2}{1+4u^2}\\
 &&-2\int_0^{1/2}dy\frac{1-4y^2+4\ln\left(\frac{3}{4}+y^2\right)}{1-4y^2}\int_0^{2\arccosh\left(\frac{1+4y^2}{1-4y^2}\right)}du\left(\text{Li}_2\left(\frac{1}{2\left(1+\cosh(\frac{u}{2}\right)}\right)-\zeta_2\right)\nonumber\\
 \lefteqn{\Mvec\left[\frac{x\HA_{\sf w_8,w_8,1,0}(x)}{x-1}\right](0)\ =}\nonumber\\
 &&16\int_0^{\sqrt{3}/2}dy\frac{1+4y^2+4\ln\left(\frac{3}{4}-y^2\right)}{1+4y^2}\int_y^{\sqrt{3}/2}du\frac{\text{Li}_2\left(\frac{3}{4}-u^2\right)}{1+4u^2}\\
 &&-2\int_0^{1/2}dy\frac{1-4y^2+4\ln\left(\frac{3}{4}+y^2\right)}{1-4y^2}\int_0^{2\arccosh\left(\frac{1+4y^2}{1-4y^2}\right)}du\text{Li}_2\left(1-\frac{1}{2\left(1+\cosh(\frac{u}{2}\right)}\right)\nonumber\\
 \lefteqn{\Mvec\left[\frac{x\HA_{\sf w_8,w_8,1,1}(x)}{x-1}\right](0)\ =}\nonumber\\
 &&8\int_0^{\sqrt{3}/2}dy\frac{1+4y^2+4\ln\left(\frac{3}{4}-y^2\right)}{1+4y^2}\int_y^{\sqrt{3}/2}du\frac{\ln\left(\frac{3}{4}-u^2\right)^2}{1+4u^2}\\
 &&-8\int_0^{1/2}dy\frac{1-4y^2+4\ln\left(\frac{3}{4}+y^2\right)}{1-4y^2}\int_0^ydu\frac{\ln\left(\frac{3}{4}+u^2\right)^2}{1-4u^2}\nonumber\\
 \lefteqn{\Mvec\left[\frac{x\HA_{\sf w_8,w_8,-1,0}(x)}{x-1}\right](0)\ =}\nonumber\\
 &&16\int_0^{\sqrt{3}/2}dy\frac{1+4y^2+4\ln\left(\frac{3}{4}-y^2\right)}{1+4y^2}
\int_y^{\sqrt{3}/2}du\frac{\frac{\zeta_2}{2}+\ln(\frac{1}{4}+u^2)\ln(\frac{5}{4}+u^2)
+\text{Li}_2\left(-\frac{1}{4}-u^2\right)}{1+4u^2}\nonumber\\
 &&-2\int_0^{1/2}dy\frac{1-4y^2+4\ln\left(\frac{3}{4}+y^2\right)}{1-4y^2}\int_0^{2\arccosh\left(\frac{1+4y^2}{1-4y^2}\right)}du\Bigg(\frac{\zeta_2}{2}\nonumber\\
 &&+\ln\left(\frac{1}{2\left(1+\cosh(\frac{u}{2}\right)}\right)\ln\left(1+\frac{1}
{2\left(1+\cosh(\frac{u}{2}\right)}\right)+\text{Li}_2\left(-\frac{1}{2\left(
1+\cosh(\frac{u}{2}\right)}\right)\Bigg)
\end{eqnarray}
%----------------------------------------------------------------------------------------------------------------------------------
%----------------------------------------------------------------------------------------------------------------------------------
A few other constants have already been given in Ref.~\cite{Ablinger:2014yaa}.
%----------------------------------------------------------------------------------------------------------------------------------

\vspace{5mm}
\noindent
{\bf Acknowledgment.}~
We would like to thank A.~De Freitas, A. Hasselhuhn, and F.~Wi\ss{}brock for discussions. This work was supported 
in part by DFG Sonderforschungsbereich Transregio 9, Computergest\"utzte Theoretische Teilchenphysik, the Austrian 
Science Fund (FWF) grants P20347-N18 and SFB F50 (F5009-N15), the European Commission through contract PITN-GA-2010-264564 
({LHCPhenoNet}) and PITN-GA-2012-316704 ({HIGGSTOOLS}).
%----------------------------------------------------------------------------

%----------------------------------------------------------------------------------------------------------------------------------
\end{document}